\documentclass[11pt]{article}

\usepackage{amsfonts,amssymb,amsmath,amsthm,epsfig,euscript,bbm,amsthm}
\usepackage{algorithm}
\usepackage{algpseudocode}
\usepackage{amsthm}

\usepackage{tikz}
\usetikzlibrary{arrows}
\usepackage {graphicx}
\usepackage{enumerate}
\usepackage{threeparttable,subcaption}
 \usepackage{booktabs}

\usepackage{subcaption}
\usepackage{caption}
\usepackage[authoryear]{natbib}
\usepackage[titletoc,title]{appendix}
\usepackage[pdftex,colorlinks]{hyperref}
\usepackage{mathtools}
 \usepackage{xcolor}

\usepackage{graphicx}
\usepackage{subcaption}
 
 \usepackage{xr}
 \definecolor{orange}{RGB}{230,170,120}
  \definecolor{green}{RGB}{120,200,120}

 \hypersetup{
 	colorlinks=true,
 	linkcolor=blue,
 	filecolor=blue,      
 	urlcolor=cyan,
 }
 \hypersetup{colorlinks,linkcolor={blue},citecolor={blue},urlcolor={red}} 
 \usepackage{geometry}                
 \usepackage{multirow}
\usepackage[section]{placeins}
 \usepackage{bbm}
 \usepackage{graphicx}
 \usepackage{amssymb}
 \usepackage{epstopdf}
 \usepackage{tikz} 
 \usepackage{amsmath} 
\usepackage[utf8]{inputenc}
\usepackage{algorithm}

 \theoremstyle{plain}
 \newtheorem{thm}{Theorem}[section]
 \newtheorem{lem}[thm]{Lemma}
 
 \newtheorem*{cor}{Corollary}
 
 \theoremstyle{definition}
 \newtheorem{defn}{Definition}[section]
 
 \newtheorem{exmp}{Example}[section]
  \newtheorem{ass}{Assumption}
  
 \theoremstyle{definition}
 \newtheorem{rem}{Remark}
 
 \makeatletter
 \def\BState{\State\hskip-\ALG@thistlm}
 \makeatother
 
 \usepackage{authblk}

\def\spacingset#1{\renewcommand{\baselinestretch}%
{#1}\small\normalsize} \spacingset{1}
 
\definecolor{coquelicot}{rgb}{1.0, 0.22, 0.0}

 \newcommand{\blind}{0}
 
\usepackage{setspace}
\begin{document}
  
\if0\blind
{  
  \title{\bf Synthetic learner: 
  Model-Free Inference on Treatments over Time\footnote{Draft version July, 2022. First version: June, 2019.  We thank the Editor, anonymous referees, and Graham Elliott,  Kaspar W\"{u}thrich, and Yinchu Zhu for helpful comments. }}
 \author{Davide Viviano$^1$\footnote{
 Present address: Stanford Graduate School of Business, 655 Knight Way, Stanford, CA 94305. Email: dviviano@stanford.edu. This work was mostly conducted while at the Department of Economics, University of California at San Diego, La Jolla, CA, 92093. }
   $\qquad$  
    Jelena Bradic$^2$ \footnote{Department of  Mathematics and Halicio\u{g}lu Data Science Institute, 
        University of California at San Diego, La Jolla, CA, 92093. Email: jbradic@ucsd.edu.} \footnote{Jelena Bradic gratefully acknowledges the support of the grant NSF-DMS \#1712481. } \\
 \textcolor{white}{a} Stanford GSB$^1$ $\qquad$   UC San Diego$^{1,2}$
}
    \date{}
  \maketitle
} \fi

\if1\blind
{
  \bigskip
  \bigskip
  \bigskip
  \ 
 \\
 
 \ 
 \\

  \begin{center}
    {\LARGE\bf Synthetic Learner: 
  Model-free Inference on Treatments over Time}
\end{center}
  \medskip
} \fi

\bigskip  	 
 
\begin{abstract}

Understanding the effect of a particular treatment or a policy pertains to many areas of interest, ranging from political economics, marketing to healthcare.  In this paper, we develop a non-parametric algorithm for detecting the effects of treatment over time in the context of Synthetic Controls.  
The method builds on counterfactual predictions from many algorithms without necessarily assuming that the algorithms correctly capture the model.
We introduce an inferential procedure for detecting treatment effects and show that the testing procedure is asymptotically valid for stationary, beta mixing processes without imposing any restriction on the set of base algorithms under consideration. We discuss consistency guarantees for average treatment effect estimates and derive regret bounds for the proposed methodology. The class of algorithms may include Random Forest, Lasso, or any
other machine-learning estimator. Numerical studies and an application illustrate the advantages of the method.

\end{abstract}

\noindent%

\vfill

\spacingset{1.5}

  	\section{Introduction} \label{sec:intro}

This paper discusses estimation and inference on the effect of a policy intervention on a single unit observed over multiple periods and exposed to treatment from one point in time onwards. We consider an aggregate time-series set-up where researchers observe the outcome of the unit of interest $T_0$ periods before treatment and $T - T_0$ after the treatment. Researchers' main goal is to conduct inference on the effect of a trajectory of treatment effects over the post-treatment period.  
Namely, by denoting $\tau_t = Y_{0t}^1 - Y_{0t}^0$ the difference between the two potential outcomes at time $t$, researchers want to test whether $\{\tau_t\}_{t > T_0} = \tau^o$, for some null trajectory of interest $\tau^o$. Their second goal is to precisely estimate the average effect over time.

In the same spirit of synthetic controls \citep{abadie2003economic, abadie2010synthetic}, and aggregate panel data models \citep{hsiao2012panel}, we construct counterfactuals using information from $n$ (possibly finitely many) control units and covariates observed over $T$ periods. We exploit information over the time series for conducting asymptotic inference. The applications of interest are those where we observe individuals at a relatively frequent level, e.g., quarterly or monthly, over multiple years. Examples are studying the effects of taxation \citep{bai2014property}, changes in welfare programs \citep{maclean2019losing}  or (geo-localized) marketing experiments \citep{brodersen2015inferring,varian2016causal}.  We discuss additional applications at the end of Section \ref{sec:lit}. 


The first contribution of this paper is to derive an inferential procedure that is valid for general parametric and non-parametric (machine-learning) estimators. 
We propose a resampling mechanism to guarantee exact asymptotic size control without imposing restrictions on the set of estimators. The key idea for inference is to combine the sample splitting procedure with the block-bootstrap \citep{politis1992circular} and exploit the differentiability properties of the proposed procedure. Our approach does not require correct model specifications.

As a second contribution, we introduce an ensemble procedure that combines many such estimators, e.g., Synthetic controls, factor models estimators, and fully non-parametric estimators such as Random Forest or Kernel Smoothing, into a single prediction. We penalize the methods with the worst out-of-sample performance over the pre-treatment period. 
The method's goal is to increase precision. Also, it permits replacing heuristic model selection criteria whose properties are unknown with time-dependent observations with a theoretically grounded procedure. This has important implications for applied economists, who often face the difficult model selection choice from a dictionary that includes many methods  (e.g., factor models, synthetic control, difference-in-differences).\footnote{See also the discussion in \cite{athey2019ensemble}.} The ensemble method that we propose builds on the literature on exponential aggregation \citep{rigollet2012sparse, cesa2006prediction}, time-series and forecasts' combinations \citep{timmermann2006forecast, elliott2004optimal}, which we study here for causal inference. 
We show that the procedure consistently recovers the average treatment effect through a bias adjustment, and it inherits strong oracle properties for its prediction performance. 


Throughout the text, we assume exogeneity of the treatment time and stationarity for inference. 
These conditions are common when conducting inference with Synthetic Controls in the presence of a long time series\footnote{For inference via Synthetic Control \cite{chernozhukov2018exact} impose stationarity of the residuals under correct model specification and similar stationarity assumptions as to the one discussed above under misspecification to show the validity of permutation tests.  
Stationarity, strong mixing conditions on the joint distribution of the residual errors and covariates, are also imposed for valid inference in synthetic control settings in \cite{carvalho2018arco}, while covariance stationarity conditions are imposed in \cite{li2017estimation}. Further discussion is in Section \ref{sec:meth}. }, since they permit to conduct inference without relying on symmetry assumptions of placebo testing \citep{firpo2018synthetic, ben2018augmented}. They imply that a pre-post treatment comparison of the means returns a consistent (but possibly inefficient) estimate of the average effect on the treated. 
In separate sections, we relax stationarity in two directions: 
(i) allowing for time-varying fixed effects, encompassing popular two-way fixed-effect models; (ii) deriving prediction guarantees for arbitrary non-stationary settings.\footnote{These are in Section \ref{sec:time_fixed} and  Section \ref{sec:weights} respectively.}

 We conclude our discussion with a simulation study and an empirical application. We show that our procedure leads to larger power than existing methods while controlling the size of the test. In an application, we study the effect of Tennessee's health-insurance dis-enrollment program on health-related outcomes. Using survey data from Behavioral Risk Factor Surveillance System Data, we show that the program decreased health insurance coverage and the likelihood of visiting a doctor.

The paper is organized as follows. In Section \ref{sec:meth} we introduce the set-up, and identification strategy. In Section \ref{sec:inference}, we introduce the method for estimating counterfactuals and inference on sharp nulls. In Section \ref{sec:average}, we discuss estimation of the average effect on the treated. In Section \ref{sec:weights} we discuss prediction guarantees under non-stationarity. Section \ref{sec:experiments} discusses numerical experiments. Section \ref{sec:real} discusses  applications in health economics. Finally, in Section \ref{sec:carryover} we provide extensions in the presence of carry-over effects.

\subsection{Related Work} \label{sec:lit}

While non-parametric estimators have found widespread use in microeconomic applications \citep{athey2019machine}, their analysis (and consequently their applicability) in the presence of aggregate data has received much less attention. However, with aggregate data, such estimators can improve precision and better disentangle the effect of the policy from idiosyncratic shocks. 
This paper proposes a method that enables counterfactual prediction and hypothesis testing in the context of Synthetic Controls (SC) using predictions arising from many parametric and non-parametric estimators.


 
 Recent literature has proposed a wide variety of methods for predicting counterfactuals in the setting under consideration, including factor and panel data models \citep{bai2009panel, hsiao2012panel}, synthetic controls \citep{abadie2010synthetic, xu2017generalized,doudchenko2016balancing, arkhangelsky2019synthetic, ferman2016revisiting}, two-ways fixed effects models \citep{imai2021use}, ridge regression methods \citep{ben2018augmented}, kernel balancing  \citep{hazlett2018trajectory}, functional methods \citep{gunsilius2020distributional} among others.  
Additional references include 
\cite{athey2018matrix}  who proposes a matrix completion methods for SC;  \cite{athey2019ensemble} shows in a simulation exercise that ensemble methods outperform individual SC predictions on many economic data sets; 
\cite{amjad2018robust} proposes singular value thresholding,  whereas matching has been discussed by \cite{imai2018matching}.  Difference-in-difference methods were recently discussed in \cite{athey2018design} and \cite{arkhangelsky2019synthetic} in the context of staggered adoption.
 
 However, selection among these methods remains an open research question \citep{hsiao2018panel}. Combining different predictions offers a simple data-adaptive procedure to exploit information from all these models while improving the prediction performance \citep{timmermann2006forecast}. As a result, 
 our contribution can be viewed as complementary to this literature. To the best of our knowledge, we provide the first set of conditions under which predictions made by any or many machine learning methods, including Random Forests, can be used to develop valid tests for SC.

 A closely related method to our inference procedure is the permutation-based inference method, discussed in \cite{chernozhukov2017exact}. Their method accommodates a single (linear) model specification only and imposes stability conditions of the estimator. Also, while the permutation-based methods estimate counterfactuals using the entire sample (pre and post-treatment periods) as a training set, imposing the sharp null hypothesis of no treatment effects, here we estimate counterfactuals using the pre-treatment period only. This approach permits that estimation of the counterfactual does not depend on the outcomes observed over the post-treatment period. Our approach is particularly suited whenever the post-treatment period is proportional to the pre-treatment period as in the context of our empirical application.\footnote{In our empirical application, the post-treatment period is approximately twenty, and the pre-treatment period is approximately sixty. The reader may refer to Example \ref{exmp:mean} for an illustration of the benefits.}

Our paper relates more broadly to inference using penalized methods, including \cite{chernozhukov2018t} and \cite{carvalho2018arco}, which propose penalized linear regressors for asymptotic inference on treatment effects in an SC setup. However, \cite{carvalho2018arco} requires a consistent estimation of treatment effects, which is not required in our setting and focuses on a penalized model only. Differently,  \cite{chernozhukov2018t} proposes a bias-adjustment procedure for asymptotic inference, which in our framework is not required for hypothesis testing but used for average treatment effect on the treated (ATT) estimation. Our resampling mechanism allows for more generality than the tests in \cite{chernozhukov2018t} since we do not require the use of penalized linear regressor and the related conditions for estimation of counterfactuals, but we allow for general non-parametric estimators. Related literature also includes \cite{li2017estimation}, and \cite{hsiao2012panel}, which discuss properties of constrained least-squares methods under stationarity and correct model specification only. \cite{li2019statistical} discusses a sub-sampling procedure for inference with constrained least-squares estimators only.


Finally, we relate to 
\cite{kunzel2017meta} who discuss model averaging with $i.i.d.$ data and classifies method into three classes, denoted as S, T, and X learners. This paper pioneers the idea of T-learning in Synthetic control setting, offering an alternative and simple weighting scheme which is inspired by the literature on boosting \citep{schapire2012boosting} and online learning \citep{cesa1997use, cesa1999prediction, cesa2006prediction}.

To conclude, we can list several applications of interest. The first set of relevant applications includes changes in policy at a regional or state level. For instance, studying the effects of (i) exogenous variations such as environmental disasters on policy-changes \citep{potrafke2020green}, or economic outcomes \citep{cavallo2013catastrophic}; (ii) of Medicaid expenditure expansions or contractions \citep{tello2016effects}; (iii) of tax-reform on prices \citep{bai2014property}; or (iv) of policy reforms on economic growth \citep{billmeier2013assessing}. The second avenue of interesting applications includes experiments on online platforms, for which synthetic controls are becoming increasingly popular, especially in the presence of geo-localized experiments \citep{varian2016causal, li2017estimating}. The third set of applications includes studying the effect of new algorithms or financial instruments \citep{xie2014impact, bojinov2019time}.

\section{Setup and Identification} \label{sec:meth}

Throughout this article for each unit $i$ we observe outcome variable $Y_{it}$. We denote with $Y_{0t}$ the unit treated if $D_t  = 1$ and under control otherwise;    the remaining $i=1,\cdots, n$ units, $Y_{1t}, \dots, Y_{nt}$  are units always observed in the control state. Additional covariate information for each unit are denoted in compact form as $Z_{it}$.  $Z_{it}$ may also contain past covariates and past outcomes.

\subsection{Estimands and Null Hypothesis}

Following the literature on panel data and synthetic control models  \citep[e.g.,][]{hsiao2018panel, abadie2010synthetic}, 
we define the treatment assignment and the outcome of interest, respectively as 
\begin{equation} 
D_t = 1\{t > T_0\}, \quad Y_{0t} = D_t Y_{0t}^1 + (1 - D_t) Y_{0t}^0,  \quad Y_{jt} = Y_{jt}^0, \quad j > 0 
\end{equation} 
where $Y_{0t}^1, Y_{0t}^0$ denote the potential outcomes under treatment and control for the unit of interest $i = 0$, and $Y_{jt}^0$ denotes the potential outcome under control of unit $j$. Our definition of potential outcomes implicitely imposes that SUTVA holds \citep{rubin1990formal} and no carry-over effects \citep{imai2013estimating}. Extensions in the presence of carry-overs are discussed in Section \ref{sec:carryover} .

 In applications, researchers may want to test whether the difference between two potential outcomes $Y_{0t}^1 - Y_{0t}^0$ equals zero for all post-treatment periods.  Namely, they may be interested in conducting inference on the time-specific treatment effects defined below.  
 
 \begin{defn}[Time-specific treatment effect] The time-specific treatment effect is defined as follows: $\tau_t = Y_{0t}^1 - Y_{0t}^0$.  
 \end{defn} 

Here, $\tau_t$ defines the difference in potential outcomes at time $t$.  
We begin our discussion by introducing the null hypothesis of interest.  
 
 \begin{defn}[Sharp Null hypothesis] \label{defn:sharp} Define the sharp null hypothesis as 
\begin{equation} \label{eqn:null}  
 H_0:  \tau_t = \alpha_t^o, \quad  t \in \{T_0 + 1, \cdots, T\}, 
 \end{equation} 
 for a known sequence $\{\alpha_t^o\}_{t > T_0}$. 
 \end{defn} 
 
 Equation \eqref{eqn:null} imposes that the potential outcome over the post-treatment period equals the potential outcome under control plus a known (possibly time-varying) constant. 
 For example, we may consider $a_t^o = 0$ or we may consider also testing a linear trend of the form  $a_t^o = \delta (t - T_0)$  for an arbitrary $\delta \in \mathbb{R}$. 
Our results also extend when we test the average of $\tau_t$.\footnote{Namely, we may also test $\mathbb{E}[\tau_t] = \alpha_t^o$, hence allowing potential outcomes under treatment and control having different idiosyncratic shocks. This is omitted for the sake of brevity only and discussed in the Appendix (Algorithm \ref{alg:alg4}) and Remark \ref{rem:average_null}.} Finally, note that more generally, we can incorporate a more general class of hypothesis 
$H_0: Y_{0t}^1   = f(Y_{0t}^0,a_{t}^o), \  a_t^o \in \mathbb{R}, \  t > T_0,
$  for a function $f$ being invertible in its first argument, omitted for the sake of brevity.  We define 
$$
Y_{0t}^o = 
\begin{cases} 
& Y_{0t} - \alpha_t^o \quad  t > T_0, \\ 
& Y_{0t}  \quad t \le T_0. 
\end{cases} 
$$ 
the (observed) potential outcome under control under the null hypothesis $H_0$.

Testing for treatment effects may not be satisfactory to researchers, who may also be interested in reporting estimates of average treatment effects. 
This is defined below.

\begin{defn}[Average Treatment Effect on the Treated] \label{defn:estimand} The average treatment effect on the treated is defined as 
\begin{equation}
\tau =  \frac{1}{T - T_0} \sum_{t > T_0} \mathbb{E}\Big[Y_{0t}^1 - Y_{0t}^0\Big].
\end{equation}  
\end{defn} 
Here, $\tau$ denotes the average effect on the treated, averaged also over the post-treatment period. The expectation is taken over the idiosyncratic shocks.\footnote{Note that whenever $\tau_t$ is non-random $\tau = \frac{1}{T - T_0} \sum_{t > T_0} \tau_t$.} 


\subsection{Identification Conditions}

The first condition that we impose is stationarity. This is common in the literature on Synthetic Control, see  \cite{carvalho2018arco, li2017estimation, chernozhukov2018exact}.\footnote{Stationarity and beta-mixing conditions as stated above cover a large class of {\sc arma} processes \citep{pham1985some}, {\sc {\sc ar-arch}} processes \citep{lange2011estimation}, Markov Switching Processes \citep[see for example][]{lee2005probabilistic}, {\sc garch} \citep{carrasco2002mixing}, to cite some. }   We relax it in Section \ref{sec:weights}, where we allow for non-stationary observations. 

\begin{ass}[Stationarity] \label{ass:stationarity}  Suppose that $(Y_{0t}^0, Y_{1n}, \cdots, Y_{nt}, Z_{0t}, \cdots, Z_{nt}) \sim \mathcal{D}_0$ is stationary.
\end{ass}

Assumption \ref{ass:stationarity} imposes stationarity. In Section \ref{sec:time_fixed} we show how our results for inference extend under non-stationarity, in the presence of an (unknown) time-varying fixed-effects.

The second condition we impose is an identification assumption.

\begin{ass}[Identification Condition] \label{ass:ident} Suppose that $T_0 \perp (Y_{0t}^0, Y_{0t}^1, Y_{1:n,t}, Z_{0:n,t})_{t=1}^T$.
\end{ass}

Assumption \ref{ass:ident} states that the timing of the treatment is exogenous. The same conditions can be found in previous literature on Synthetic Controls. 
For instance, recent literature on Synthetic Control \citep{abadie2010synthetic, chernozhukov2018double, arkhangelsky2019synthetic, chernozhukov2018t, li2017estimation} treated $T_0$ as deterministic, in which case exogeneity of $T_0$ implicitely holds. See for example, the discussion in \cite{ferman2016revisiting} and \cite{bottmer2021design}. \cite{carvalho2018arco} also explicitely imposes exogeneity similarly to the above condition.

In the context of our empirical application, where the treatment consists of a dis-enrollment from Medicaid occurring in the early 2000s in Tennessee, and the outcomes are health-related outcomes, as argued in \cite{argys2017losing}, if the policy can be attributed to budget deficit interpretable as an exogenous variation, the assumption directly holds. We warn the reader, however, that failure of the assumption may invalidate the inferential strategy. 

 Motivated by Assumption \ref{ass:ident} we will implicitely condition on $T_0$ throughout the rest of our discussion unless otherwise specified, since, conditional on $T_0$ the distribution of observables and unobserved potential outcomes remains invariant. We return instead to non-stationary conditions in Section \ref{sec:weights} where we relax Assumption \ref{ass:stationarity} and \ref{ass:ident}.

\begin{exmp}[Stationary factor models]
Suppose that
\begin{equation} \label{eqn:factor} 
Y_{jt}^0 = \mu_j + \theta_t + \lambda_j F_t + \gamma(Z_{jt}) + u_{j,t}, \quad j \in \{0, \cdots, n\},  
\end{equation}   
with $u_{j,t}$ denoting stationary idiosyncratic errors. Let $\theta_t \sim \mathcal{N}(0,1)$ and exogenous, with, $F_t$ denoting common stationary unobserved exogenous factors and $\lambda_j$ the (exogenous) individual specific effect, and $\gamma(Z_{jt})$ a stationary components which depends on covariates. Then Assumption \ref{ass:stationarity} holds. \qed 
\end{exmp}


\subsection{Testable Implications and Identification of Average Effects}

We conclude this discussion with two lemmas. The first lemma provides a testable implication for the sharp null hypothesis in Definition \ref{defn:sharp}. This is stated below.

\begin{lem}[Sharp Null: Testable Implication] \label{lem:identi2} Under the null hypothesis in Equation \eqref{eqn:null} and Assumptions \ref{ass:stationarity}, \ref{ass:ident}, then $(Y_{0t}^o, Y_{1t}, \cdots, Y_{nt}, Z_{0t}, \cdots, Z_{nt} )$ is stationarity and independent of $T_0$ for all $t \in \{1, \cdots, T_0, \cdots, T\}$. 
\end{lem}

              The above condition implies that the distribution before and after the intervention period $T_0$ must remain invariant under the null hypothesis. This assumption is testable since we observe the empirical distribution of the above vector before and after the intervention time $T_0$. Lemma \ref{lem:identi2} is at the basis of our approach for testing the null hypothesis, which we discuss in Section \ref{sec:inference}.


 The second goal is to estimate the average effect. Identification of $\tau$ is discussed in the following lemma.  
             
\begin{lem}[Identification of $\tau$] \label{lem:1} Let Assumptions \ref{ass:stationarity}, \ref{ass:ident} hold. Then 
$$
\tau = \frac{1}{T - T_0} \sum_{t > T_0} \mathbb{E}\Big[Y_{0t}^1|T_0\Big] - \frac{1}{T_0} \sum_{s = 1}^{T_0} \mathbb{E}\Big[Y_{0s}^0|T_0\Big] . 
$$  
\end{lem} 

Lemma \ref{lem:1} is an identification result. It states that the post-treatment difference in expectation equals the target estimand $\tau$.  While Lemma \ref{lem:1} is not invoked for constructing our test, it is used for estimating the average effect discussed in Section \ref{sec:average}. 

The proofs of the lemmas are contained in Appendix \ref{app:lemmas}.

            \begin{rem}[Extensions with time fixed effects] In Section \ref{sec:time_fixed} we show how our results directly extend to the case where
  $$
  Y_{0t}^0 = \kappa_t^0 + \iota_j + \varepsilon_{j,t}^0, \quad \mathbb{E}[\varepsilon_{j,t}^0] = 0, 
  $$  
  where $\kappa_t^0$ denotes a time-fixed effect. This extension, while simple, has important implications: it permits incorporating non-stationary unobserved components. 
For this case, stationary can fail as long as the control group is ``representative" of the treated unit, i.e., time fixed effects are the same between the control and treated unit. This follows in the same spirit of two-way fixed-effect models \citep{imai2021use} that are commonly encountered in applications see, for example, \cite{garthwaite2014public}.\footnote{For example, in the context of our application for studying the effect of Tennessee dis-enrollment health-insurance program, Obamacare between 2010 and 2014 may act as a time-varying confounder. Therefore, we use as the control group for the effect of dis-enrollment in Tennessee the outcome from the other Southern States that, similarly to Tennessee, did not expand Medicaid between 2010-2014 due to Obamacare.} In the presence of time-fixed effects, identification is achieved by a difference-in-difference as opposed to a pre-post treatment comparison as discussed in Lemma \ref{lem:1b}. 
\qed 
  \end{rem}

 \section{Counterfactuals Predictions and Hypothesis Testing} \label{sec:inference}

 In this section, we discuss the problem of estimating the counterfactual prediction $\hat{Y}_{0t}^0, t > T_0$ and conducting inference on the sharp null hypothesis in Definition \ref{defn:sharp}.  
 
   Throughout our discussion, for expositional convenience, we denote in the compact form $$
   X_t =(Y_{1t}, \dots, Y_{nt}, Z_{0t}^\top, Z_{1t}^\top, \dots,Z_{nt}^\top) \in \mathcal{X}.
   $$

To test the null hypothesis of interest, we first estimate the true potential outcome $Y_{0t}^0, t > T_0$, unobserved over the post-treatment period. We define $\widehat{Y}_{0t}^0$ its estimate constructed as follows 
\begin{equation} \label{eqn:jj} 
 \widehat{Y}_{0t}^0 = w(F_0)^\top g(X_t). 
 \end{equation} 
 Here $w(F_0)$ denotes a generic functional of the empirical distribution $F_0$ of $(Y_{0t}, X_t)$ over the pre-treatment period, where $Y_{0t}$ serves as the main outcome of interest. The choice of $w(F_0)$ can be arbitrary, with the only condition required that $w(\cdot)$ is a Hadamard differentiable functional (see Assumption \ref{ass:df}).\footnote{Definition of Hadamard differentiability is provided in Appendix \ref{sec:definitions}.} In Section \ref{sec:weights} we provide explicit expressions for $w(\cdot)$.\footnote{Since the functions $g(\cdot)$ can contain an intercept, the component $w(F_0)^\top g(X_t)$ also estimates the (time-invariant) shift in mean after subctracting $\bar{Y}_t$. See, for instance, Example \ref{exmp:mean}.}

  It is important to note that the functions $g(X_t)$ can be data-dependent. Such functions are estimated as described in Algorithm \ref{alg:alg2}. Namely, first, we divide the pre-treatment period into two blocks
$t \in \{T_-, \cdots, 0\}$ and $t \in \{1, \cdots, T_0\}$. We define 
$
F_-
$ the empirical distribution for $t < 1$ of $(Y_{0t}, X_t)$.  We construct the predictors as follows 
\begin{equation} \label{eqn:pred} 
\Big\{ x \mapsto g_j(x; F_-) \Big\} 
\end{equation}  
with $F_-$ denoting the training set for such predictors, and $g_j$ denoting some pre-specified regressor. For example $g_1(X_t; F_-)$ can denote the prediction of a Random Forest, trained over the sample $t < 0$, with empirical distribution $F_-$. Whenever clear from the context, we will omit the second argument $F_-$ from the function $g_j(\cdot)$. Throughout the rest of our discussion, we will fix the size of $T_-$ and consider asymptotics as $T \rightarrow \infty, T_0 \propto T$.

We conclude our discussion with two cases of interest. 

\paragraph{Case 1: Ensemble} To gain further intuition on Equation \eqref{eqn:jj}, observe that we can interpret $w(F_0)$ as some data-dependent weights, while the functions $g(X_t)$ are interpreted as predictors or ``experts" that, based on information $X_t$, predict the counterfactual outcome $Y_{jt}$ (net of time-fixed effects) at time $t$. This follows in the same spirit of forecasts combinations \citep{timmermann2006forecast}, with $g_i$ denoting some ``experts" or learners. Section \ref{sec:weights} discusses examples and properties of weighting schemes. The construction of the weights is based on out-of-sample performance: it uses information $F_0$ which has not been used for the training of the algorithms $g(\cdot)$ (that instead use information before time $t = 0$).

 \paragraph{Case 2: Single Regressor}   Equation \eqref{eqn:jj} is, however, more general than only ensemble methods, since it also allows for estimation with a single regressor $g_1(X_t)$ (e.g., the Synthetic Control).

\begin{algorithm} [h]   \caption{Counterfactual Estimation with Sample Splitting}\label{alg:alg2}
    \begin{algorithmic}[1]
    \Require Observations $\{Y_{0t}, X_t\}_{t=T_{-}}^{T}$, time of the treatment-$T_0$,   learners  $F \mapsto g_1(\cdot,F),\dots, g_p(\cdot,F)$
    \State Split the pre-treatment period into two parts: $t \in [T_{-},0]$ and $t \in [1,T_0]$
    \State Form  predictions   $g_j(\cdot; F_-)$ with $F_-$ being the empirical distribution of $\{Y_{0t}, X_t\}_{t < 1}$ , $j \in 1,\dots,p$.      
             \State  Use the second pre-treatment period,  $\{Y_{0t}, X_t\}_{t =1}^{T_0}$, to estimate the weights of the learners,  $w(F_0)$.
 \State  Compute the predicted counterfactual 
 
\centerline{ $\hat Y_{0t} ^0 = \sum_{j=1}^p w_j(F_0) g_j(X_t)$ for  $ t > T_0$ }

        \Return the predictions $(\hat{Y}_{T_0 + 1}^0, \cdots, \hat{Y}_T^0)$.

         \end{algorithmic}
\end{algorithm}

 \subsection{Testing the Sharp Null Hypothesis}

Given $\widehat{Y}_{0t}$, we construct a test statistics and test with coverage $1 - \alpha$ having the following form: 

 \begin{equation} \label{eqn:teststat1}
         \quad \mathcal{T} = (T - T_0)^{-1/2}\sum_{t = T_0 + 1}^T \left(Y_{0t}^o - \hat Y_{0t} ^0 \right)^2, \quad \phi_\alpha(\mathcal{T}) = 1\Big\{\mathcal{T} \ge q_{1- \alpha}^* \Big\} 
        \end{equation}
        where $q_{1- \alpha}^*$ is the estimated $1 - \alpha$ quantile of $\mathcal{T}$. The test statistic depends on the prediction error between the estimated counterfactual outcome and the potential outcome \textit{under} the null hypothesis. 
        Observe that under Equation \eqref{eqn:jj}, we can represent in compact form the test statistic as a functional of the \textit{empirical} distributions 
        \begin{equation} \label{eqn:tt}  
       \mathcal{T}(F_0, F_1) =  (T - T_0)^{1/2} \int \Big(y -  w(F_0) g(x)\Big)^2 dF_1(x, y) , 
        \end{equation} 
        where $F_1$ denotes the \textit{empirical} distribution over the post-treatment period of $(Y_{0t}^o, X_t)$.   
        
        We obtain the critical quantile $q_{1 - \alpha}^*$ using the block bootstrap \citep{politis1992circular, politis1994stationary}. Formally, we resample the entire vector $(Y_{0t}, X_t)_{t=1}^{T}$,  over the pre-treatment period, constructing first $T_0$ units serving as pre-treatment period's observations and $T - T_0$ units serving as post-treatment period's observations. 
 We then construct the empirical measure of the bootstrapped sample $(F_0^*, F_1^*)$ and compute $\mathcal{T}^* = \mathcal{T}(F_0^*, F_1^*)$. A formal description is included in Algorithm \ref{alg:boot}.

The main intuition behind the inferential procedure is the following. 
We use a portion of the data to train predictors, while the remaining observations are used for the bootstrap estimate of the critical value when conducting hypothesis testing. We only estimate learners once and not on each bootstrapped sample. Figure \ref{fig:algsketch} provides a graphical illustration.


\begin{algorithm}    \caption{Testing Sharp Nulls: Basic Algorithm}\label{alg:boot}
    \begin{algorithmic}[1]
    \Require Observations $\{Y_{0t}, X_t\}_{t > 1}$, predictors $g_1(\cdot), \cdots, g_p(\cdot)$ estimated as in Algorithm \ref{alg:alg2}. 
      \For{$b=1,\dots, B$}
       \State Sample observations with replacement $\{Y_{0t}^o, X_t\}_{t > 1}$, and obtain bootstrap sample $\{Y_{0t}^{o*}, X_t^*\}_{t > 1}$ by performing circular block bootstrap  on  $\{Y_{0t}^o,X_t \}$ for $t \in \{1,\dots, T\}$; 
       \item Construct the empirical measure from the bootstrap sample $F_0^*, F_1^*$, and the test statistic $\mathcal{T}^* = \mathcal{T}(F_0^*, F_1^*)$; 
 \EndFor
 
           \State Compute $ q_{1-\alpha}^*$ as $(1-\alpha)$-th quantile  of the sample 
      
        \Return Reject the null hypothesis if
           $\mathcal{T} >q_{1-\alpha}^*$.          \end{algorithmic}
\end{algorithm}

Below, we formalize the validity of the bootstrap.\footnote{Lemma \ref{lem:identi2} provides the main intuition. The lemma implies that the distribution before and after the treatment must remain the same under the null hypothesis. Therefore, intuitively, we may expect that $(F_0^*, F_1^*)$ centered around the true empirical distribution converges to the same empirical process (after appropriate rescaling) of the limiting process of $(F_0 - \mathcal{D}_0, F_1 - \mathcal{D}_0)$, under the null hypothesis. We can then invoke Hadamard differentiability properties to show the bootstrap's validity.  We use such properties in the derivation of the validity of the bootstrap. } We impose the following condition.

\begin{ass} \label{ass:stat} Assume that $\{Y_{0t}^0, X_t\}_{t \geq 1}$ is $\beta$-mixing with mixing coefficients  $\sum_{k=1}^{\infty} (k+1)^2 \beta(k) < \infty$. In addition,  $g(\cdot), Y_{0t}^0$ are uniformly bounded almost surely. 
\end{ass} 

\begin{ass} \label{ass:df} Suppose that $w(\cdot)$ is Hadamard differentiable at $\mathcal{D}_0$ and uniformly bounded.  
\end{ass} 

In Appendix \ref{sec:hadamrd} we show that Assumption \ref{ass:df} holds for exponential weights considered in the following sections. We now introduce the first theorem. 

\begin{thm} \label{thm:1} \textit{Let Assumptions \ref{ass:stationarity}-\ref{ass:df} hold. 
 Let $\lim\sup_{T \rightarrow \infty}  b(T)/ \sqrt{T} < \infty$ and $\lim_{T \rightarrow \infty}   b(T) \rightarrow \infty$.
  Then,  under the null hypothesis, whenever  $p < \infty$
$$
\begin{aligned} 
\sup_x \Big|\mathbb{P}(\mathcal{T}^* - \mathcal{T}\le x|Y_{1:T},X_{1:T},T_0, H_0) -\mathbb{P}(\mathcal{T} - \mathbb{E}[\mathcal{T}] \le x| T_0, H_0) \Big| &= o_p(1) , \quad \text{ for } \quad T_0 \propto T \rightarrow \infty. 
\end{aligned} 
$$            
}  
\end{thm}

 \begin{cor}[Size control] Let the conditions in Theorem \ref{thm:1} hold. Then 
 $
 \lim_{T \rightarrow \infty} P\Big(\phi_\alpha(\mathcal{T}) = 1 | H_0\Big) = \alpha. 
 $
 \end{cor} 
 The above corollary follows from the validity of the bootstrap \citep[see for example,][]{fang2014inference} and it guarantees exact asymptotic coverage.

 The proof of Theorem \ref{thm:1} is contained in Appendix \ref{sec:proof}. Theorem \ref{thm:1} does not impose any restriction on the predictor other than Hadamard differentiability. Conditions on Hadamard differentiability are often imposed in the literature \citep{belloni2017program}, and require that weights are smooth functionals of the data. These are satisfied under mild conditions, and simple examples of weights that satisfy such conditions include least squares \citep{lunde2017bootstrapping} or exponential weights we discuss in Section \ref{sec:weights} (see Appendix \ref{sec:hadamrd} for a formal discussion). Note that the theorem is valid under misspecification. 
 
It is interesting to observe that if we were doing classical statistical inference, we would need to account for the estimation error generated by each individual prediction. Instead, the sample splitting procedure in Algorithm \ref{alg:alg2} combined with the resampling method in Algorithm \ref{alg:boot} guarantees valid inference, and it overcomes the complicated generated regressor problem. The key intuition is that the empirical distributions corresponding to the group of observations used for the training and those for the resampling are asymptotically independent under standard mixing conditions. As a result, \textit{sample splitting} guarantees that we can condition on the initial training period $F_-$, without affecting the asymptotic properties of the bootstrap. 
In this sense, our analysis extends the standard sample-splitting procedures employed in the $i.i.d.$ setting \citep{rinaldo2019bootstrapping} to dependent observations.

\begin{figure}[h]
\spacingset{1}
\centering
\includegraphics[scale=0.4]{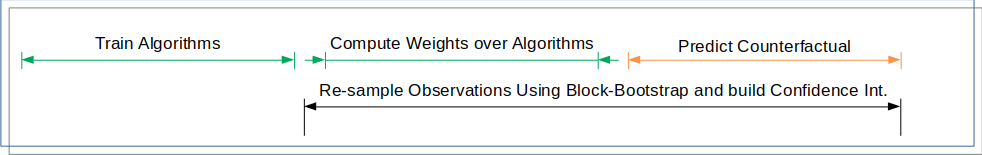}
\caption{The algorithm, in the presence of multiple post-treatment periods, works as follows: train learners on an initial sample, compute the weights on a consecutive block of observations and then predict the counterfactual. The green color denotes the pre-treatment period while the orange color denotes the post-treatment period;   bootstrap observations after imposing the null hypothesis.  } \label{fig:algsketch}
\end{figure}

 \begin{exmp}[Inference with the sample mean] \label{exmp:mean}
For an illustrative example, consider the following simple model:  
$$
Y_{jt} = \mu + \alpha 1\{t > T_0, j = 0\} + \varepsilon_{jt}, \quad \mathbb{E}[\varepsilon_t | T_0] = 0 \quad \forall t,
$$ and as estimator the difference-in-difference estimator 
$$
\hat{Y}_{0t}  = \frac{1}{|T_-|} \sum_{s < 0} Y_{0s},  
$$ 
where, for illustrative purposes, we use sample splitting for its construction.\footnote{Note that we could also have taken $\frac{1}{T_0} \sum_{s =1}^{T_0} Y_{0s}$, since the difference in means is Hadamard differentiable, with $w(F_0) g(X_t) = \int y dF_0(y)$ with $g(X_t) = 1$.}
Our corresponding test statistic takes the following form: 
 $$
 \mathcal{T} = \frac{1}{\sqrt{T - T_0}} \sum_{t > T_0} \Big| \varepsilon_{0t} - \frac{1}{|T_-|} \sum_{s < 0}  \varepsilon_{0s} + \alpha\Big|^2.
 $$ 
It is interesting to compare to the test statistic obtained from the permutation based method. This constructs the estimated counterfactual using the sample mean over the time-window $t \in \{1, \cdots, T_0, \cdots, T\}$, imposing a sharp null hypothesis \citep{chernozhukov2018exact}. Its corresponding test statics takes the following form 
\begin{equation}
\mathcal{T}^c = \frac{1}{\sqrt{T - T_0}} \sum_{t > T_0} \Big| \varepsilon_{0t} - \frac{1}{T} \sum_{t=1}^T  \varepsilon_{0t} + (1 - \lambda) \alpha \Big|^2, \quad \lambda  = \frac{T - T_0}{T}.  
\end{equation}   
To gain further intuition, let $\lambda \approx 1$ (i.e., the post-treatment period is larger than the pre-treatment period). Then, the distribution of $\mathcal{T}^c$, corresponding to the permutation-based approach, does not depend on the treatment effect $\alpha$, and therefore it cannot detect treatment effects. On the other hand, the dependence of the test statistic $\mathcal{T}$ with $\alpha$ remains invariant as $\lambda$ changes. \qed

\end{exmp}

 \begin{rem}[Choice of the test statistics] Alternative test statistics can be constructed by taking the $k$-norm of the residual error may also be considered \citep{chernozhukov2018exact}, 
        $$
\Big((T - T_0)^{-1/2}\sum_{t > T_0} \left(Y_{0t}^o - \hat Y_{0t} ^0 \right)^k\Big)^{1/k}, 
        $$
        for a general $k$ of the form. In this paper, the choice of $\mathcal{T}$ is motivated by two reasons: (i) the test statistics well captures permanent effects (i.e., effects exhibited over each post-treatment period) compared to test statistics having $k > 2$, which instead are better suited for the different case of large but temporary treatment effects, as in the case of $k = \infty$;\footnote{This is of interest also in our empirical applications, where the effects of a decrease in expenditures in Medicaid is expected to have long term effects on health-care outcomes, instead of large and temporary effects.} (ii) the test statistic presents desirable differentiable properties compared to the case of $k = 1$, which instead would not exhibit differentiability properties necessary for the validity of the resampling mechanism \citep[see, for example][]{fang2014inference}.
        \qed  
 \end{rem}

\begin{rem}[Testing the weak null hypothesis] \label{rem:average_null}  Our framework for inference discussed in Section \ref{sec:inference} also extends to null hypothesis of the form 
$$
H_0^{avg}: \mathbb{E}\Big[Y_{0t}^1 - Y_{0t}^0\Big] = a^o, \quad a^o \in \mathbb{R}, t > T_0. 
$$ 
While we omit details for brevity, we note that, in this secenario, the test statistics takes the form $\mathcal{T}_A = \Big((T - T_0)^{-1/2} \sum_{t > T_0} \Big(Y_{0t}^o - \hat{Y}_{0t}^0\Big) \Big)^2$. The test statistics is attractive in the presence of additive treatment effects and small deviations,\footnote{See for example the discussion at Page 65 in \cite{imbens2015causal}.} while the square (instead of the absolute value) guarantees differentiability. \qed  
\end{rem}

 \section{Average Treatment Effects, Bias Adjustment and Time-Varying Fixed Effects} \label{sec:average}

 In applications, we may also be interested in estimating the average treatment effect consistently in Definition \ref{defn:estimand}. 
A direct corollary of Lemma \ref{lem:1} is that a simple treated and controls difference, under standard mixing conditions, provides a consistent estimate of the treatment effect. However, such an estimate can present significant variance since we do not control the variation captured by covariates $X_t$, which can be non-vanishing with high-dimensional controls.

An alternative estimator for the average treatment effect takes the following form instead: 
$$
\frac{1}{T - T_0} \sum_{t > T_0 } \Big(Y_{0t} - \widehat{Y}_{0t}^0\Big).   
$$  

Unfortunately, such an estimator may be inconsistent for the true $\tau$, due to misspecification of the functions $g(x)$. We consider a bias adjustment instead.

The bias adjustment works as follows. Define $F_{0, -1/2}$ the empirical distribution of $(Y_t, X_t)$ for $t \in \{1, \cdots, T_0/2\}$, and $F_{0, 1/2}$ the empirical distribution of $(Y_t, X_t)$ for $t \in \{T_0/2  + 1, \cdots, T_0\}$. Then we construct an estimator of the form 
\begin{equation} \label{eqn:tau_hat} 
\begin{aligned} 
\widehat{\tau} & = \frac{1}{T - T_0} \sum_{t > T_0 } \Big(Y_{0t} - \widehat{Y}_{0t}^0 \Big) - \frac{1}{T_0/2} \sum_{t = T_0/2 + 1 }^{T_0} \Big(Y_{0t} - \widehat{Y}_{0t, -1}^0 \Big), \quad \widehat{Y}_{0t, -1}^0 = w(F_{0, -1/2}) g(X_t)
\end{aligned} 
\end{equation}  
where $\widehat{Y}_{0t, -1}^0$ is the estimator whose weights are constructed using only the first half of the treatment period, and $g$ is constructed as in Algorithm \ref{alg:alg2}. The first difference is taken over the post-treatment period, with weights computed over the entire pre-treatment period. The second difference is taken over the second half of the pre-treatment period, but with weights computed over the first half. The second difference defines the bias adjustment. 

To gain further intuition on the bias adjustment, note that we can write 
$$
\hat{\tau} = \tau + o_p(1)  - \underbrace{\int \Big(w(F_0)g(x) - w(F_{0, -1/2})g(x)\Big) d(F_1(y, x) - F_{0, 1/2}(y, x))}_{ = o_p(1)} 
$$ 
where $F_1$ denotes the empirical distribution over the treated period, $F_{0, -1/2}$ denotes the empirical distribution over the first half of the pre-treatment period, and $F_{0,1/2}$ denotes the empirical distribution over the second half. The above estimator is a difference in difference, where the first component converges to $\tau$ while the second component is of order $o_p(1)$ as shown in the following theorem.

We first impose the following condition. 

\begin{ass}[Potential outcome $Y_{0t}^1$] \label{ass:pot1} Assume that 
$$
Y_{0t}^1 = \mu_{0t}^1 + \varepsilon_{0t}^1, \quad \mathbb{E}[\varepsilon_{0t}^1] = 0, 
$$ 
where $\mu_{0t}^1 < \infty$ denotes some (possibly non-stationary) expectation and $\varepsilon_{0t}^1$ is stationary. 
\end{ass} 

Assumption \ref{ass:pot1} states that the potential outcome under control can be decomposed into two components: a (possibly non-stationary) expectation and an additive stationary idiosyncratic shock. We can now introduce the following theorem.

\begin{thm} \label{thm:consistency} Let Assumption \ref{ass:stationarity}, \ref{ass:ident}, \ref{ass:df}, \ref{ass:pot1} holds.  Assume that $(\varepsilon_{0t}^1, Y_{0t}^0, X_t)$ is $\beta$-mixing with mixing coefficients  $\sum_{k=1}^{\infty} (k+1)^2 \beta(k) < \infty$, with $g(\cdot; F_-)$ uniformly bounded. Then 
$
\hat{\tau} - \tau \rightarrow_p 0.
$
\end{thm} 

Theorem \ref{thm:consistency} shows consistency of the estimated effect. It imposes standard mixing conditions. Its proof is in Appendix \ref{sec:proof}.

\subsection{Extension to Time-Varying Fixed Effects} \label{sec:time_fixed}

In this section, we turn to the case where time-fixed effects occur and illustrate how our results extend to this case without necessitating estimating the fixed effect. 

We consider the following data generating process.

\begin{ass}[Fixed effects and Stationarity] \label{ass:stationarityb} Suppose that 
$$
Y_{jt}^0 =  \kappa_t^0 + \iota_j^0 + \varepsilon_{jt}^0, \quad \mathbb{E}[\varepsilon_{jt}^0] = 0 \quad j \in \{0, \cdots, n\}, t \le T, 
$$   
where $\iota_j^0 < \infty, \kappa_t^0 < \infty$ denotes individual and time-fixed effect and 
$
(\varepsilon_{0t}^0, \varepsilon_{1t}^0, \cdots, \varepsilon_{nt}^0, Z_{0t}, \cdots, Z_{nt}) \sim \mathcal{D} ,  
$ 
define arbitrary unobservables and observables which follow a stationary process. 
\end{ass} 

The above assumption allows for the failure of stationary as long as the time fixed effects are the same between the control and treated unit. The following lemma illustrates identification for this case. 

\begin{lem}[Identification of $\tau$] \label{lem:1b} Let Assumptions \ref{ass:ident}, \ref{ass:stationarityb} hold. Then 
$$
\tau = \frac{1}{T - T_0} \sum_{t > T_0} \Big\{\mathbb{E}\Big[Y_{0t}^1|T_0\Big] - \frac{1}{n} \sum_{j > 0} \mathbb{E}\Big[Y_{jt}^0\Big|T_0\Big] \Big\} - \frac{1}{T_0} \sum_{s = 1}^{T_0} \Big\{\mathbb{E}\Big[Y_{0s}^0|T_0\Big]  - \frac{1}{n} \sum_{j > 0} \mathbb{E}\Big[Y_{js}^0\Big|T_0\Big] \Big\}. 
$$  
\end{lem} 

The above lemma shows that we can identify the treatment effect by taking a difference in difference. The following lemma discusses testable implications for the sharp null hypothesis.
 
\begin{lem}[Sharp Null: Testable Implication] \label{lem:identi2b} Under the null hypothesis in Equation \eqref{eqn:null} and Assumptions \ref{ass:ident}, \ref{ass:stationarityb} then $(Y_{0t}^o - \bar{Y}_t, Y_{1t} - \bar{Y}_t, \cdots, Y_{nt} - \bar{Y}_t, Z_{0t}, \cdots, Z_{nt} )$ is stationarity and independent of $T_0$ for all $t \in \{1, \cdots, T_0, \cdots, T\}$, with $\bar{Y}_t = \frac{1}{n} \sum_{j = 1}^n Y_{jt}$. 
\end{lem} 

The proofs of the above lemmas are in Appendix \ref{app:lemmas}. 
The counterfactuals' estimation follows similarly as before, with a minor modification: we subtract the controls' average outcome from the outcome of the treated unit and from the control unit to guarantee stationarity. This is illustrated below. 
$$
\hat{Y}_{0t}^0 - \bar{Y}_t = w(\tilde{F}_0)g(\tilde{X}_t), \quad \bar{Y}_t = \frac{1}{n} \sum_{j = 1}^n Y_{jt}, \quad \tilde{X}_t = \Big(Y_{1t} - \bar{Y}_t, \cdots, Y_{nt} - \bar{Y}_t, Z_{0t}, \cdots, Z_{nt}\Big), 
$$ 
where $\tilde{F}_0$ denotes the empirical distribution of $(Y_{0t} - \bar{Y}_t, \tilde{X}_t)$ over the pre-treatment period. 
 Here, $\bar{Y}_t$ denotes the time-specific average of the \textit{control} units, but not of the treated. Theorem \ref{thm:1} and Theorem \ref{thm:consistency} directly follows once we subctract from the outcomes $(Y_{0t}, \cdots, Y_{nt})$ the controls' average $\bar{Y}_t$ to guarantee stationarity.

\section{Ensamble Weights and Prediction Guarantees under Non-stationarity} \label{sec:weights}

This section discusses a particular weighting scheme, exponential weights, and derives its properties in terms of prediction guarantees.

 \subsection{Ensamble: Exponential weights}	\label{subsec:weights}

Although least-squares have been considered \citep[see, e.g.,][]{kunzel2017meta,polley2010super, elliott2004optimal}, these can perform poorly when the number of learners is large compared to the sample size.\footnote{Observe that weights computed via least-squares are Hadamard differentiable \citep{lunde2017bootstrapping}, hence satisfying the conditions in Theorem \ref{thm:1}.} Equal weighting, on the other hand, can perform poorly in the presence of many uninformative learners. To equip the method to have a better performance in the presence of a large number of algorithms, some of which may potentially be ineffective,
 we discuss an alternative weighting scheme.
 
  With a slight abuse of notation, we index weights by the time period as $w(t)$, with $t \in \{1, \cdots, T_0\}$. For example $w(\tilde{t})$ denotes the weight estimated using the empirical distribution from time $t = 1$ to time $t= \tilde{t}$.
 
 The weigthing scheme we focus are exponential weights of the following form: we write 
 \begin{equation}
	\label{eqn:expweights1}
	w^{(j)}(T_0) = \frac{ \exp \left\{ -\eta \sum_{s=1}^{T_0}  \Bigl(Y_{0s}  - g_{j}(X_s) \Bigl)^2\right\}}{\sum_{i=1}^{p} \exp\left\{ -\eta \sum_{s=1}^{T_0}   \Bigl(Y_{0s} -  g_{i}(X_s) \Bigl)^2 \right\} } . 
	\end{equation}
 Such weights have been widely discussed in literature on exponential aggregation, see e.g., \cite{cesa1999prediction}, \cite{rigollet2012sparse} among others. 
 
 By choosing the tuning parameter $\eta \propto 1/T_0$ exponential weights can be written as differentiable functionals of $F_0$. This is discussed in Appendix \ref{sec:hadamrd}.
 
We can motivate Equation \eqref{eqn:expweights1} as the solution to a penalized surrogate risk minimization, discussed below.  
 
\begin{exmp}[Model average as a surrogate risk minimization problem] \label{exmp:jhg} Consider the following optimization problem \citep[e.g.,][]{rigollet2012sparse} 
$$
\text{min}_{w \in \mathcal{W}} \{\sum_{t=1}^{T_0} w_j (Y_{0t} -  \sum_{j=1}^{p} g_j(X_t))^2 + \text{pen}(w) \}, \quad \mathcal{W}= \{ w \in \mathcal{R}^{p}: w_j \ge 0; \sum_{j=1}^{p} w_j = 1\}, 
$$
where $\text{pen}(w)$ denotes a penalty on the weights. Under convexity of the loss function, the above optimization problem is interpreted as minimizing a surrogate loss function of the penalized empirical risk.\footnote{Observe that we can write the risk as $\{\sum_{t=1}^{T_0} l(Y_{0t}, \sum_{j=1}^{p} w_j f_{j,t}) + \text{pen}(w) \}
$
where, by convexity, $\sum_{j=1}^p w_j l(Y_{0t}, \sum_{j=1}^{p} f_{j,t}) \ge l(Y_{0t}, \sum_{j=1}^{p} w_j f_{j,t})$. Surrogacy is often considered in decision problem for its computational appeals. }
By letting $\text{pen}(w) = \frac{1}{\eta} \sum_{j=1}^{p} w_j \text{log}(w_j)$, the solution to the problem reads as 
\begin{equation} \label{eqn:exp1} 
w^{(j)}(T_0) = \frac{ \exp(-\eta \sum_{t=1}^{T_0} (Y_{0t} - g_{j}(X_t))^2)}{\sum_{i=1}^{p} \exp(-\eta \sum_{t=1}^{T_0} (Y_{0t} - g_{i}(X_t))^2)}.  
\end{equation} 
Intuitively, the method assigns larger weights to those predictors that have the lowest out-of-sample loss function. These weights inherit oracle guarantees discussed in Section \ref{sec:pp}. \qed 

\end{exmp}

\subsection{Prediction Guarantees}  \label{sec:pp}

We now derive prediction guarantees of the exponential weights without imposing stationarity conditions. Our result illustrates oracle guarantees of the exponential aggregation method \citep{cesa2006prediction}, here applied to the different contexts of counterfactual prediction.

 We study the behavior of our algorithm trained only on $t-1$ observation and evaluated at the $th$ observation, i.e., we are willing to provide theoretical guarantees on the following cumulative loss. 
\begin{equation} \label{eqn:cum1} 
T_0^{-1} \sum_{t=1}^{T_0} (\hat Y_{0t}^0( \mathcal{F}_{t-1}) - Y_{0t}^0)^2,
\end{equation} 
where $\hat{Y}_t^0(\mathcal{F}_{t-1}) = w(t-  1) g(X_t)$. Here, $\hat Y_{0t}^0( \mathcal{F}_{t-1})$ denotes the prediction at time $t$ using only information at time $t-1$.
Since $\hat Y_{0t}^0( \mathcal{F}_{t-1})$ is estimated only on the previous data and evaluated at $X_t$, this notion of performance is rooted in out-of-sample performance metric. 

We first study the cumulative loss in \eqref{eqn:cum1} compared to the smallest cumulative loss incurred by any of the algorithms under consideration, defined as
 \begin{align*}
 \mathcal{R} &=  T_0^{-1} \sum_{t=1}^{T_0} (\hat{Y}_t^0(\mathcal{F}_{t-1}) - Y_{0t}^0)^2 - \min_{i \in \{1,\dots,p\}} T_0^{-1} \sum_{t=1}^{T_0}  (g_i(X_t)  - Y_{0t}^0)^2. 
\end{align*}

 In the following theorem, we consider the case where $T_0 = \lambda T$ where $\lambda \in (\gamma,1 - \gamma)$ is potentially a random variable for some constant $\gamma > 0$. 
\begin{thm} \label{thm:regret0}
Suppose that $(Y_{0t},  g(X_t)) \in [-M, M]^{p + 1}$, for some $M< \infty$. Consider an exponential weighting scheme as in \eqref{eqn:expweights1} with $\eta \propto \sqrt{{\log(p)}/{T_0}}$. 
Then with probability at least $1 - 2\delta$, 
$
\begin{aligned} 
&\mathcal{R}  \le C_0 \sqrt{\frac{ {\log}(p/\delta)}{ \gamma T}} 
\end{aligned}
$
 for $C_0 < \infty$ being a   constant independent of $T_0$ or $p$. 
\end{thm}

The proof is presented in Appendix \ref{app:5}.
 
\noindent Theorem \ref{thm:regret1} provides an error bound for the empirical one step ahead prediction error. Remarkably, it does not require any stationarity assumption. 
 The bound scales logarithmically with the number of learners, and it scales at square-root $T$ with the length of the sequence. 

In the following lines, we provide stronger guarantees with respect to the \textit{conditional} expectation of $Y_{0t}^0$. For the next theorem to hold, we need to introduce an additional condition, which replaces Assumptions \ref{ass:stationarity} and \ref{ass:ident} that we imposed in previous sections.  

 \begin{ass} \label{ass:regretmodel} (Additive Error Model and Sequential Ignorability) Let the following hold
\begin{equation} \label{eqn:dgp1} 
\begin{aligned}
&Y_{0t}^0  = \mu_t(X_t) + \varepsilon_{0t}.
\end{aligned} 
\end{equation}  
where   $\mathbb{E}[\varepsilon_{0t} |X_t, \mathcal{F}_{t-1}] = 0$. Assume in addition that $\varepsilon_{0t} \perp T_0 |X_t, \mathcal{F}_{t-1}$, where $\mathcal{F}_{t-1}$ denotes the filtration at time $t - 1$. 
\end{ass}

Assumption \ref{ass:regretmodel} states that the potential outcome can be decomposed into two main components, a conditional expectation function $\mu_t(\cdot)$ and idiosyncratic shocks. Note that such a condition is only required to derive guarantees with respect to $\mu_t(\cdot)$.  
It is natural to compare the cumulative loss in \eqref{eqn:cum1} with the smallest cumulative loss incurred by any of the algorithms under consideration. We define such a metric of comparison as the Conditional Mean Proxy Regret (CMPR). 
\begin{align*}
 \mathcal{R} ^{\mu} &=  T_0^{-1} \sum_{t=1}^{T_0} (\hat{Y}_t^0(\mathcal{F}_{t-1}) - \mu_t(X_t))^2 - \min_{i \in \{1,\dots,p\}} T_0^{-1} \sum_{t=1}^{T_0}  (g_i(X_t)- \mu_t(X_t))^2. 
\end{align*}
The above definition incorporates notions of performance with respect to the conditional mean function (as opposed to the outcome itself). Our definitions above combine definitions in the literature on prediction of individual sequences \citep{cesa1999prediction} with the literature on causal inference. The main difference with standard notions
of regret is that CMPR is based on the unobserved deviation of the predicted counterfactual from the
conditional mean evaluated at $X_t$, and not just on the cumulative loss of the predictor

\begin{thm} \label{thm:regret1}
 Let Assumption \ref{ass:regretmodel} hold and let $(Y_{0t}, \mu_t(X_t), g(X_t)) \in [-M, M]^{p + 2}$, for some $M< \infty$. Consider an exponential weighting scheme as in \eqref{eqn:expweights1} with $\eta \propto \sqrt{{\log(p)}/{T_0}}$. 
Then with probability at least $1 - 2\delta$, 
$$
\begin{aligned} 
&\mathcal{R} ^{\mu} \le C_0 \sqrt{\frac{ {\log}(p/\delta)}{ \gamma T}} ,   
\end{aligned}
$$
 for $C_0$ being a   constant independent of $T_0$ or $p$. 
\end{thm}

The proof is presented in Appendix \ref{app:5}.
 
\noindent Theorem \ref{thm:regret1} provides an error bound for the empirical one step ahead prediction error with respect to the \textit{conditional mean}.

 If we are willing to assume more, in that our class of algorithms contains one learner that consistently estimates the unknown model, then previous results imply that our synthetic learner will preserve that consistency regardless of the number of learners in the entire class. We consider below asymptotics for $T \rightarrow \infty$ and $T_0 = \lambda T$ where $\lambda \in (0,1)$ is potentially a random variable. 
 
\begin{cor}
Suppose that  the number of learners is such that $\log(p)/T^{1/2} = o(1)$ and conditions in Theorem \ref{thm:regret1} hold. Assume also that the following holds 
$
\min_{i \in \{1, \dots,p\}} T_0^{-1}\sum_{t=1}^{T_0} |\mu_t(X_t) -  g_{i}(X_t)|^2 = o_p(1).
$
Then,
$$
T_0^{-1} \sum_{t=1}^{T_0}  (\hat{Y}_t^0(\mathcal{F}_{t-1}) - \mu_t(X_t))^2 = o_p(1).
$$
for $T\rightarrow \infty$, $T_0 = \lambda T$.
\end{cor}





\section{Numerical Experiments} \label{sec:experiments}

In this section, we study the performance of the method in the presence of linear and non-linear outcome models, allowing for the presence of many non-informative methods. We compare the methodology to existing testing procedures, including permutation tests of Synthetic Control as well as the Difference-in-Difference method, and showcase a significant improvement.

 \subsection{Experimental Setups}

 We describe our experiments in terms of the outcome model as well as the model of the design of the covariates and the error terms. In our first experiment,  \textbf{DGP1},  we considered a simple Linear Outcome Model
 $$
Y_{0t} =  X_t \beta + a_t D_t + \epsilon_t
$$
and tested the ability of our method to detect changes in the treatment effect $a_t$. This example is intended to model a setting where classical Synthetic Control method is optimal.  Here we set
 $\beta_j = 1/(1+j)^2$, $j=1,\dots,p$, with the last beta chosen such that $ \sum_{j } \beta_j =1$, where we consider $p \in \{10, 50\}$. 
 The parameter $\beta$ will be kept as above for all our experiments.
 We considered a simple  {\sc ar}  model for the errors $\epsilon_t$ with 
  $\epsilon_t =  0.6\epsilon_{t-1} + v_t$ and  $v_t \sim \mathcal{N}(0,1 - 0.6^2)$.
Control units are generated according to a factor model as 
$$
X_{j,t} = \mu_j + \theta_t +  \lambda_j F_t + u_t 
$$ 
with unit specific term $\lambda_j = \mu_j =  (1+j)/j$ a time random effect $\theta_t \sim \mathcal{N}(0,1)$ and an unobserved factors  $F_t \sim \mathcal{N}(0,1)$. Errors $u_t$ follow an  {\sc ar}  model
$u_t = 0.6u_{t-1} + h_t$ with $h_t \sim \mathcal{N}(0, 1 - 0.6^2)$.   In our second experiment, \textbf{DGP2}, we considered a 
 Logistic-like Outcome Model $$
Y_{0t} =a_t D_t + \exp(X_t \beta + \epsilon_t)/(1+\exp(X_t \beta + \epsilon_t))  
$$
with $\epsilon_t = 0.5 \epsilon_{t-1} + 0.3 v_{t-1} + v_t$. This experiment has three settings: (a), (b) and (c). Setting (a) and (b) assume   $v_t \sim \mathcal{N}(0,\sigma^2)$ with $(a)\sigma = 0.1$ and $(b) \sigma = 1$, respectively. Setting (c) assumes $ \epsilon_t = 0.8 \epsilon_{t-1} + v_t$,  with
$
v_t = \sqrt{h_t} z_t, \quad h_t = 0.001 + 0.99 v_{t-1}^2
$
with $z_t \sim \mathcal{N}(0,1)$({\sc ar-arch} process). We report here (a) and (b).  In addition we let 
$
X_t = h_t + u_t
$
with $h_t$ being i.i.d over time with  $ \mathcal{N}(0, \Sigma)$ distribution  with $\Sigma_{i,j} = 0.5^{|i-j|}$ and $u_t = 0.8 u_{t-1} + k_t$ with $k_t \sim \mathcal{N}(0,1 - 0.8^2)$. This setting is designed to test the ability of the proposed Synthetic Learner to adapt to nonlinear outcome model. 
We consider our third setting, \textbf{DGP3}, that follows a factor model
 $$
Y_{0t} = 0.5  + a_t D_t + \theta_t + 0.5 F_t + \epsilon_t .
$$
Error and design structures are the same as that of \textbf{DGP1}.  \textbf{DGP4} considers an interaction outcome model that is polynomial in structure
$$
Y_{0t} = a_t D_t  +(X_{1,t} + X_{2,t} + \dots + X_{10,t})^2 + \epsilon_t 
$$
with $\epsilon_t$ being the same as in  \textbf{DGP2}(a) design $X_t$ is the same as throughout   \textbf{DGP2}; \textbf{DGP5} postulates a cosine, hence periodic, type of outcome model
 $$
Y_{0t} =  \cos(X_t\beta + \epsilon_t)  + a_t D_t .
$$
Error and design setting have three components: (a), (b) and (c) that are following the setup of  \textbf{DGP2} (a), (b) and (c), respectively. Finally, \textbf{DGP6} is a simple non-stationary model, which follows similarly to DGP3, but with $F_t \sim \mathcal{N}(\mathrm{cos}(t), 1)$, with $\mathrm{cos}(t)$ capturing a non-stationary component.\footnote{Note that the component cannot be removed from simple transformations such as differentiating.}  As discussed in the following subsections, we choose $p = 10$ for smaller $T (T \in \{60, 80, 100\})$ and $p = 50$ for larger $T \ge 300$.

\subsection{Testing} 

We consider testing the following hypothesis $H_0$:  
$$
H_0: Y_{0t}^1 - Y_{0t}^0 = 0, \quad  t > T_0.
$$ 
We consider the Synthetic Learner with experts, including a naive XGboost (which uses the default tuning parameter of the package XGboost in R), Support Vector Regression, and {\sc arima}(0,1,1) with external regressors together with $50$ non-informative learners.  Non-informative experts are randomly drawn from a multivariate gaussian with a full covariance matrix.

\subsubsection{Power Study: Comparison with SC and DID}

First, we compare Synthetic Learner's performance to existing procedures whose theoretical properties are well studied. In particular, we compare the Synthetic Control (SC) with weights being constrained to sum to one  and an intercept according to Equation (7) and (8) in \cite{chernozhukov2017exact}, as well as 
the Difference in Difference (DiD) estimator, namely
$$
\hat{Y}_t^{DiD} = \hat{\alpha} + (\hat{\beta} + \hat{\Delta})   \mathbbm{1}_{t > T_0}
$$ 
with coefficient computed as in a standard DiD with the two periods corresponding to pre and post-treatment periods.
 We consider the test statistics for  
  Synthetic Control  

\begin{equation} \label{eqn:t_stat_sc} 
\frac{1}{\sqrt{T - T_0}} \sum_{t = T_0 + 1}^T |Y_{0t} - a_t^o - X_t \hat{w}^0_{SC}|^2
\end{equation}
where $\hat{w}_{SC}^0$ are computed via constrained Least Squares, with coefficients summing to one for Synthetic Control. Finally, we consider 
\begin{equation} \label{eqn:t_stat_sc} 
\frac{1}{\sqrt{T - T_0}} \sum_{t = T_0 + 1}^T |Y_{0t} - a_t^o - \hat{Y}_t^{DiD}|^2
\end{equation}
for Difference-in-Differences. 
 In Figures \ref{fig:plotsim60}, \ref{fig:plotsim_non_stat} and \ref{fig:plotsim100} we compare the performance of our method to permutation tests where $\hat{w}_{SC}$ and $\hat{Y}_t^{DiD}$ must be computed on the entire sample, as described in \cite{chernozhukov2017exact}. We run the Synthetic Learner after training on the period running from $1$ to $T_- = T_0/2, T \in\{60, 100\}, p = 10, T - T_0 = 10$, and we use the remaining observations for computing weights and bootstrap. We consider different treatment effects, denoted by $\alpha$, and report on the x-axis the effect of the policy $\alpha$ divided by the (unconditional) standard deviation of the outcome.

We present power plots  across different $T$ in Figures \ref{fig:plotsim60},  and \ref{fig:plotsim100} for $T = 60$ and $T = 100$ respectively, while Figure \ref{fig:plotsim_non_stat} collects results for the non-stationary DGP.
Across all figures, we observe an improvement over permutation tests with both SC and DiD methods, with more significant improvements for the non-linear DGPs. 

Improvements can be due to two factors: bootstrap outperforming permutation as well as Synthetic Learner's better performance in comparison to SC and DiD. Table \ref{tab:perm_vs_boot} (see the discussion below) and Appendix \ref{app:oracle_quantile} provide suggestive evidence that improvements are due to both factors.  
 
In particular, in Appendix \ref{app:oracle_quantile} we also consider the oracle case where the critical value is known. This case permits to compute $\hat{w}_{SC}$ and $\hat{Y}_t^{DiD}$ only using information until time $T_0$, as discussed in \cite{doudchenko2016balancing}. We show improvements also in this setting.  In Appendix \ref{app:vary_T}, we report a more extensive study with $T_0$ and $T - T_0$ vary and show the robustness of our results to different settings.

\begin{figure}[!ht]
\centering
\includegraphics[scale=0.5]{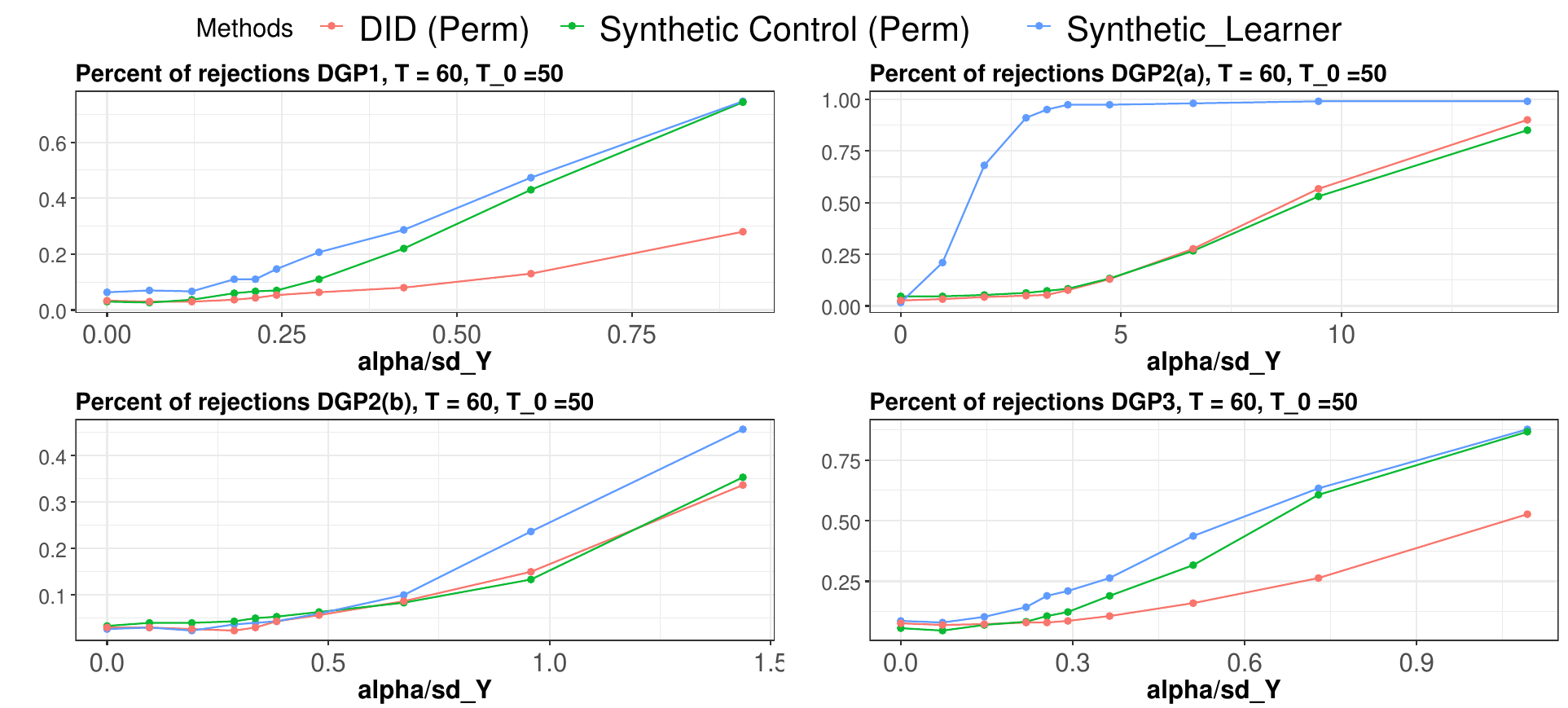}
\includegraphics[scale=0.5]{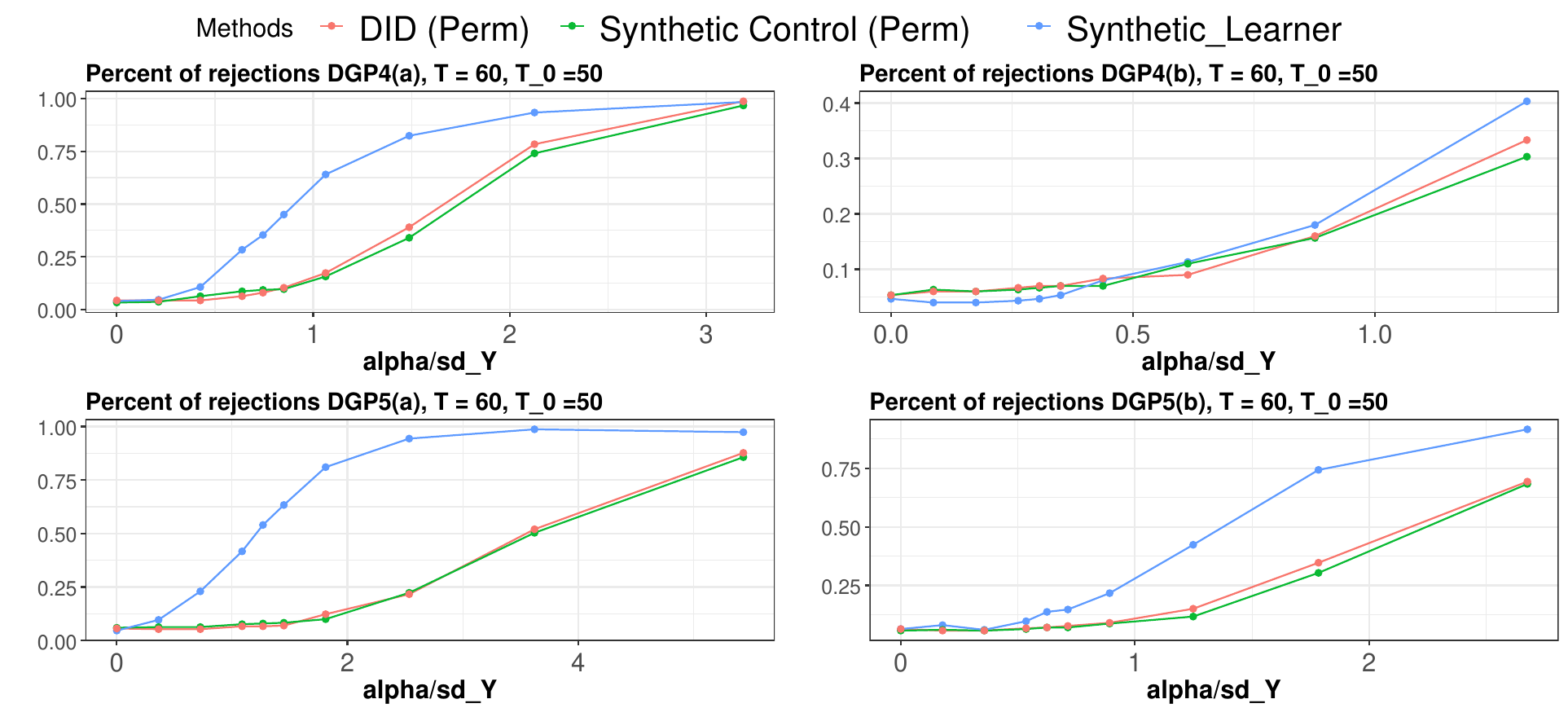}
\caption{$T = 60, p = 10, T_0 = 50$. Percentage of rejections of the null hypothesis of no treatment effects over $300$ repetitions. The x-axis reports the policy's effect rescaled by the outcome's standard deviation. Synthetic learner has XGboost,Support Vector Regression and {\sc arima}(0,1,1) and $50$ additional non informative predictions. The blue line denotes the proposed method, the red line denotes the difference-in-differences, and the green line denotes the Synthetic Control. }
\label{fig:plotsim60}
\end{figure}

\begin{figure}[!ht]
\centering
\includegraphics[scale=0.4]{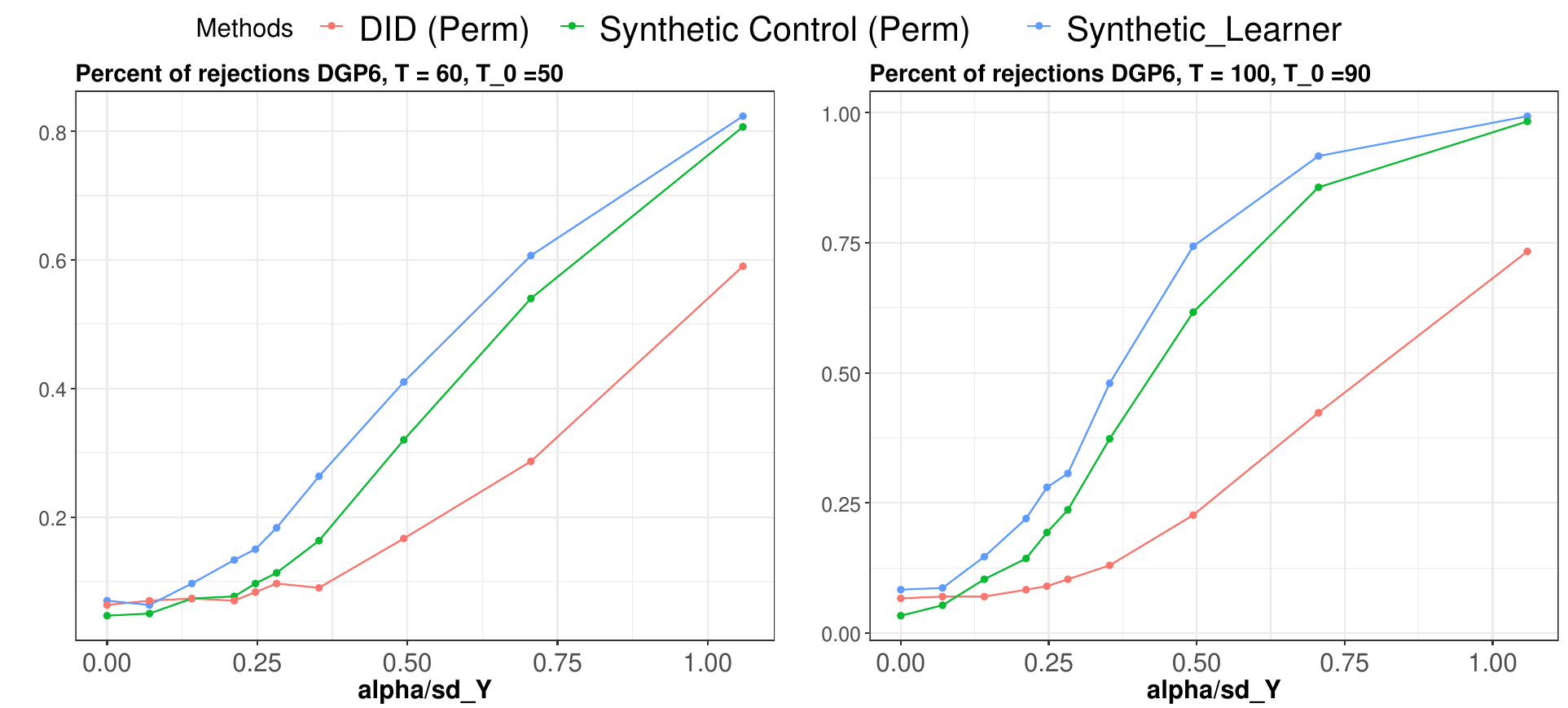}
\caption{$DGP6, T \in \{60, 100\}, p = 10$. Percentage of rejections of the null hypothesis of no treatment effects over $300$ repetitions. The x-axis reports the policy's effect rescaled by the outcome's standard deviation. Synthetic learner has XGboost,Support Vector Regression and {\sc arima}(0,1,1) and $50$ additional non informative predictions.  The blue line denotes the proposed method, the red line denotes the difference-in-differences, and green line denotes the Synthetic Control. }
\label{fig:plotsim_non_stat}
\end{figure}

\begin{figure}[!ht]
\centering
\includegraphics[scale=0.5]{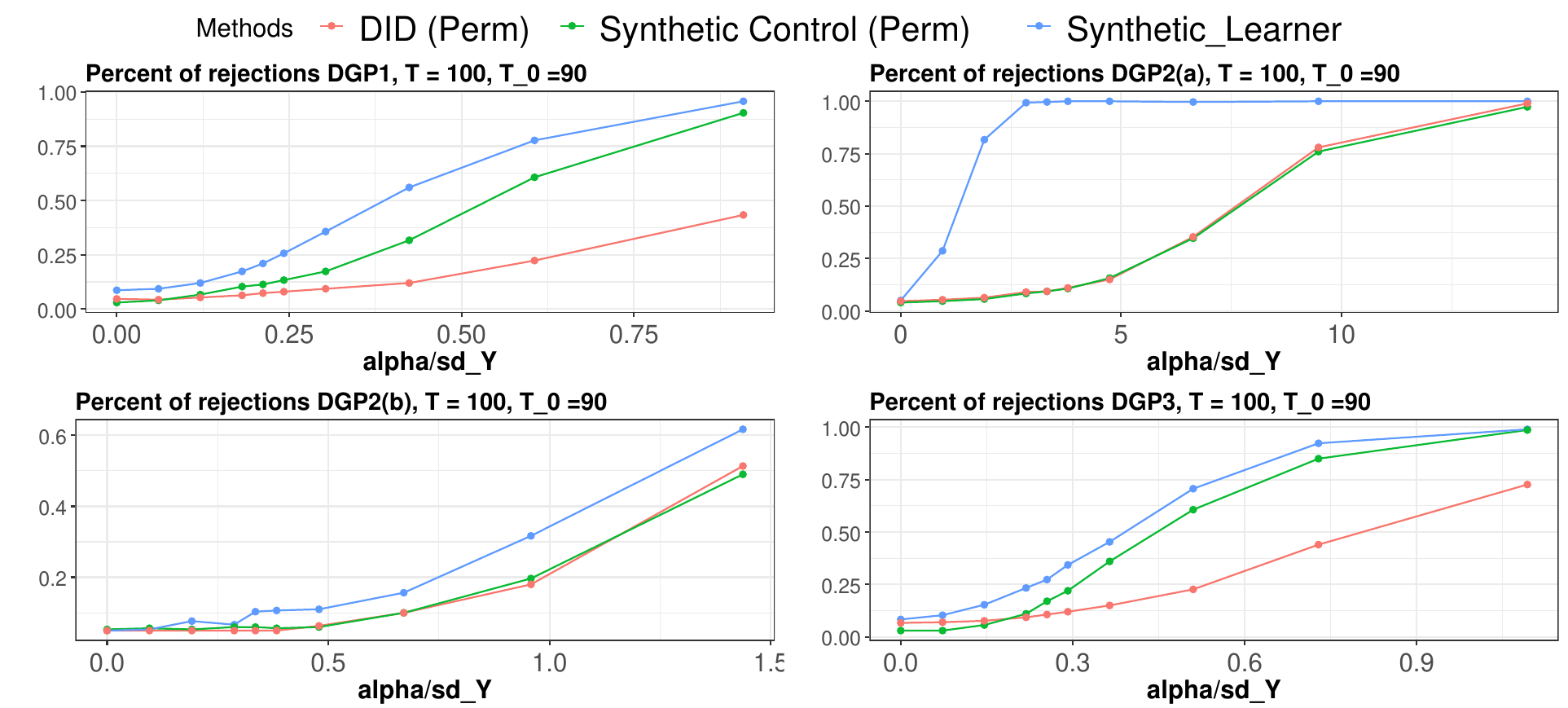}
\caption{$T=100, p = 10$ Percentage of rejections of the null hypothesis of no treatment effects over $300$ repetitions. The x-axis reports the policy's effect rescaled by the outcome's standard deviation.  Synthetic learner has XGboost,Support Vector Regression and {\sc arima}(0,1,1) and $50$ additional non informative predictions.  The blue line denotes the proposed method, the red line denotes the difference-in-differences, and the green line denotes the Synthetic Control. }
\label{fig:plotsim100}
\end{figure}


 \subsubsection{Size Control}

 Next, we study the size of our procedure for $T \in \{60, 80\}$ and $T^* = T - T_0 \in \{5, 10, 20\}$, as we vary the post-treatment period from small to longer post-treatment period. These are reported in Table \ref{tab:size}. We observe that in a finite sample, for $T$ relatively small, our test controls size across all DGP, except DGP3 and DGP6, where we observe a small size distortion of five percentage points for a short post-treatment period ($T^* = 5$). These results provide suggestive evidence of the correct size of the proposed test also in a finite sample.

 \begin{table}[!htbp] \centering 
  \caption{Size of the Synthetic Learner for $T \in \{60, 80\}$ and varying post-treatment period length $T^*$, for tests with size $5\%$. The first panel reports the size for $T = 60$ and the second panel for $T = 80$. } 
  \label{tab:size} 
\begin{tabular}{@{\extracolsep{5pt}} cccc|ccc} 
\\[-1.8ex]\hline 
\hline \\[-1.8ex] 
 & $T = 60$ &  & & $T  = 80$ &  &  \\ 
\hline \\[-1.8ex] 
 & $T^*  = 5$ & $T^*  = 10$ & $T^*  = 20$ & $T^*  = 5$ & $T^*  = 10$ & $T^*  = 20$ \\ 
\hline \\[-1.8ex] 
DGP1 & $0.087$ & $0.063$ & $0.037$ & $0.077$ & $0.063$ & $0.030$ \\ 
DGP2(a) & $0.037$ & $0.017$ & $0.023$ & $0.050$ & $0.053$ & $0.037$ \\ 
DGP2(b) & $0.053$ & $0.027$ & $0.013$ & $0.080$ & $0.080$ & $0.023$ \\ 
DGP2(c) & $0.063$ & $0.043$ & $0.033$ & $0.080$ & $0.070$ & $0.033$ \\ 
DGP3 & $0.090$ & $0.087$ & $0.050$ & $0.097$ & $0.060$ & $0.040$ \\ 
DGP4(a) & $0.060$ & $0.040$ & $0.027$ & $0.040$ & $0.063$ & $0.033$ \\ 
DGP4(b) & $0.067$ & $0.047$ & $0.040$ & $0.090$ & $0.053$ & $0.043$ \\ 
DGP4(c) & $0.080$ & $0.057$ & $0.037$ & $0.057$ & $0.067$ & $0.020$ \\ 
DGP5(a) & $0.073$ & $0.047$ & $0.020$ & $0.060$ & $0.077$ & $0.037$ \\ 
DGP5(b) & $0.050$ & $0.063$ & $0.027$ & $0.067$ & $0.043$ & $0.030$ \\ 
DGP5(c) & $0.103$ & $0.040$ & $0.020$ & $0.060$ & $0.057$ & $0.040$ \\ 
DGP6 & $0.090$ & $0.070$ & $0.090$ & $0.107$ & $0.067$ & $0.040$ \\ 
\hline \\[-1.8ex] 
\end{tabular} 
\end{table} 

 \subsubsection{Oracle Study: Learners' performance}

In Table \ref{tab:oracle_base} we study the performance of the algorithm and each base algorithm for $T = 60, T - T_0 = 10$. Table \ref{tab:oracle_base} reports the power of each base algorithm and its corresponding weight assigned by the Synthetic Learner. We observe two striking facts: first, the largest weight is assigned to the best performing algorithm; second, the Synthetic Learner always performs approximately the same or \textit{better} than any base algorithm under consideration. This result suggests the benefits of the ensemble procedure: the procedure combines predictions of different methods to maximize prediction (and ultimately power) optimally. For the sake of brevity, we report results for two linear DGPs (DGP1 and DGP2) and two non-linear ones.

 \begin{table}[!htbp] \centering 
  \caption{$T = 60, p = 10$. Power of the Synthetic Learner and of each base algorithm. In each panel, the first two columns report the power for $\alpha \in \{0.1, 0.3\}$, and the third column reports the assigned weight of the Synthetic Learner.  } 
  \label{tab:oracle_base} 
\begin{tabular}{@{\extracolsep{5pt}} cccc|ccc} 
\\[-1.8ex]\hline 
\hline \\[-1.8ex] 
 & DGP1\textcolor{white}{(a)} &  &  & DGP2(a) &  & \\ 
 & 0.1 & 0.3 & Weights & 0.1 & 0.3 & Weights\\ 
\hline \\[-1.8ex] 
Synthetic Learner  & $0.070$ & $0.110$ & $$  & $0.210$ & $0.910$ & $$ \\ 
XGBoost& $0.030$ & $0.047$ & $0.153$ & $0.157$ & $0.823$ & $0.345$\\ 
SVM  & $0.040$ & $0.037$ & $0.044$  & $0.213$ & $0.910$ & $0.375$  \\ 
Arima & $0.060$ & $0.103$ & $0.803$ & $0.163$ & $0.693$ & $0.280$ \\ 
\hline \\[-1.8ex] 
\end{tabular} 
\begin{tabular}{@{\extracolsep{5pt}} cccc|ccc} 
\\[-1.8ex]\hline 
\hline \\[-1.8ex] 
 & DGP2(b) &  &  & DGP3\textcolor{white}{(a)}  &  &  \\ 
 & 0.1 & 0.3 & Weights & 0.1 & 0.3 & Weights \\ 
\hline \\[-1.8ex] 
Synthetic Learner & $0.030$ & $0.036$ & $$ & $0.080$ & $0.143$ & $$ \\ 
XGBoost & $0.020$ & $0.023$ & $0.132$ & $0.053$ & $0.056$ & $0.217$ \\ 
SVM & $0.030$ & $0.033$ & $0.831$  & $0.057$ & $0.050$ & $0.091$ \\ 
Arima & $0.020$ & $0.026$ & $0.037$ & $0.057$ & $0.106$ & $0.693$ \\ 
\hline \\[-1.8ex] 
\end{tabular} 
\end{table}

   \subsubsection{Endogenous time of intervention}
   
   Next, we study the problem in the presence of an endogenous time of the treatment. 
   In Figure \ref{fig:plotsim2d} we report a representative set of results under the endogenous time of the treatment, where we simulate
   $$
   T_0 = 280 + \min\{50, 1 + (\exp(1/\lambda) - 1/\lambda) \vee 1\}
   $$
where we choose $\lambda = |\sum_{j,t} X_{j,t}|$. The model follows a proportional hazard type model, similar to what is discussed in \cite{shaikh2019randomization}, centered on $T_0 = 281$, with the time of treatment depending on other units' outcomes and constrained between $281$ and $320$.  Figure \ref{fig:plotsim2d} collects results with an endogenous time of treatment when critical quantiles are estimated via resampling, where the confidence intervals for the Synthetic Control and DiD are constructed using the permutation-based method in \cite{chernozhukov2017exact}. Results are consistent with the case of an exogenous treatment timing.

%


\begin{figure}[!ht]
\spacingset{1}
\centering

\includegraphics[scale = 0.4]{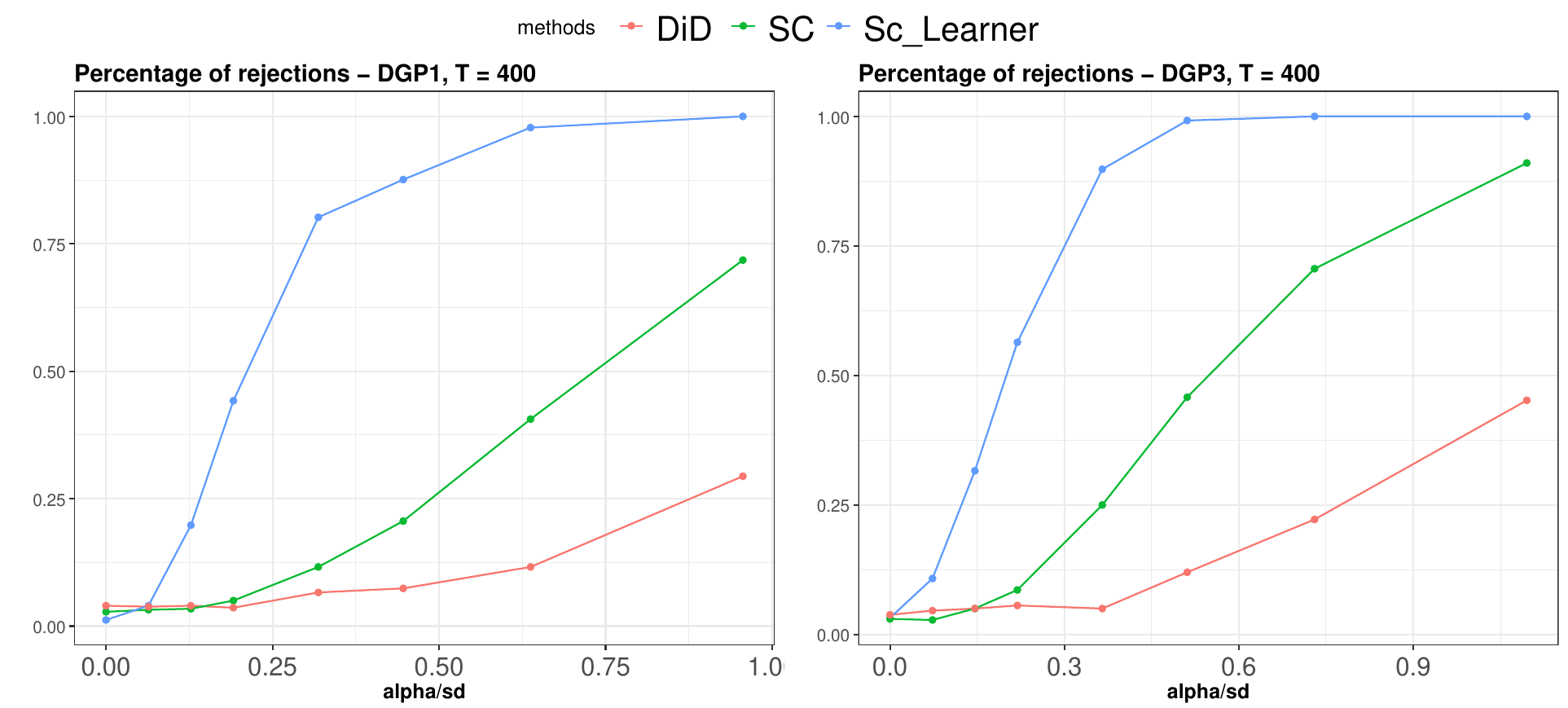}

 \caption{Percentage of rejections   over $500$ repetitions with $T=400$, $T_0=280$, $p=50$ and $T_-=140$ when the critical quantile is estimated via resampling and the time of treatment is \textit{endogenous}. The x-axis reports the policy's effect rescaled by the outcome's standard deviation.  
  } \label{fig:plotsim2d}
 \end{figure}

\subsubsection{Variability in the quality of the learners}

Next, we study the variability of the proposed method concerning the number and quality of learners included in learners' classes. We consider four different variations of the Synthetic Learner: Exponential and  Least Squares weighting with $10$ and $100$ new non-informative learners. To guarantee the feasibility of the optimization problem given a large number of learners, we consider a large $T = 300$. 
Figure 
  \ref{fig:plotsim3} contains the results. There we observe that many non-informative learners do little to nothing to the proposed Synthetic Learner. In sharp contrast, Least Squares' weighting suffers a substantial loss in power when the number of non-informative learners is increased.

\begin{figure} [!ht]
\spacingset{1}
\centering

\includegraphics[scale=0.5]{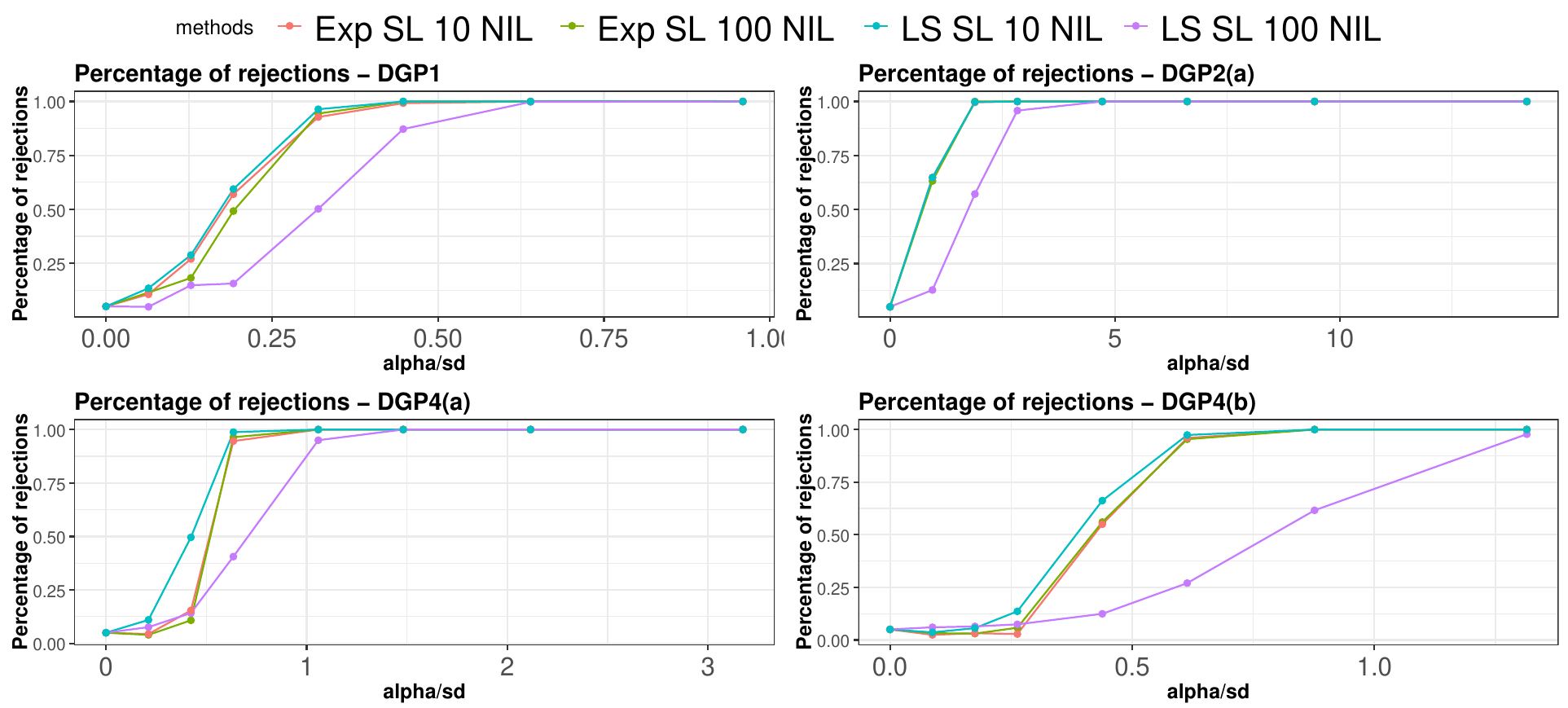}
		\caption{Percentage of rejections of null hypothesis $H_0$ over $500$ repetitions with $T=300$, $T_0=280$, $J=50$ and $T_-=140, p = 50$. The x-axis reports the policy's effect rescaled by the outcome's standard deviation. As base learner we consider XGboost, Support Vector Regression, {\sc arima}(0,1,1) and either $100$ or $10$ additional non informative learners(NIL). We denote exp SL as the Synthetic Learner using exponential weights and LS SL as the Synthetic Learner using Least Squares weights. 
   } \label{fig:plotsim3}
\end{figure}


\subsubsection{Bootstrap vs permutations}

We compare the performance of the circular bootstrap against permutations proposed in \cite{chernozhukov2017exact}. We consider only one learner: OLS. We compute the OLS coefficient for the bootstrap method using only the first $T_0/2$ observations, and we bootstrap the remaining ones. We estimate the full sample's coefficient for the permutation method after imposing the null hypothesis of no effect. We consider  the  true effect is either $\alpha_t = 0.2$ or $\alpha_t = 0.3$. We vary $T \in \{60, 80\}$ and we consider $T - T_0 = 10$.  Results are collected in Table \ref{tab:perm_vs_boot}.
For the non-linear design, we mostly observe significant improvements in power, up to approximately fifty percentage points. Only for a few designs we observe comparable performances or slightly inferior by one to three percentage points.

\begin{table}[!htbp] 
\spacingset{1}
\centering 
  \caption{We compare the percentage of rejection of the sharp null hypothesis $\alpha_t^o = 0$ over 300 replications. We use the bootstrap method, while rejections via permutations are presented in the parenthesis. We predict the counterfactuals only using Least Squares. $\alpha$ denote the true policy effect.  $T_0 = T - 10$. 
  } 
  \label{tab:perm_vs_boot} 
\begin{tabular}{@{\extracolsep{0pt}} ccc|cccc} 
\\[-1.8ex]\hline 
\hline \\[-1.8ex] 
    &  $T = 80$ & & $T = 60$ &  & \\ 
      &$\alpha=0.2$ & $\alpha=0.3$ & $\alpha=0.2$ & $\alpha=0.3$  \\ 
\hline \\[-1.8ex] 
DGP1    & $0.090(0.040)$ & $0.107(0.067)$ & $0.073(0.037)$ & $0.057(0.060)$  \\ 
DGP2(a)  & $0.697(0.507)$ & $0.940(0.793)$  & $0.420(0.317)$ & $0.747(0.667)$\\ 
DGP2(b)  & $0.030(0.067)$ & $0.040(0.087)$ & $0.040(0.047)$ & $0.030(0.070)$  \\ 
DGP2(c)  & $0.690(0.167)$ & $0.833(0.303)$ & $0.430(0.177)$ & $0.567(0.253)$ \\ 
DGP3    &$0.033(0.050)$ & $0.047(0.023)$ & $0.030(0.027)$ & $0.037(0.053)$ \\ 
DGP4(a) &  $0.143(0.087)$ & $0.300(0.280)$ & $0.113(0.107)$ & $0.213(0.203)$\\ 
DGP4(b) &  $0.070(0.040)$ & $0.090(0.077)$ & $0.043(0.053)$ & $0.063(0.037)$ \\ 
DGP4(c) &  $0.180(0.063)$ & $0.297(0.173)$ & $0.110(0.080)$ & $0.230(0.187)$ \\ 
DGP5 & $0.113(0.077)$ & $0.167(0.133)$ & $0.057(0.050)$ & $0.127(0.057)$ \\ 
DGP6 & $0.027(0.060)$ & $0.057(0.060)$ & $0.033(0.027)$ & $0.047(0.067)$ \\ 
\hline \\[-1.8ex] 
\end{tabular} 

\end{table}

\section{The Effect of Public Health Insurance Ineligibility on Access to Medical Care} \label{sec:real}

Understanding the effect of public health insurance coverage on health care access is a major concern in health economics \citep{kolstad2012impact, long2009another, baicker2013oregon, anderson2012effect, garthwaite2014public}. 

The TennCare dis-enrollment program represents the largest reduction in public health insurance coverage ever experienced in the US. Between 2005 and 2006, approximately 170,000 individuals lost public health insurance coverage. Most of these individuals were childless adults who gained public health insurance coverage approximately ten years before, in 1994, during the Medicaid program expansion in Tennessee. In this section, we study the effect of the reform over childless adults on delayed medical care access due to medical costs. This population is of particular interest since most of the Affordable Care Act expansions target childless adults. 
\cite{tello2016effects} estimates that the TennCare dis-enrollment significantly decreased the likelihood of having health insurance between 2 and 5 percent.  
The author estimates an increase of between two and three percentage points in the probability of not going to a medical center when sick.\footnote{The reader might refer to Panel C, Table 5 in \cite{tello2016effects}.}  
Our analysis provides supportive evidence for the claim, with positive effects ranging between one and five percentage points but with higher uncertainty in the absence of stationary. Details are discussed in the following lines.


\subsection{Data} 

We use BFRSS data\footnote{  Behavioral Risk Factor Surveillance System Data: \href{url}{https://www.cdc.gov/brfss/annual\_data/annual\_data.htm}.} to investigate the effect of the reform on the percentage of people who cannot afford healthcare expenses for medical costs. BFRSS is a national survey that has been continuously run over the years since 1984. The survey contains individual-specific information, including residence, state of health, access to health coverage, and others. The survey is run on a rolling basis, and the dataset can be organized as a long sequence of monthly and quarterly observations since we can cluster observations by the date of the interview. On average, we observe $150$ childless adults between $18$ and $64$ years old in Tennessee per month from $2017$ to $1993$. To overcome survey variability, we aggregate data at a \textit{quarter} level. The outcome variable is the percentage of childless adults who answered yes to the following survey question:``Was there a time in the past 12 months when you needed to see a doctor but could not because of the cost?''.

\begin{figure}[!ht]
\centering
\includegraphics[scale=0.5]{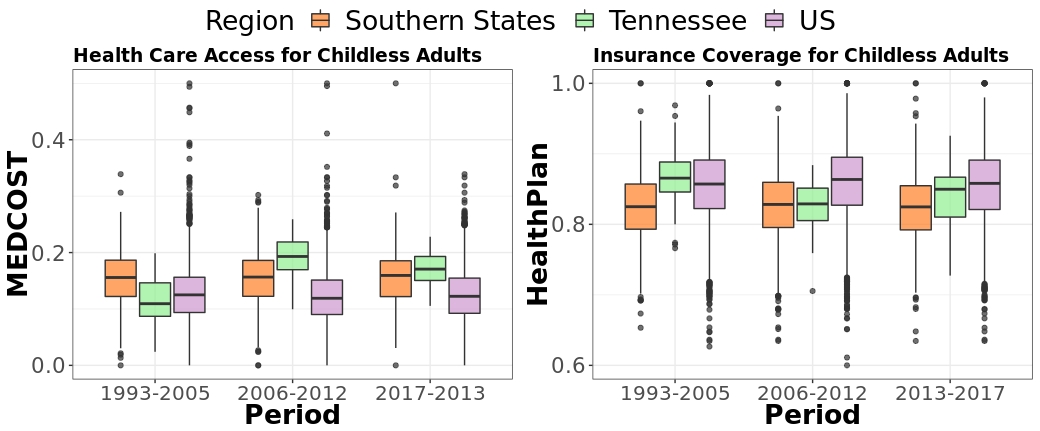}
\caption{Sample distribution of childless adults between $18$ and $64$ years old in Tennessee, the Southern States, and the US who were not able to afford health care expenses (left panel) and who are covered by health insurance (right panel). BFRSS data.
 }
\label{fig:boxplotstat}
\end{figure}

\begin{figure}[!ht]
\spacingset{1}
\centering
\includegraphics[height=5cm,width=12cm]{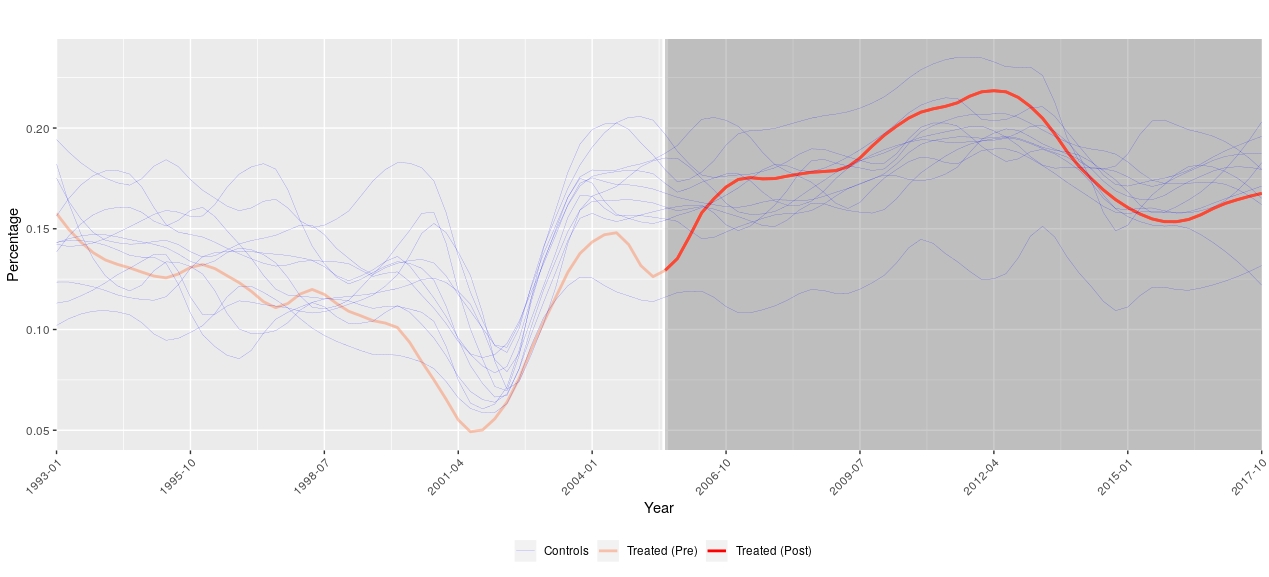}
\caption{BRFFS data. Percentage of adults in Tennessee (treated unit, red) and  Southern States (control, blue) who could not afford health care from January $1993$ to December $2017$. 
} \label{fig:appl1}
\end{figure}



In Figure \ref{fig:boxplotstat}, we report the distribution of respondents who were not able to afford medical costs in the past $12$ months(left panel) and who are covered by health insurance\footnote{For the latter questions, we count the number of individuals who answer yes to the question: ``Do you have any kind of health care coverage, including health insurance, prepaid plans such as HMOs, or government plans such as Medicare or Indian Health Service?'' We consider observations who answer ``I do not know'' as not having a plan.}(right panel), after clustering over the period $1993$-$2005$, $2006$-$2012$ and $2013$-$2017$ for Tennessee, other Southern States, and the United States. 
We observe a shift in the mean of Tennessee's outcome over these three periods, with a larger shift in the period just after the policy, between $2006$ and $2012$, while the variance remains approximately stable. 

To check for stationarity of observed time-series, we test for unit roots at $95\%$ confidence level. We reject the null hypothesis of a unit root in the time series of interest displayed in Figure \ref{fig:appl1}.\footnote{We use an Augmented Dickey-Fuller test, with constant and without time trend, and include one, two,  or three lags. P-values are respectively $<0.01$ for the first two tests and $0.05$ for the latter.} However, we warn the reader that a lack of stationarity or confounders may invalidate the analysis. To accommodate failures of stationarity in the presence of time-varying fixed effects, we consider two alternative estimators with and without fixed effects adjustments, as discussed in the following subsection.

\subsection{Estimation} 

As proposed in the Synthetic Control literature \citep{abadie2010synthetic}, we impute the potential outcome under no dis-enrollment using a set of control variables in the other states. 
The starting date of the treatment corresponds to the second half of $2005$.\footnote{The overall dis-enrollment started in July 2005, and it lasted until June 2006. Most childless adults who dis-enrolled during this period were not able to requalify for Medicare \citep{garthwaite2014public}.}

While the dis-enrollment program may be mostly attributable to an exogenous budget deficit \citep{argys2017losing}, the series may still be affected by confounding sources over the period study, one of which is Obamacare's launch in $2014$.\footnote{The Affordable Health Care Act, also known as Obama Care, was officially approved in $2010$, but the major change entered into force in $2014$.} 
To control for potential confounders related to Obamacare, we consider as the post-treatment period the series until $t = 2014$, while we replicate the analysis also including periods until $2017$ in Appendix \ref{sec:a_emp}. The selection into Obamacare from the states may reflect structural differences among different states. Motivated by this observation,
we use a pool of control units only those Southern States that, similarly to Tennessee, did not expand Medicaid between $2010$ and $2014$, namely South and North Carolina, Mississippi, Alabama, Florida, and Georgia. In the Appendix, we replicate the analysis with all the states.

We construct the ``Synthetic Control'' using the Synthetic Learner described in the current paper. We consider the share of individuals in other countries who could not afford necessary health care expenses as control variables. To allow for time-varying fixed effects, we consider two variations of the Synthetic Learner, with and without fixed effects adjustments (see Section \ref{sec:time_fixed}). We refer to the Synthetic Learner with fixed effect adjustment as ``Demeaned Synthetic Learner" and SL otherwise.

 We train Random Forest, Lasso, and a Factor model\footnote{To guarantee the validity of the algorithm through the sample splitting procedure, the factor model consists of estimating the principal component over the training period, regressing the principal component on the control states over the training period and making counterfactual predictions using the predicted factor on the remaining periods.}, and Diff-in-Diff mean proxy as discussed in \cite{doudchenko2016balancing} as base predictors.

Random Forest also contains additional covariates, such as the employment level in each state. Hyperparameters for the base predictors are chosen via cross-validation within the sample used to train such predictors. We construct weights with a tuning parameter $\eta = \frac{1}{\sqrt{T} \mathrm{Var}(Y_t)}$ where the rescaling by the variance guarantees that the estimated weights are scale-invariant.\footnote{Since the loss is the squared loss, by rescaling by the variance we have losses of the form $(Y_t - \hat{Y}_t)^2/\mathrm{Var}(Y_t) = (Y_t/\mathrm{SD}(Y_t) - \hat{Y}_t/\mathrm{SD}(Y_t))^2 $ which are unit free in the outcome's unit. We rescale by $1/\sqrt{T}$ following the theoretical results of predictions. In the Appendix, we report results also after choosing different rescaling. } In Appendix \ref{sec:a_emp}, we show the robustness of our results for several other choices of $\eta$. We consider two alternative sample splitting rules. First, observations between $1998$ and $2006$ are used to train the algorithms; observations between $1993$ and $1997$ are used to compute the weights and bootstrap. Second, the reverse is considered, where the training occurs over the earliest pre-treatment period. Observations from January $2006$ onwards are used to compute the test statistic.

\subsection{Results}

In Table \ref{table:effect1} we report the estimated test statistic for testing the null hypothesis of no effect, namely $H_0: Y_{0t}^1 - Y_{0t}^0 = 0, \  t > T_0$, over the post-treatment period 2010-2014.  The table reports the critical quantiles, the test statistic, and the estimated ATT when predictors are trained on the period closest to the post-treatment period (Period 1) or an earlier period (Period 2).
The ATT oscillates between one and four percentage points, and its sign remains robust throughout all the setups considered. 
Significant effects are detected when training predictors on Period 1 for the test with a size of ten percent for Synthetic Learner. 
 When estimating treatment effects by training the predictors more distant from the treatment timing (Period 2), results become non-significant, possibly reflecting higher uncertainty. Similarly, when we consider the adjustment for fixed effects, we observe p-values close to or larger than twenty percent, suggesting higher uncertainty for this case.

The Synthetic Learner predicts an effect larger than Lasso's one but smaller than a factor model. The reader may refer to the Appendix for further details.  In Appendix \ref{sec:a_emp}, we report results over the period 2010-2017, showing attenuated results over the time window 2010-2017.

\begin{table}[!htbp] \centering 
  \caption{ $90\%$ and $80\%$ critical values, t-statistic, and ATT using the southern states as controls. The effect estimated is over the time window 2010-2014 (first row), and consecutive windows 2011-2014 ($m = 1yr$), 2012-2014 ($m = 2yr$), 2013-2014 ($m = 3yr$).  Period 1 collects results when learners are estimated using the window between $1998$-$2006$, and weights are estimated over the period $1993$-$1997$. Period 2 corresponds to the opposite scenario. Demeaned SL denotes the SL with time-varying fixed effects.} 
  \label{table:effect1} 
\begin{tabular}{@{\extracolsep{5pt}} ccccc|cccc} 
\\[-1.8ex]\hline 
\hline \\[-1.8ex] 
Period 1 & SL &  & & & Demeaned SL &  &  &  \\ 
 & CV90 & CV80 & t stat & ATT & CV90 & CV80 & t stat & ATT \\ 
\hline \\[-1.8ex] 
m = 0 & $1.332$ & $1.249$ & $1.354$ & $1.844$ & $0.903$ & $0.846$ & $0.836$ & $0.994$ \\ 
m = 1yr & $1.261$ & $1.176$ & $1.281$ & $1.998$ & $0.794$ & $0.745$ & $0.729$ & $1.060$ \\ 
m = 2yr & $1.217$ & $1.133$ & $1.273$ & $2.166$ & $0.758$ & $0.711$ & $0.696$ & $1.104$ \\ 
m = 3yr & $1.160$ & $1.070$ & $1.229$ & $2.253$ & $0.714$ & $0.664$ & $0.651$ & $1.096$ \\ 
\hline \\[-1.8ex] 
Period 2 & SL &  & & & Demeaned SL &  &  &  \\ 
 & CV90 & CV80 & t stat & ATT & CV90 & CV80 & t stat & ATT \\ 
\hline \\[-1.8ex] 
m = 0 & $1.264$ & $1.173$ & $0.691$ & $5.223$ & $0.519$ & $0.471$ & $0.243$ & $3.137$ \\ 
m = 1yr & $1.211$ & $1.118$ & $0.622$ & $5.362$ & $0.460$ & $0.415$ & $0.163$ & $3.142$ \\ 
m = 2yr & $1.189$ & $1.100$ & $0.625$ & $5.516$ & $0.453$ & $0.407$ & $0.154$ & $3.202$ \\ 
m = 3yr & $1.154$ & $1.060$ & $0.611$ & $5.587$ & $0.450$ & $0.404$ & $0.149$ & $3.197$ \\ 
\hline \\[-1.8ex] 
\end{tabular} 
\end{table}

    In Table \ref{tab:weights} we collect the weights assigned to each base-algorithm. We observe that the Synthetic Learner assigns a larger weight to Lasso
in the absence of fixed effects, and a larger weight to Random Forest in the presence of fixed effects.

\begin{table}[!htbp] \centering 
  \caption{Weights estimated by the Synthetic Learner and by the Synthetic Learner after subtracting the control's mean to allow for fixed effects over Period 1.} 
  \label{tab:weights} 
\begin{tabular}{@{\extracolsep{5pt}} ccc} 
\\[-1.8ex]\hline 
\hline \\[-1.8ex] 
 & Synthetic Learner & Demeaned Synthetic Learner \\ 
\hline \\[-1.8ex] 
Factor & $0.242$ & $0.202$ \\ 
CV Lasso & $0.298$ & $0.253$ \\ 
Random Forest & $0.243$ & $0.295$ \\ 
DID & $0.217$ & $0.251$ \\ 
\hline \\[-1.8ex] 
\end{tabular} 
\end{table}

     In Figure \ref{fig:placebo}, we report the test statistics and the acceptance region for Tennessee and for placebo tests performed on the other Southern States that did not adopt Medicaid expansion. A placebo test consists of testing a policy's effect from $2006$ to $2014$ in a state different from Tennessee. Since none of the other Southern States had significant changes in the Medicaid system, we would expect no rejections for all Southern States except Tennessee. This is shown in Figure \ref{fig:placebo}. We observe that we do not reject the null hypothesis when using only the simple DiD method, potentially due to the underpowered test. This result is consistent with what we observed in simulations, where the synthetic learner outperformed other methods in terms of power.

\begin{figure}[!ht]
\centering
\includegraphics[scale=0.3]{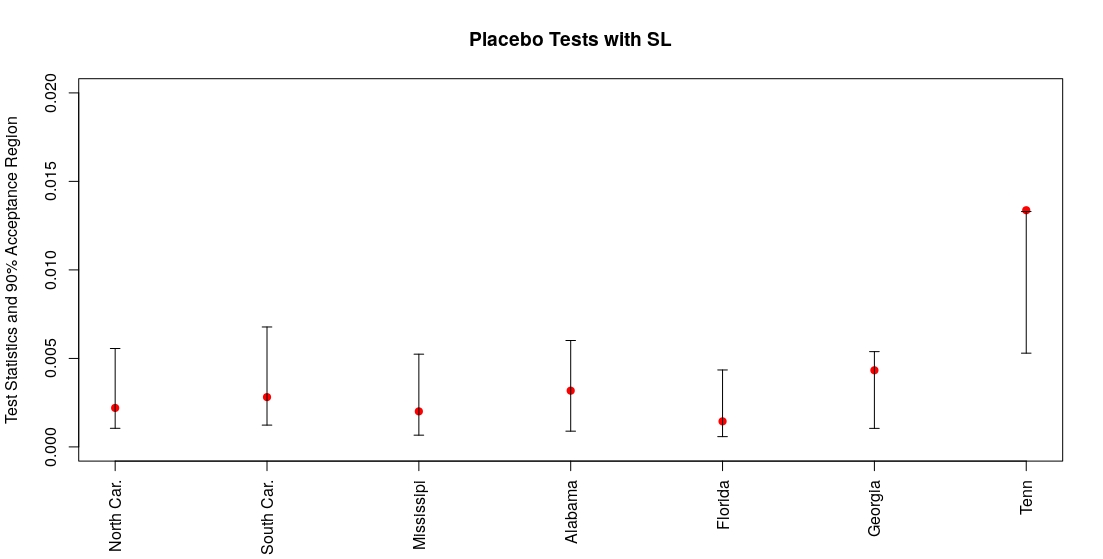}
\includegraphics[scale=0.3]{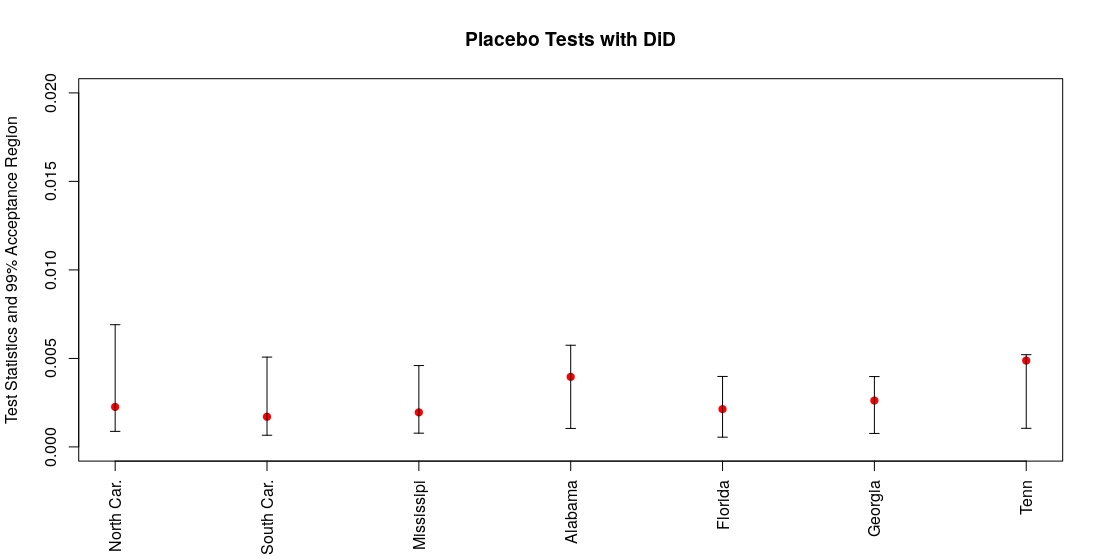}
\caption{We test significant changes in the percentage of childless adults who are not able to afford medical expenses in those Southern states that did not adopt Medicaid Obamacare. We report the test statistic (red dot) and $90\%$ confidence region for each of the states, including Tennessee, over Period 1. Left-panel reports the test when using Synthetic Learner. Right panel when only using Difference in Differences to predict the counterfactual. } \label{fig:placebo}
\end{figure}

\subsection{Discussion: Assumptions and Possible Sources of Confounding} 

We conclude this section with a discussion on the assumptions and possible sources of confounding. The Tennessee dis-enrollment program was a result of a budget deficit. Whenever the budget deficit was due to an exogenous variation \citep{argys2017losing}, the variation in state-level Medicaid expenses can be interpreted as exogenous to childless adults' financial status validating the exogeneity of the treatment timing. In this scenario, the ATT estimator is a consistent estimator of the underlying treatment effect. However, this condition may not necessarily hold. For instance, the budget deficit may be attributed to the increase in Medicaid expenses in previous years \citep{garthwaite2014public}, which may itself depend on individuals' past average income. In such a case, exogeneity may be replaced by conditional exogeneity given the past filtration without violating the prediction guarantees of the proposed algorithm. 

The results of our testing procedure should instead be interpreted as conditional on the treatment assignment mechanism, similar to what is discussed in \cite{ferman2016revisiting}. Here, the assumption of stationarity may fail if, for example, spillover effects of the dis-enrollment program occur over adjacent states, therefore changing the distribution of control units (see the discussion in Appendix \ref{sec:aa1}). While the study of Synthetic controls under spillovers goes beyond this paper's scope, we observe that in this scenario, it may be necessary to estimate counterfactuals also on the other states, which may be affected by the policy intervention. 

Confounding may also result from other events, such as the 1996 Clinton welfare reform and the great depression. In the former case, although the reform in Tennessee was targeted at families with children (so-called Family First policies\footnote{The reader may refer to \url{https://haslam.utk.edu/sites/default/files/ffoct00.pdf}.}), which are excluded from our analysis, it may act as a confounder in our analysis in the presence of general equilibrium effects. Time-fixed effects can (partially) accommodate for these cases \textit{if} the confounder affects the mean only, and the effect is homogenous across the states considered in our analysis. These restrictions motivate studies that focus on sets of control units most similar to the treated, which, in our case, correspond to a subset of Southern States.

 \section{Extension: Carry-over Effects} \label{sec:carryover}
 
 In many applications, treatment effects may carry over in time \citep{imai2021use}. Here we extend the proposed framework to carry-over effects as follows.  We consider binary treatment and  denote the treatment path up to time $t$ as a vector $\mathbf{d} _{1:t}\in \{0,1\}^{t}$.
 Following the potential outcomes framework, we then
posit the existence of potential outcomes $\tilde{Y}_{0t}(\mathbf{d}_{1:t})$, 
corresponding respectively to the
response the  treated subject would have experienced  at time $t$  while being exposed to the treatment  
assignment contained in the treatment path $\mathbf{d}_{1:t}$.
Formulating treatments and potential outcomes as paths were introduced initially by \cite{robins1986new}.

   Notation implicitly assumes no lead effects \citep{athey2018design}. Also, we require that the realizations of potential outcomes do not depend on past $m$ lags or more. Using the same notation as in \cite{rambachan2019nonparametric}, for   $t \in \{T_-,  \dots ,0, 1,  \dots , T_0, \dots, T\}$  we assume throughout the rest of this paper that the following holds.
\begin{ass}[Finite carry-over]  
For all $\mathbf{d}_{T_-:t}$, $\tilde{Y}_{0t}(\mathbf{d}_{T_-:t}) = Y_{0t}(\mathbf{d}_t, \cdots, \mathbf{d}_{t-m})$, for some function $Y_{0t}(\cdot)$
\end{ass} 
The assumption explicitely defines carry-over effects of size $m$. The no-anticipation assumption has been previously discussed in \cite{abbring2007econometric}, \cite{athey2018design}, while the restricted carryover effect is analogous to the identification assumption stated in \cite{imai2018matching}, \cite{bojinov2018time}, \cite{blackwell2018make} among others.  The estimand of interest is now defined as follows
$$
\mathbb{E}\Big[Y_{0t}(\mathbf{1}) - Y_{0t}(\mathbf{0})\Big], \quad t > T_0 + m
$$
   which denotes the (long-run) ATT, comparing two policies always and never implemented. 
  
   The key idea in this setting consists in estimating treatment effects after removing the $m$ lag components. Formally, we construct an ATT estimator of the form 
   \begin{equation} 
 \widehat{ATT}_m = {(T - T_0)}^{-1} \sum_{t > T_0 + m}  (Y_{0t}^1 -\hat Y_{0t} ^0 ) - (T_0/2)^{-1}\sum_{t = T_0/2 +1}^{T_0} (Y_{0t}^0 - \hat Y_{0t, -1} ^0 )
   \end{equation} 
   where the estimator averages after $m$ periods that the policy has been implemented. Our testing procedure remains invariant (and valid) after removing the periods $t \in \{T_0, \cdots, T_0 + m\}$.

 \begin{exmp}[Why considering carry-overs?]   Wrongly assuming the absence of carry-over effects can lead to misspecified causal estimands and hence possibly biased estimates.  
  For example, consider a simple case 
\begin{equation} \label{exmp:3}
Y_t(\mathbf{d}_{(t-m):t}) = Y_t(\mathbf{0}) + \sum_{s=0}^m \alpha_{s+1} d_{t-s} \quad \Rightarrow \quad  Y_t(\mathbf{1}) - Y_t(\mathbf{0}) = \sum_{s=0}^m \alpha_{s+1} d_{t-s}, 
\end{equation}
  for a sequence of constants  $\alpha_{s} \in \mathbb{R}$. The naive ATE estimate, defined as a difference between pre- and post-treatment  averages is possibly biased. In fact, it's mean  equal 
$
 |T - T_0|^{-1} {\sum_{s=0}^{m-1} (m - s) \alpha_{s+1}}{} + \sum_{s=0}^m \alpha_{s},  
$
where $ |T - T_0|^{-1} {\sum_{s=0}^{m-1} (m - s) \alpha_{s+1}}{}$ defines its bias. \qed 
\end{exmp}

\section{Discussion}

In this paper, we have introduced a novel strategy for estimating treatment effects and testing the null hypothesis of interest in the presence of time-dependent observations. We developed a novel algorithm, denoted as Synthetic Learner, that predicts the counterfactual building on multiple regression methods. Our framework provides a starting point for performing estimation and inference, which is valid regardless of the class of models under consideration. 

The presence of one single treated unit at a given time of adopting the policy brings substantial challenges from an identification perspective. We considered three scenarios of interest. First, (i) the adoption date is deterministic, $T_0$ and fixed treatment effects similarly to \cite{chernozhukov2017exact}, \cite{chernozhukov2018t}, \cite{arkhangelsky2019synthetic} among others.
We show that, under stationarity and mixing conditions, our algorithm controls the nominal size regardless of the class of base algorithm under consideration, even in the presence of misspecification bias. Extending this result to non-deterministic $T_0$ is conceptually feasible in the presence of multiple units treated at different points in time. 
(ii) We consider a random and exogenous time of treatment, and under stationarity assumptions, we show that the estimator for the average treatment effect is consistent under weak assumptions, letting $T_0$ be non-deterministic. Finally, (iii) we let the treatment time be \textit{sequentially} exogenous, without assuming any stationarity condition. We provide bounds on the predictive performance under this complex scenario. 
Our paper also opens new questions on constructing valid machine learning methods for causal inference when units exhibit dependence. We leave it to future research its study for inference on conditional average treatment effects under heterogeneous effects and endogenous treatment time.

\bibliography{mybibliography2}

\begin{thebibliography}{}

\bibitem[\protect\citeauthoryear{Abadie, Diamond, and Hainmueller}{Abadie
  et~al.}{2010}]{abadie2010synthetic}
Abadie, A., A.~Diamond, and J.~Hainmueller (2010).
\newblock Synthetic control methods for comparative case studies: Estimating
  the effect of california’s tobacco control program.
\newblock {\em Journal of the American statistical Association\/}~{\em
  105\/}(490), 493--505.

\bibitem[\protect\citeauthoryear{Abadie and Gardeazabal}{Abadie and
  Gardeazabal}{2003}]{abadie2003economic}
Abadie, A. and J.~Gardeazabal (2003).
\newblock The economic costs of conflict: A case study of the basque country.
\newblock {\em American economic review\/}~{\em 93\/}(1), 113--132.

\bibitem[\protect\citeauthoryear{Abbring and Heckman}{Abbring and
  Heckman}{2007}]{abbring2007econometric}
Abbring, J.~H. and J.~J. Heckman (2007).
\newblock Econometric evaluation of social programs, part iii: Distributional
  treatment effects, dynamic treatment effects, dynamic discrete choice, and
  general equilibrium policy evaluation.
\newblock {\em Handbook of econometrics\/}~{\em 6}, 5145--5303.

\bibitem[\protect\citeauthoryear{Amjad, Shah, and Shen}{Amjad
  et~al.}{2018}]{amjad2018robust}
Amjad, M., D.~Shah, and D.~Shen (2018).
\newblock Robust synthetic control.
\newblock {\em The Journal of Machine Learning Research\/}~{\em 19\/}(1),
  802--852.

\bibitem[\protect\citeauthoryear{Anderson, Dobkin, and Gross}{Anderson
  et~al.}{2012}]{anderson2012effect}
Anderson, M., C.~Dobkin, and T.~Gross (2012).
\newblock The effect of health insurance coverage on the use of medical
  services.
\newblock {\em American Economic Journal: Economic Policy\/}~{\em 4\/}(1),
  1--27.

\bibitem[\protect\citeauthoryear{Argys, Friedson, Pitts, and
  Tello-Trillo}{Argys et~al.}{2020}]{argys2017losing}
Argys, L.~M., A.~I. Friedson, M.~M. Pitts, and D.~S. Tello-Trillo (2020).
\newblock Losing public health insurance: Tenncare reform and personal
  financial distress.
\newblock {\em Journal of Public Economics\/}~{\em 187}, 104202.

\bibitem[\protect\citeauthoryear{Arkhangelsky, Athey, Hirshberg, Imbens, and
  Wager}{Arkhangelsky et~al.}{2021}]{arkhangelsky2019synthetic}
Arkhangelsky, D., S.~Athey, D.~A. Hirshberg, G.~W. Imbens, and S.~Wager (2021).
\newblock Synthetic difference-in-differences.
\newblock {\em American Economic Review\/}~{\em 111\/}(12), 4088--4118.

\bibitem[\protect\citeauthoryear{Athey, Bayati, Doudchenko, Imbens, and
  Khosravi}{Athey et~al.}{2021}]{athey2018matrix}
Athey, S., M.~Bayati, N.~Doudchenko, G.~Imbens, and K.~Khosravi (2021).
\newblock Matrix completion methods for causal panel data models.
\newblock {\em Journal of the American Statistical Association\/}~{\em
  116\/}(536), 1716--1730.

\bibitem[\protect\citeauthoryear{Athey, Bayati, Imbens, and Qu}{Athey
  et~al.}{2019}]{athey2019ensemble}
Athey, S., M.~Bayati, G.~Imbens, and Z.~Qu (2019).
\newblock Ensemble methods for causal effects in panel data settings.
\newblock In {\em AEA Papers and Proceedings}, Volume 109, pp.\  65--70.

\bibitem[\protect\citeauthoryear{Athey and Imbens}{Athey and
  Imbens}{2019}]{athey2019machine}
Athey, S. and G.~W. Imbens (2019).
\newblock Machine learning methods that economists should know about.
\newblock {\em Annual Review of Economics\/}~{\em 11}, 685--725.

\bibitem[\protect\citeauthoryear{Athey and Imbens}{Athey and
  Imbens}{2022}]{athey2018design}
Athey, S. and G.~W. Imbens (2022).
\newblock Design-based analysis in difference-in-differences settings with
  staggered adoption.
\newblock {\em Journal of Econometrics\/}~{\em 226\/}(1), 62--79.

\bibitem[\protect\citeauthoryear{Bai, Li, and Ouyang}{Bai
  et~al.}{2014}]{bai2014property}
Bai, C., Q.~Li, and M.~Ouyang (2014).
\newblock Property taxes and home prices: A tale of two cities.
\newblock {\em Journal of Econometrics\/}~{\em 180\/}(1), 1--15.

\bibitem[\protect\citeauthoryear{Bai}{Bai}{2009}]{bai2009panel}
Bai, J. (2009).
\newblock Panel data models with interactive fixed effects.
\newblock {\em Econometrica\/}~{\em 77\/}(4), 1229--1279.

\bibitem[\protect\citeauthoryear{Baicker, Taubman, Allen, Bernstein, Gruber,
  Newhouse, Schneider, Wright, Zaslavsky, and Finkelstein}{Baicker
  et~al.}{2013}]{baicker2013oregon}
Baicker, K., S.~L. Taubman, H.~L. Allen, M.~Bernstein, J.~H. Gruber, J.~P.
  Newhouse, E.~C. Schneider, B.~J. Wright, A.~M. Zaslavsky, and A.~N.
  Finkelstein (2013).
\newblock The oregon experiment—effects of medicaid on clinical outcomes.
\newblock {\em New England Journal of Medicine\/}~{\em 368\/}(18), 1713--1722.

\bibitem[\protect\citeauthoryear{Belloni, Chernozhukov, Fern{\'a}ndez-Val, and
  Hansen}{Belloni et~al.}{2017}]{belloni2017program}
Belloni, A., V.~Chernozhukov, I.~Fern{\'a}ndez-Val, and C.~Hansen (2017).
\newblock Program evaluation and causal inference with high-dimensional data.
\newblock {\em Econometrica\/}~{\em 85\/}(1), 233--298.

\bibitem[\protect\citeauthoryear{Ben-Michael, Feller, and
  Rothstein}{Ben-Michael et~al.}{2021}]{ben2018augmented}
Ben-Michael, E., A.~Feller, and J.~Rothstein (2021).
\newblock The augmented synthetic control method.
\newblock {\em Journal of the American Statistical Association\/}~{\em
  116\/}(536), 1789--1803.

\bibitem[\protect\citeauthoryear{Billmeier and Nannicini}{Billmeier and
  Nannicini}{2013}]{billmeier2013assessing}
Billmeier, A. and T.~Nannicini (2013).
\newblock Assessing economic liberalization episodes: A synthetic control
  approach.
\newblock {\em Review of Economics and Statistics\/}~{\em 95\/}(3), 983--1001.

\bibitem[\protect\citeauthoryear{Blackwell and Glynn}{Blackwell and
  Glynn}{2018}]{blackwell2018make}
Blackwell, M. and A.~N. Glynn (2018).
\newblock How to make causal inferences with time-series cross-sectional data
  under selection on observables.
\newblock {\em American Political Science Review\/}~{\em 112\/}(4),
  1067–1082.

\bibitem[\protect\citeauthoryear{Bojinov and Shephard}{Bojinov and
  Shephard}{2019}]{bojinov2019time}
Bojinov, I. and N.~Shephard (2019).
\newblock Time series experiments and causal estimands: exact randomization
  tests and trading.
\newblock {\em Journal of the American Statistical Association\/}~{\em
  114\/}(528), 1665--1682.

\bibitem[\protect\citeauthoryear{Bottmer, Imbens, Spiess, and Warnick}{Bottmer
  et~al.}{2021}]{bottmer2021design}
Bottmer, L., G.~Imbens, J.~Spiess, and M.~Warnick (2021).
\newblock A design-based perspective on synthetic control methods.
\newblock {\em arXiv preprint arXiv:2101.09398\/}.

\bibitem[\protect\citeauthoryear{Boucheron, Lugosi, and Massart}{Boucheron
  et~al.}{2013}]{boucheron2013concentration}
Boucheron, S., G.~Lugosi, and P.~Massart (2013).
\newblock {\em Concentration inequalities: A nonasymptotic theory of
  independence}.
\newblock Oxford university press.

\bibitem[\protect\citeauthoryear{Bradley et~al.}{Bradley
  et~al.}{2005}]{bradley2005basic}
Bradley, R.~C. et~al. (2005).
\newblock Basic properties of strong mixing conditions. a survey and some open
  questions.
\newblock {\em Probability surveys\/}~{\em 2}, 107--144.

\bibitem[\protect\citeauthoryear{Brodersen, Gallusser, Koehler, Remy, Scott,
  et~al.}{Brodersen et~al.}{2015}]{brodersen2015inferring}
Brodersen, K.~H., F.~Gallusser, J.~Koehler, N.~Remy, S.~L. Scott, et~al.
  (2015).
\newblock Inferring causal impact using bayesian structural time-series models.
\newblock {\em The Annals of Applied Statistics\/}~{\em 9\/}(1), 247--274.

\bibitem[\protect\citeauthoryear{Carrasco and Chen}{Carrasco and
  Chen}{2002}]{carrasco2002mixing}
Carrasco, M. and X.~Chen (2002).
\newblock Mixing and moment properties of various garch and stochastic
  volatility models.
\newblock {\em Econometric Theory\/}~{\em 18\/}(1), 17--39.

\bibitem[\protect\citeauthoryear{Carvalho, Masini, and Medeiros}{Carvalho
  et~al.}{2018}]{carvalho2018arco}
Carvalho, C., R.~Masini, and M.~C. Medeiros (2018).
\newblock Arco: an artificial counterfactual approach for high-dimensional
  panel time-series data.
\newblock {\em Journal of Econometrics\/}.

\bibitem[\protect\citeauthoryear{Cavallo, Galiani, Noy, and Pantano}{Cavallo
  et~al.}{2013}]{cavallo2013catastrophic}
Cavallo, E., S.~Galiani, I.~Noy, and J.~Pantano (2013).
\newblock Catastrophic natural disasters and economic growth.
\newblock {\em Review of Economics and Statistics\/}~{\em 95\/}(5), 1549--1561.

\bibitem[\protect\citeauthoryear{Cesa-Bianchi, Freund, Haussler, Helmbold,
  Schapire, and Warmuth}{Cesa-Bianchi et~al.}{1997}]{cesa1997use}
Cesa-Bianchi, N., Y.~Freund, D.~Haussler, D.~P. Helmbold, R.~E. Schapire, and
  M.~K. Warmuth (1997).
\newblock How to use expert advice.
\newblock {\em Journal of the ACM (JACM)\/}~{\em 44\/}(3), 427--485.

\bibitem[\protect\citeauthoryear{Cesa-Bianchi and Lugosi}{Cesa-Bianchi and
  Lugosi}{2006}]{cesa2006prediction}
Cesa-Bianchi, N. and G.~Lugosi (2006).
\newblock {\em Prediction, learning, and games}.
\newblock Cambridge university press.

\bibitem[\protect\citeauthoryear{Cesa-Bianchi, Lugosi, et~al.}{Cesa-Bianchi
  et~al.}{1999}]{cesa1999prediction}
Cesa-Bianchi, N., G.~Lugosi, et~al. (1999).
\newblock On prediction of individual sequences.
\newblock {\em The Annals of Statistics\/}~{\em 27\/}(6), 1865--1895.

\bibitem[\protect\citeauthoryear{Chernozhukov, Chetverikov, Demirer, Duflo,
  Hansen, Newey, and Robins}{Chernozhukov
  et~al.}{2018}]{chernozhukov2018double}
Chernozhukov, V., D.~Chetverikov, M.~Demirer, E.~Duflo, C.~Hansen, W.~Newey,
  and J.~Robins (2018).
\newblock Double/debiased machine learning for treatment and structural
  parameters.
\newblock {\em The Econometrics Journal\/}~{\em 21\/}(1), C1--C68.

\bibitem[\protect\citeauthoryear{Chernozhukov, W{\"u}thrich, and
  Yinchu}{Chernozhukov et~al.}{2018}]{chernozhukov2018exact}
Chernozhukov, V., K.~W{\"u}thrich, and Z.~Yinchu (2018).
\newblock Exact and robust conformal inference methods for predictive machine
  learning with dependent data.
\newblock In {\em Conference On Learning Theory}, pp.\  732--749. PMLR.

\bibitem[\protect\citeauthoryear{Chernozhukov, Wuthrich, and Zhu}{Chernozhukov
  et~al.}{2018}]{chernozhukov2018t}
Chernozhukov, V., K.~Wuthrich, and Y.~Zhu (2018).
\newblock A $ t $-test for synthetic controls.
\newblock {\em arXiv preprint arXiv:1812.10820\/}.

\bibitem[\protect\citeauthoryear{Chernozhukov, W{\"u}thrich, and
  Zhu}{Chernozhukov et~al.}{2021}]{chernozhukov2017exact}
Chernozhukov, V., K.~W{\"u}thrich, and Y.~Zhu (2021).
\newblock An exact and robust conformal inference method for counterfactual and
  synthetic controls.
\newblock {\em Journal of the American Statistical Association\/}~{\em
  116\/}(536), 1849--1864.

\bibitem[\protect\citeauthoryear{Doudchenko and Imbens}{Doudchenko and
  Imbens}{2016}]{doudchenko2016balancing}
Doudchenko, N. and G.~W. Imbens (2016).
\newblock Balancing, regression, difference-in-differences and synthetic
  control methods: A synthesis.
\newblock Technical report, National Bureau of Economic Research.

\bibitem[\protect\citeauthoryear{Elliott and Timmermann}{Elliott and
  Timmermann}{2004}]{elliott2004optimal}
Elliott, G. and A.~Timmermann (2004).
\newblock Optimal forecast combinations under general loss functions and
  forecast error distributions.
\newblock {\em Journal of Econometrics\/}~{\em 122\/}(1), 47--79.

\bibitem[\protect\citeauthoryear{Fang and Santos}{Fang and
  Santos}{2018}]{fang2014inference}
Fang, Z. and A.~Santos (2018).
\newblock Inference on directionally differentiable functions.
\newblock {\em The Review of Economic Studies\/}~{\em 86\/}(1), 377--412.

\bibitem[\protect\citeauthoryear{Ferman and Pinto}{Ferman and
  Pinto}{2016}]{ferman2016revisiting}
Ferman, B. and C.~Pinto (2016).
\newblock Revisiting the synthetic control estimator.

\bibitem[\protect\citeauthoryear{Firpo and Possebom}{Firpo and
  Possebom}{2018}]{firpo2018synthetic}
Firpo, S. and V.~Possebom (2018).
\newblock Synthetic control method: Inference, sensitivity analysis and
  confidence sets.
\newblock {\em Journal of Causal Inference\/}~{\em 6\/}(2).

\bibitem[\protect\citeauthoryear{Garthwaite, Gross, and Notowidigdo}{Garthwaite
  et~al.}{2014}]{garthwaite2014public}
Garthwaite, C., T.~Gross, and M.~J. Notowidigdo (2014).
\newblock Public health insurance, labor supply, and employment lock.
\newblock {\em The Quarterly Journal of Economics\/}~{\em 129\/}(2), 653--696.

\bibitem[\protect\citeauthoryear{Gunsilius}{Gunsilius}{2020}]{gunsilius2020distributional}
Gunsilius, F. (2020).
\newblock Distributional synthetic controls.
\newblock {\em arXiv preprint arXiv:2001.06118\/}.

\bibitem[\protect\citeauthoryear{Hazlett and Xu}{Hazlett and
  Xu}{2018}]{hazlett2018trajectory}
Hazlett, C. and Y.~Xu (2018).
\newblock Trajectory balancing: A general reweighting approach to causal
  inference with time-series cross-sectional data.
\newblock {\em https://ssrn.com/abstract=3214231\/}.

\bibitem[\protect\citeauthoryear{Hsiao, Steve~Ching, and Ki~Wan}{Hsiao
  et~al.}{2012}]{hsiao2012panel}
Hsiao, C., H.~Steve~Ching, and S.~Ki~Wan (2012).
\newblock A panel data approach for program evaluation: measuring the benefits
  of political and economic integration of hong kong with mainland china.
\newblock {\em Journal of Applied Econometrics\/}~{\em 27\/}(5), 705--740.

\bibitem[\protect\citeauthoryear{Hsiao and Zhou}{Hsiao and
  Zhou}{2019}]{hsiao2018panel}
Hsiao, C. and Q.~Zhou (2019).
\newblock Panel parametric, semiparametric, and nonparametric construction of
  counterfactuals.
\newblock {\em Journal of Applied Econometrics\/}~{\em 34\/}(4), 463--481.

\bibitem[\protect\citeauthoryear{Iavor~Bojinov}{Iavor~Bojinov}{2019}]{bojinov2018time}
Iavor~Bojinov, N.~S. (2019, Forthcoming).
\newblock Time series experiments and causal estimands: exact randomization
  tests and trading.
\newblock {\em Journal of the American Statistical Association\/}.

\bibitem[\protect\citeauthoryear{Imai and Kim}{Imai and
  Kim}{2021}]{imai2021use}
Imai, K. and I.~S. Kim (2021).
\newblock On the use of two-way fixed effects regression models for causal
  inference with panel data.
\newblock {\em Political Analysis\/}~{\em 29\/}(3), 405--415.

\bibitem[\protect\citeauthoryear{Imai, Kim, and Wang}{Imai
  et~al.}{2018}]{imai2018matching}
Imai, K., I.~S. Kim, and E.~Wang (2018).
\newblock Matching methods for causal inference with time-series cross-section
  data.
\newblock {\em https://imai.fas.harvard.edu/research/files/tscs.pdf\/}.

\bibitem[\protect\citeauthoryear{Imai, Ratkovic, et~al.}{Imai
  et~al.}{2013}]{imai2013estimating}
Imai, K., M.~Ratkovic, et~al. (2013).
\newblock Estimating treatment effect heterogeneity in randomized program
  evaluation.
\newblock {\em The Annals of Applied Statistics\/}~{\em 7\/}(1), 443--470.

\bibitem[\protect\citeauthoryear{Imbens and Rubin}{Imbens and
  Rubin}{2015}]{imbens2015causal}
Imbens, G.~W. and D.~B. Rubin (2015).
\newblock {\em Causal inference in statistics, social, and biomedical
  sciences}.
\newblock Cambridge University Press.

\bibitem[\protect\citeauthoryear{Kolstad and Kowalski}{Kolstad and
  Kowalski}{2012}]{kolstad2012impact}
Kolstad, J.~T. and A.~E. Kowalski (2012).
\newblock The impact of health care reform on hospital and preventive care:
  evidence from massachusetts.
\newblock {\em Journal of Public Economics\/}~{\em 96\/}(11-12), 909--929.

\bibitem[\protect\citeauthoryear{Kosorok}{Kosorok}{2008}]{kosorok2008introduction}
Kosorok, M.~R. (2008).
\newblock {\em Introduction to empirical processes and semiparametric
  inference.}
\newblock Springer.

\bibitem[\protect\citeauthoryear{K{\"u}nzel, Sekhon, Bickel, and Yu}{K{\"u}nzel
  et~al.}{2019}]{kunzel2017meta}
K{\"u}nzel, S.~R., J.~S. Sekhon, P.~J. Bickel, and B.~Yu (2019).
\newblock Metalearners for estimating heterogeneous treatment effects using
  machine learning.
\newblock {\em Proceedings of the National Academy of Sciences\/}~{\em
  116\/}(10), 4156--4165.

\bibitem[\protect\citeauthoryear{Lange, Rahbek, and Jensen}{Lange
  et~al.}{2011}]{lange2011estimation}
Lange, T., A.~Rahbek, and S.~T. Jensen (2011).
\newblock Estimation and asymptotic inference in the ar-arch model.
\newblock {\em Econometric Reviews\/}~{\em 30\/}(2), 129--153.

\bibitem[\protect\citeauthoryear{Lee}{Lee}{2005}]{lee2005probabilistic}
Lee, O. (2005).
\newblock Probabilistic properties of a nonlinear arma process with markov
  switching.
\newblock {\em Communications in Statistics-Theory and Methods\/}~{\em
  34\/}(1), 193--204.

\bibitem[\protect\citeauthoryear{Li}{Li}{2017}]{li2017estimating}
Li, K.~T. (2017).
\newblock Estimating average treatment effects using a modified synthetic
  control method: Theory and applications.
\newblock {\em The Wharton School, the University of Pennsylvania\/}.

\bibitem[\protect\citeauthoryear{Li}{Li}{2020}]{li2019statistical}
Li, K.~T. (2020).
\newblock Statistical inference for average treatment effects estimated by
  synthetic control methods.
\newblock {\em Journal of the American Statistical Association\/}~{\em
  115\/}(532), 2068--2083.

\bibitem[\protect\citeauthoryear{Li and Bell}{Li and
  Bell}{2017}]{li2017estimation}
Li, K.~T. and D.~R. Bell (2017).
\newblock Estimation of average treatment effects with panel data: Asymptotic
  theory and implementation.
\newblock {\em Journal of econometrics\/}~{\em 197\/}(1), 65--75.

\bibitem[\protect\citeauthoryear{Long, Stockley, and Yemane}{Long
  et~al.}{2009}]{long2009another}
Long, S.~K., K.~Stockley, and A.~Yemane (2009).
\newblock Another look at the impacts of health reform in massachusetts:
  evidence using new data and a stronger model.
\newblock {\em American Economic Review\/}~{\em 99\/}(2), 508--11.

\bibitem[\protect\citeauthoryear{Lunde and Shalizi}{Lunde and
  Shalizi}{2017}]{lunde2017bootstrapping}
Lunde, R. and C.~R. Shalizi (2017).
\newblock Bootstrapping generalization error bounds for time series.
\newblock {\em arXiv preprint arXiv:1711.02834\/}.

\bibitem[\protect\citeauthoryear{Maclean, Tello-Trillo, and Webber}{Maclean
  et~al.}{2019}]{maclean2019losing}
Maclean, J.~C., S.~Tello-Trillo, and D.~Webber (2019).
\newblock Losing insurance and behavioral health inpatient care: Evidence from
  a large-scale medicaid disenrollment.
\newblock Technical report, National Bureau of Economic Research.

\bibitem[\protect\citeauthoryear{Pham and Tran}{Pham and
  Tran}{1985}]{pham1985some}
Pham, T.~D. and L.~T. Tran (1985).
\newblock Some mixing properties of time series models.
\newblock {\em Stochastic processes and their applications\/}~{\em 19\/}(2),
  297--303.

\bibitem[\protect\citeauthoryear{Politis and Romano}{Politis and
  Romano}{1992}]{politis1992circular}
Politis, D.~N. and J.~P. Romano (1992).
\newblock A circular block-resampling procedure for stationary data.
\newblock {\em Exploring the limits of bootstrap\/}~{\em 2635270}.

\bibitem[\protect\citeauthoryear{Politis and Romano}{Politis and
  Romano}{1994}]{politis1994stationary}
Politis, D.~N. and J.~P. Romano (1994).
\newblock The stationary bootstrap.
\newblock {\em Journal of the American Statistical association\/}~{\em
  89\/}(428), 1303--1313.

\bibitem[\protect\citeauthoryear{Polley and Van Der~Laan}{Polley and Van
  Der~Laan}{2010}]{polley2010super}
Polley, E.~C. and M.~J. Van Der~Laan (2010).
\newblock Super learner in prediction.

\bibitem[\protect\citeauthoryear{Potrafke and Wuthrich}{Potrafke and
  Wuthrich}{2020}]{potrafke2020green}
Potrafke, N. and K.~Wuthrich (2020).
\newblock Green governments.
\newblock {\em arXiv preprint arXiv:2012.09906\/}.

\bibitem[\protect\citeauthoryear{Rambachan and Shephard}{Rambachan and
  Shephard}{2019}]{rambachan2019nonparametric}
Rambachan, A. and N.~Shephard (2019).
\newblock A nonparametric dynamic causal model for macroeconometrics.
\newblock {\em arXiv preprint arXiv:1903.01637\/}.

\bibitem[\protect\citeauthoryear{Rigollet, Tsybakov, et~al.}{Rigollet
  et~al.}{2012}]{rigollet2012sparse}
Rigollet, P., A.~B. Tsybakov, et~al. (2012).
\newblock Sparse estimation by exponential weighting.
\newblock {\em Statistical Science\/}~{\em 27\/}(4), 558--575.

\bibitem[\protect\citeauthoryear{Rinaldo, Wasserman, and G’Sell}{Rinaldo
  et~al.}{2019}]{rinaldo2019bootstrapping}
Rinaldo, A., L.~Wasserman, and M.~G’Sell (2019).
\newblock Bootstrapping and sample splitting for high-dimensional,
  assumption-lean inference.
\newblock {\em The Annals of Statistics\/}~{\em 47\/}(6), 3438--3469.

\bibitem[\protect\citeauthoryear{Robins}{Robins}{1986}]{robins1986new}
Robins, J. (1986).
\newblock A new approach to causal inference in mortality studies with a
  sustained exposure period—application to control of the healthy worker
  survivor effect.
\newblock {\em Mathematical modelling\/}~{\em 7\/}(9-12), 1393--1512.

\bibitem[\protect\citeauthoryear{Rubin}{Rubin}{1990}]{rubin1990formal}
Rubin, D.~B. (1990).
\newblock Formal mode of statistical inference for causal effects.
\newblock {\em Journal of statistical planning and inference\/}~{\em 25\/}(3),
  279--292.

\bibitem[\protect\citeauthoryear{Schapire and Freund}{Schapire and
  Freund}{2012}]{schapire2012boosting}
Schapire, R.~E. and Y.~Freund (2012).
\newblock {\em Boosting: Foundations and algorithms}.
\newblock MIT press.

\bibitem[\protect\citeauthoryear{Shaikh and Toulis}{Shaikh and
  Toulis}{2019}]{shaikh2019randomization}
Shaikh, A. and P.~Toulis (2019).
\newblock Randomization tests in observational studies with staggered adoption
  of treatment.
\newblock {\em University of Chicago, Becker Friedman Institute for Economics
  Working Paper\/}~(2019-144).

\bibitem[\protect\citeauthoryear{Tello-Trillo}{Tello-Trillo}{2021}]{tello2016effects}
Tello-Trillo, D.~S. (2021).
\newblock Effects of losing public health insurance on preventative care,
  health, and emergency department use: Evidence from the tenncare
  disenrollment.
\newblock {\em Southern Economic Journal\/}~{\em 88\/}(1), 322--366.

\bibitem[\protect\citeauthoryear{Timmermann}{Timmermann}{2006}]{timmermann2006forecast}
Timmermann, A. (2006).
\newblock Forecast combinations.
\newblock {\em Handbook of economic forecasting\/}~{\em 1}, 135--196.

\bibitem[\protect\citeauthoryear{Van~der Vaart}{Van~der
  Vaart}{2000}]{van2000asymptotic}
Van~der Vaart, A.~W. (2000).
\newblock {\em Asymptotic statistics}, Volume~3.
\newblock Cambridge university press.

\bibitem[\protect\citeauthoryear{Van Der~Vaart and Wellner}{Van Der~Vaart and
  Wellner}{1996}]{van1996weak}
Van Der~Vaart, A.~W. and J.~A. Wellner (1996).
\newblock Weak convergence.
\newblock In {\em Weak convergence and empirical processes}, pp.\  16--28.
  Springer.

\bibitem[\protect\citeauthoryear{Varian}{Varian}{2016}]{varian2016causal}
Varian, H.~R. (2016).
\newblock Causal inference in economics and marketing.
\newblock {\em Proceedings of the National Academy of Sciences\/}~{\em
  113\/}(27), 7310--7315.

\bibitem[\protect\citeauthoryear{White}{White}{2014}]{white2014asymptotic}
White, H. (2014).
\newblock {\em Asymptotic theory for econometricians}.
\newblock Academic press.

\bibitem[\protect\citeauthoryear{Xie and Huang}{Xie and
  Huang}{2014}]{xie2014impact}
Xie, S. and J.~Huang (2014).
\newblock The impact of index futures on spot market volatility in china.
\newblock {\em Emerging Markets Finance and Trade\/}~{\em 50\/}(sup1),
  167--177.

\bibitem[\protect\citeauthoryear{Xu}{Xu}{2017}]{xu2017generalized}
Xu, Y. (2017).
\newblock Generalized synthetic control method: Causal inference with
  interactive fixed effects models.
\newblock {\em Political Analysis\/}~{\em 25\/}(1), 57--76.

\end{thebibliography}
\bibliographystyle{chicago}

\newpage 

 \appendix

\numberwithin{equation}{section}
\makeatletter 
\newcommand{\section@cntformat}{Appendix \thesection:\ }
\makeatother
 
 \numberwithin{figure}{section}
\numberwithin{algorithm}{section}
 \numberwithin{table}{section}
\makeatletter

\setcounter{page}{1}

\begin{center}
    \LARGE Appendix to ``Synthetic Learner: Model Free Inference on Treatments over Time''
\end{center}

\section{Proof of the Lemmas in the Main Text} \label{app:lemmas}

\subsection{Proof of Lemma \ref{lem:identi2} and \ref{lem:identi2b} } 

We prove Lemma \ref{lem:identi2b}, while Lemma \ref{lem:identi2} directly follows from stationarity.  
Note first that under the null hypothesis $Y_{0t}^o = Y_{0t}^0$. As a result, we can write (under the null) $(Y_{0t}^o - \bar{Y}_t, Y_{1t}^1 - \bar{Y}_t, \cdots, Y_{nt} - \bar{Y}_t, Z_{0:n,t}) = (Y_{0t}^0 - \bar{Y}_t, Y_{1t}^1 - \bar{Y}_t, \cdots, Y_{nt} - \bar{Y}_t, Z_{0:n,t})$. Under Assumption \ref{ass:stationarity}, we have 
$$
(Y_{0t}^0 - \bar{Y}_t, Y_{1t}^1 - \bar{Y}_t, \cdots, Y_{nt} - \bar{Y}_t, Z_{0:n,t}) = (\varepsilon_{0t}^0 - \bar{\varepsilon}_{t} + \iota_0^0 - \bar{\iota}^0, \cdots,\varepsilon_{nt}^0 - \bar{\varepsilon}_{t} + \iota_n^0 - \bar{\iota}^0, Z_{0:n,t})
$$ 
where $\bar{\varepsilon}_t = \frac{1}{n} \sum_{j=1}^J \varepsilon_{jt}, \bar{\iota}^0 = \frac{1}{n} \sum_{j=1}^J \iota_{j}^0$. Stationarity directly follows from Assumption \ref{ass:stationarityb}.  

\subsection{Proof of Lemma \ref{lem:1} and \ref{lem:1b}}

 We prove Lemma \ref{lem:1b}. Lemma \ref{lem:1} follows directly from Assumption \ref{ass:stationarity} and exogeneity of $T_0$ (Assumption \ref{ass:ident}).  
First, note that under Assumption \ref{ass:ident} the distribution of potential outcomes and covariates remains invariant conditionally or unconditionally on $T_0$. 
Observe that under Assumption \ref{ass:stationarityb}, we have 
$$
\tau = \frac{1}{T - T_0} \sum_{t > T_0} \mathbb{E}\Big[Y_{0t}^1\Big] - \kappa_t - \iota_j^0 - \mathbb{E}\Big[\varepsilon_{jt}^0\Big] = \frac{1}{T - T_0} \sum_{t > T_0} \mathbb{E}\Big[Y_{0t}^1\Big] - \kappa_t - \iota_j^0. 
$$ 
where, by stationarity and the presence of the unit fixed effect $\iota_j^0$, $\mathbb{E}\Big[\varepsilon_{jt}^0\Big] = 0$, since the expectation is incorporated in the component $\iota_j^0$. 
Observe that we can write 
$$
\frac{1}{n} \sum_{j > 0} \mathbb{E}\Big[Y_{jt}^0 \Big | T_0\Big] = 
\frac{1}{n} \sum_{j > 0} \mathbb{E}\Big[Y_{jt}^0  \Big] = \kappa_t^0 + \frac{1}{n} \sum_{j > 0} \iota_j^0, 
$$ 
where, again, $\mathbb{E}[\varepsilon_{jt}^0] = 0$. Under Assumption \ref{ass:stationarityb} 
$$
\frac{1}{T_0} \sum_{s = 1}^{T_0} \Big\{\mathbb{E}\Big[Y_{0s}^0|T_0\Big]  - \frac{1}{n} \sum_{j > 0} \mathbb{E}\Big[Y_{js}^0\Big|T_0\Big] \Big\} = \iota_0^0 - \frac{1}{n} \sum_{j > 0} \iota_j^0. 
$$ 
Therefore, 
$$
-\frac{1}{T - T_0} \sum_{t > T_0} \frac{1}{n} \sum_{j > 0} \mathbb{E}\Big[Y_{jt}^0 \Big | T_0\Big] - \frac{1}{T_0} \sum_{s = 1}^{T_0} \Big\{\mathbb{E}\Big[Y_{0s}^0|T_0\Big]  - \frac{1}{n} \sum_{j > 0} \mathbb{E}\Big[Y_{js}^0\Big|T_0\Big] \Big\} = - \frac{1}{T - T_0} \sum_{t > T_0} \kappa_t^0 - \iota_0^0
$$ 
which completes the proof of Lemma \ref{lem:1b}.

\section{Theorem \ref{thm:1}} \label{sec:proof} 

\subsection{Definitions} \label{sec:definitions}
 \begin{defn}
    ($\beta$-mixing)\textit{Let Y be a stochastic process and $(\Omega, \mathcal{F},Y_{\infty})$ be the probability space. The $\beta$-mixing coefficient $\beta_Y(h)$ is given by
        $$
        \beta_Y(h) = \sup_t ||\mathcal{P}_{-\infty:t}\otimes \mathcal{P}_{(t + h):\infty} - \mathcal{P}_{-\infty:t}\mathcal{P}_{(t + h):\infty}||_{TV}
        $$
        where $||.||_{TV}$ is the total variation norm, $\mathcal{P}_{-\infty:t} \otimes \mathcal{P}_{(t+h):\infty}$ is the joint distribution and $\mathcal{P}_{-\infty:t}\mathcal{P}_{(t+h):\infty}$ is the product measure. The process is $\beta$-mixing if $\beta_Y(h) \rightarrow 0$ as $h \rightarrow \infty$} 
\end{defn} 

In this section we provide a set of definitions before discussing the main theorem. For a set $\mathbb{A}$ we denote the space of bounded functions on $\mathbb{A}$ by
$$
l^{\infty}(\mathbb{A}, \mathbb{B}) = \{f: \mathbb{A} \rightarrow \mathbb{B}, \text{ such that } \|f\|_{\infty} < \infty \}, \quad 
$$
where $\|f\|_{\infty} = \sup_{a \in \mathbb{A}}|f(a)|$.  We let $\mathcal{C}[0, 1]$ to be the space of cad-lag functions,   
 i.e. right-continuous with left-hand limits equipped with Skorokhod metric. We define the parameter of interest as a function of the joint law of the data. More precisely, for any law $\mathcal{P}$, for some parameter of interest $\theta$, we can define $\mathcal{P} \mapsto \theta(\mathcal{P})$ to be a measurable map from a domain $\mathcal{C}[0,1]$ to $\Theta$. For a metric space $\mathbb{A}$ with norm $||.||_{\mathbb{A}}$ we denote the set of Lipschitz functionals whose level and Lipschitz constant are bounded by one by 
$$
BL_1(\mathbb{A}) = \{f: \mathbb{A} \rightarrow \mathbb{\mathbb{R}}: |f(a)| \le 1\text{ and } |f(a) - f(a')| \le ||a - a'||_{\mathbb{A}}\text{ for all } a,a' \in \mathbb{A}\}.
$$
The definition above helps us discussing the definition of weak convergence, provided below. \begin{defn}(Weak Convergence) We say that $\mathbb{X}_n$ converges weakly in probability conditional on the data to $\mathbb{X}$, or $\mathbb{X}_n \rightsquigarrow \mathbb{X}$ if 
	$$
	\sup_{f \in BL_1} \Big |\mathbb{E}[f(\mathbb{X}_n)] - \mathbb{E}[f(\mathbb{X})] \Big| \rightarrow 0.
	$$
\end{defn} \noindent Consider the generic problem of studying the limiting distribution of 
$$r_n(\phi(\mathbb{X}_n) - \phi(\mathbb{X}))$$ for some $\phi: \mathbb{D}_{\phi} \subset \mathbb{D} \rightarrow \mathbb{E}$. The asymptotic distribution of interest can be derived whenever $\phi$ satisfies some differentiability requirements such that 
$$r_n \{\phi(\mathbb{X}_n) - \phi(\mathbb{X})\} = \phi'_{\theta_0}(r_n(\mathbb{X}_n - \mathbb{X})) + o_\mathbb{P}(1).$$
 The main condition on $\phi$ is that is it satisfies a notion of differentiability denoted as Hadamard differentiability. The definition is provided below. 
\begin{defn}(Hadamard Differentiable Map) \textit{Let $\mathbb{D}$ and $\mathbb{E}$ be Banach spaces with norms $||.||_{\mathbb{D}}$ and $||.||_{\mathbb{E}}$ respectively, and $\phi: \mathbb{D}_\phi \subseteq \mathbb{D} \rightarrow \mathbb{E}$. The map $\phi$, is Hadamard differentiable at $\theta \in \mathbb{D}_\phi$ tangentially to a set $\mathbb{D}_0 \subset \mathbb{D}$ if there exist a continuous linear map $\phi'_{\theta}:\mathbb{D}_0 \rightarrow \mathbb{E}$ such that
		$$
		\lim_{n\rightarrow \infty}  \left \| \frac{\phi(\theta + t_n k_n) - \phi(\theta)}{t_n} - \phi'_{\theta}(k) \right \|_{\mathbb{E}} = 0,
		$$
		$\forall$ converging sequences $t_n \rightarrow 0$, $\{t_n\}\subset \mathbb{R}$ and $k_n \in \mathbb{D}$, $k_n \rightarrow k \in \mathbb{D}_0$ as $n \rightarrow \infty$ and $\theta  + t_n k_n \in \mathbb{D}_\phi$ for all $n \ge 1$ sufficiently large.} \end{defn}  
\noindent It can be shown that Hadamard differentiability is equivalent to the difference in the previous expression in tending to zero uniformly on $k$ in compact subsets of $\mathbb{D}$ \citep{van2000asymptotic}. We move to define the parameters of interest as functionals that map  a space of bounded functions to a Banach space. We do this in the following definition.

\begin{defn} Let the weights be $w(\cdot) : \mathcal{A} \subset l^{\infty}(\mathbb{R}^p, \mathbb{R}) \rightarrow \mathcal{W}$, where $\mathcal{W} = [0,1]^p$. \end{defn} 

 
 As a second step we need to define the tests as a function of the parameter of interest and of a vector $z \in \mathbb{R}^{p+1}$. \begin{defn}We define $S_0(z, w) \in l^{\infty}(\mathbb{R}^{p+1} \times \mathcal{W}, \mathbb{R})$ where $z \in \mathbb{R}^{p+1}$ and $w \in \mathcal{W}$, with $S_0(z,w) = |z_1 - z_2w|^2$ where $z_1$ is the first entry of $z$ and $z_2$ all the remaining entries.  \end{defn} 

\noindent With an abuse of notation, we define  
$$Z_t = (Y_{0t}^o, g(X_t)),$$
the vector of observations after imposing the null hypothesis and of transformations of the learners, and refer to $Y_t$ as $Y_{0t}$ throughout our discussion. We will consider $g(\cdot)$ to be a fixed function independent of $Z_t$ for the moment (hence $Z_t$ stationary under Assumption \ref{ass:stationarity}), and return to $g(\cdot)$ being estimated with past data at the end of the proof (Appendix \ref{sec:final}).

We  denote the empirical measures for control   and treatment period, respectively,
$$\mathcal{P}_{T_0} = \frac{1}{T_0} \sum_{t=1}^{T_0} \delta_{Z_t}, \qquad \mathcal{P}_{T_1} = \frac{1}{T - T_0} \sum_{t > T_0} \delta_{Z_t}, \quad \mathcal{P}_T = \frac{1}{T} \sum_{t  = 1}^T \delta_{Z_t}.$$ Similarly $\mathcal{P}_T^*, \mathcal{P}_{T_0}^*$ denote the bootstrapped counterparts constructed based on the entire sample $\mathcal{P}_T$. 

For $z \in \mathbb{R}^{p+1}$ let the operator $\le$ be the component wise operator. 
We now let $\mathcal{F}$ be the function class:
$$
\mathcal{F} = \{f_s: s \in \mathbb{R}^{p+1}, f_s(z) = 1_{z \le s}\}.
$$ 
It follows that the empirical distribution function for the control and treatment period can be expressed point wise as 
$$F_{T_0}(s) = \mathcal{P}_{T_0}f_s, \quad F_{T_1}(s) = \mathcal{P}_{T_1} f_s$$ respectively and similarly this hold for bootstrap measures. Notice that we can see $F_{T_0}(\cdot)$ and $F_{T_1}(\cdot)$ as elements of $l^{\infty}(\Omega \times \mathcal{F}, \mathbb{R})$ and similarly $F_{T_0}^*$, $F_T^*$ as element of $l^{\infty}(\Omega \times \bar{\Omega} \times \mathcal{F}, \mathbb{R})$ where $\bar{\Omega}$ is the probability space associated with bootstrap weights. For a fixed sample path we can view this mappings as belonging to $l^{\infty}(\mathcal{F}, \mathbb{R})$. We define
\begin{equation} \label{eqn:process}
\mathbb{H}_T(f_s) = 
\sqrt{T_0}\left [
\begin{tabular}{c}
\text{$\mathcal{P}_{T_0} f_s - \mathcal{P} f_s$} \\
\text{$\mathcal{P}_{T_1} f_s - \mathcal{P} f_s$}
\end{tabular}
\right ], \quad \mathbb{H}_T^*(f_s) = 
\sqrt{T_0}\left [
\begin{tabular}{c}
\text{$\mathcal{P}_{T_0}^* f_s - \mathcal{P}_{T} f_s$} \\
\text{$\mathcal{P}_{T_1}^* f_s - \mathcal{P}_T f_s$}
\end{tabular}
\right ].
\end{equation}
Here $\mathcal{P}$ denotes the distribution of $(Y_t^o, g(X_t))$ for fixed $g(\cdot)$ (i.e., the distribution induced by applying a function $[i(\cdot), g(\cdot)]$ to the vector $(Y_t^o, X_t)$, with $g(\cdot)$ evaluated at $T_-$ independent data points from $(Y_t^o, X_t)_{t=1}^T$, and conditional on such data points, and $i(\cdot)$ being the identity function).  

We express the test statistics of interest  as functionals of $\mathbb{H}_{T}$ and a  fixed null trajectory $\{a_t^o\}$.  
That is, for Equation \eqref{eqn:teststat1} (divided by $(T_0 - T)^{1/2}$) we can define $(A,B) \rightarrow T(A,B)$ with 
$$T(A,B) = \int S_0(z, w(A)) dB.$$

\subsection{Auxiliary Lemmas} 

\begin{lem}[Functional Delta Method, \cite{van1996weak} Theorem 3.9.4] \label{lem:delta_method} Let $\mathbb{D}$ and $\mathbb{E}$ be metrizable topological vector spaces. Let $\phi: \mathbb{D}_\phi \subset \mathbb{D} \mapsto \mathbb{E}$ be Hadamard-differentiable at $\theta$ tangentially to $\mathbb{D}_0$. Let $X_n : \Omega_n \mapsto \mathbb{D}_\phi$ be maps with $r_n(X_n - \theta) \rightsquigarrow X$ for some sequence of constants $r_n \rightarrow \infty$, where $X$ is separable and takes its values in $\mathbb{D}_0$. Then $r_n (\phi(X_n) - \phi(\theta)) \rightsquigarrow \phi'_\theta(X)$. If $\phi'_\theta$ is defined and continuous on the whole of $\mathbb{D}$, then the sequence $r_n(\phi(X_n) - \phi(\theta)) - \phi'_\theta(r_n(X_n - \theta))$ converges to zero in outer probability. 
\end{lem} 

\noindent We will use the above lemma and modifications of the arguments from the functional delta method for the bootstrap \citep[Theorem 3.9.11 in][]{van1996weak}, to derive the validity of the bootstrap procedure. 

 \begin{lem}\label{lem:lundeShal}\citep[Lemma 3.8 in][]{lunde2017bootstrapping} Let $Z_t$ be a $\beta$-mixing (stationary) process with mixing rate that decays at least at a cubic rate. Consider $\mathbb{H}_t$ defined in \eqref{eqn:process}. Then 
	$$
	\mathbb{H}_t \rightsquigarrow \mathbb{H} = \mathbb{G} \times \mathbb{G}
	$$
	where $\mathbb{H}$ is a bivariate Gaussian process with $\times$ symbol denoting independence. Furthermore, is a mean zero Gaussian Process with covariance structure given by \begin{equation}
	\Gamma(f,g) = \lim_{k \rightarrow \infty} \sum_{i=1}^{\infty} \{\mathbb{E}[f(Z_k) g(Z_i)] - \mathbb{E}[f(Z_k)] \mathbb{E}[g(Z_i)]\}, \quad \forall f,g \in \mathcal{F}.
	\end{equation}  
	 \end{lem} 
	 
	 	  \begin{lem}\label{lem:kosorok}\citep[][Theorem 11.26]{kosorok2008introduction}  \textit{Let Y be a stationary sequence in $\mathbb{R}^d$ with marginal distribution P and let $\mathcal{F}$ be a class of functions in $L_2(P)$. Let $\mathbb{G}_n^*(f) = Y_n^*f - Y_n f$. Also assume that $Y_1^*, Y_2^*,\dots,Y_n^*$ are generated by the circular block bootstrap procedure with $b(n) \rightarrow \infty$ as $n \rightarrow \infty$ and that there exist a $2 < v < \infty$, $q > v/(v-2)$ and $0< \rho < (v-2)/[2v - 2]$ such that: \begin{enumerate}
		\item $\lim\sup_{k \rightarrow \infty} k^{q}\beta(k) < \infty;$
		\item $\mathcal{F}$ is permissible, VC and has envelope F satisfying $PF^v < \infty$;
		 		\item $\lim\sup_{n \rightarrow \infty}n^{-\rho}b(n) < \infty$.
	\end{enumerate}
	Then $\mathbb{G}_n^* \rightsquigarrow \mathbb{G} \in l^\infty(\mathcal{F})$ where $\mathbb{G}$ is a mean $0$ Gaussian Process with covariance structure $\Gamma(f,g) = \lim_{k \rightarrow \infty} \sum_{i=1}^{\infty} \{\mathbb{E}[f(Y_i) g(Y_k)] - \mathbb{E}[f(Y_i)]\mathbb{E}[g(Y_k)]\}, \quad \forall f,g, \in \mathcal{F}$.} 
\end{lem}

\begin{lem} \label{lem:vandervaar}\citep[][Lemma A.3.3]{lunde2017bootstrapping}
\textit{Given a continuous function $A$ and a function of bounded variation $B$ in the sense of Hardy-Krause in the hyper-rectangol $\mathcal{R} = \prod_{i=1}^d [a_i,b_i]$ define 
	$$
	\phi(A,B) = \int_{[a,b]} A dB.
	$$
	Then, $\phi: C(\mathcal{R}) \times BV_M(\mathcal{R}) \rightarrow \mathbb{R}$ is Hadamard differentiable at each $(A,B) \in \mathbb{D}_\phi$ such that $\int |dA| < \infty$. The derivative is given by
	$$
	\phi'_{A,B}(a,\beta)= \int_{[a,b]} A d\beta + \int_{[a,b]}a dB.
	$$}
\end{lem}

\begin{lem} \label{lem:meas} Consider two multivariate processes $\{X_t,Y_t\}$ and $\{f(X_t), Y_t\}$ for some measurable function $f$. Let $\beta_1(h)$ and $\beta_2(h)$ be respectively the beta-mixing coefficient of $\{X_t,Y_t\}$ and $\{f(X_t), Y_t\}$. Then 
$$
\beta_2(h) \le \beta_1(h) .
$$
\end{lem}
\begin{proof}[Proof of Lemma \ref{lem:meas}]
 For two random variables $(X,Y)$ and a measurable function $f$, $\sigma(f(X)) \subseteq \sigma(X)$, since the pullback $f\circ X(A)^{-1} = X^{-1}(f^{-1}(A)) \in \sigma(X)$ by measurability of $f$ for a given event $A$. Henceforth for any $t$, 
 $$
 \sigma((f(X_t), Y_t))  \subseteq \sigma((X_t, Y_t)).
 $$
This implies that 
$$
\begin{aligned} 
&\sup_{A \in \sigma((f(X_t), Y_t)), B \in \sigma((f(X_{t + h}), Y_{t + h}))} \sum_i \sum_j |P(A_i \cap B_j) - P(A_i)P(B_j)| \\ &\le  \sup_{A \in \sigma((X_t, Y_t)), B \in \sigma((X_{t + h}, Y_{t + h})} \sum_i \sum_j|P(A_i \cap B_j) - P(A_i)P(B_j)|.
\end{aligned} 
$$
where the supremum is over all paris of finite partitions $\{A_i\}, \{B_j\}$ 
such that $A_i \in \mathcal{A}$ and $B_j \in \mathcal{B}$ where $\mathcal{A}$ and $\mathcal{B}$ are the sigma algebras generated by the random variables of interest. Henceforth, exploting the definition of $\beta$-mixing given in \cite{bradley2005basic} we have
$$
\begin{aligned} 
&\beta_2(h) = \sup_t \sup_{A \in \sigma((f(X_t), Y_t)), B \in \sigma((f(X_{t + h}), Y_{t + h}))} \sum_i \sum_j |P(A_i \cap B_j) - P(A_i)P(B_j)| \\ &\le\sup_t\sup_{A \in \sigma((X_t, Y_t)), B \in \sigma((X_{t + h}, Y_{t + h})} \sum_i \sum_j|P(A_i \cap B_j) - P(A_i)P(B_j)| = \beta_1(h).
\end{aligned}  
$$
\end{proof}

\begin{lem}\citep[][Lemma 3.7]{lunde2017bootstrapping} \label{lem:lunde_shal2} Let $Q$ and $R$ be probability distributions on $(\Omega^j, \mathcal{A}^j)$ where $j$ is allowed to be $\infty$. Define the space of histories $(H_t, \mathcal{H}_t)$ , where $H_t = \Omega^t$ and $\mathcal{H}_t = \mathcal{A}^t$ and let $h_t \in \mathcal{H}_t$ be an history. Suppose $Q$, $R$ permit regular conditional probabilities, denoted $Q(\cdot | h_t)$, and $R(\cdot| h_t)$, respectively. Then, for $||\cdot||_{TV}$ denoting the total variation distance, 
$$
||Q - R||_{TV} < ab \Rightarrow Q(||Q(\cdot| h_t) - R(\cdot|h_t)||_{TV} > a) < b. 
$$  
\end{lem}

\subsection{Proof of Theorem \ref{thm:1}} \label{sec:final}

Recall that the theorem is derived under the null hypothesis $H_0$. 
Consider the case where $g(\cdot)$ are fixed functions. We then return to the case where these are estimated based on $F_-$ (the hold-out sample) at the end of the proof. 
We first prove that $\mathbb{H}_T^*$ and $\mathbb{H}_T$ which we define as the empirical and bootstrapped measures  as in \eqref{eqn:process},   converge to the same process $\mathbb{H}$. 
 We let $2T_0 = T$ for notational convenience, but the results directly hold as $T_0 \propto T$ after rescaling by an appropriate constant.\footnote{In that case, we would also rescale appropriately each component on $\mathbb{H}_T$ by its corresponding constant term.}  By Lemma \ref{lem:meas} beta-mixing conditions on $(Y_t^o,X_t)$ imply the same conditions on the beta-mixing coefficients of $(Y_t^o, g(X_t))$. Similarly, stationarity of $(Y_t^o, X_t)$ also implies stationarity of $(Y_t^o, g(X_t))$. \\ \\ 
To apply the functional delta method we first need to show that  $\mathbb{H}_T^*$ and $\mathbb{H}_T$ converge to the same process up to a multiplicative constant. By Lemma \ref{lem:lundeShal} $\mathbb{H}_T \rightarrow_d \mathbb{H}$. Therefore, by Lemma \ref{lem:delta_method}, we have that for an Hadamard differentiable map as defined in Lemma \ref{lem:delta_method}, 
$$
\sqrt{T_0} \Big(\phi([\mathcal{P}_{T_0}, \mathcal{P}_{T_1}]) - \phi([\mathcal{P}, \mathcal{P}])\Big) \rightarrow_d \phi'_{\mathcal{P}}(\mathbb{H}).  
$$
By Lemma \ref{lem:kosorok}, $\mathcal{P}_{T_0}^* - \mathcal{P}_{T}$ and $\mathcal{P}_{T}^* - \mathcal{P}_{T}$ (appropriately rescaled) converge marginally to a Gaussian Process with covariance matrix described in Lemma \ref{lem:kosorok}.\footnote{Conditions in Lemma \ref{lem:kosorok} are justified for the following reasons. As discussed in Section 9.1.1, \cite{kosorok2008introduction} the class $\mathcal{F}$ has bounded VC dimension. In addition, the class is also permissible since it satisfies the two requirements of permissibility: we can index the class by a set $T = \mathbb{R}^p$ that is a valid Polish space equipped with Borel sigma field.} Since $\mathcal{P}_{T_0}^*, \mathcal{P}_T^*$ are drawn independently conditional on $\mathcal{P}_T$, we can conclude that $\mathbb{H}_T^* \rightarrow_d \mathbb{H}$ conditionally on $\mathcal{P}_T$.   \\ \\
We now follow the proof of Theorem 3.9.11 in \cite{van1996weak} with some minor modifications given the presence of $(\mathcal{P}_{T_0}, \mathcal{P}_{T_1})$ to show consistency of the bootstrap for a genering Hadamard differentiable functional $\phi$ as defined in Lemma \ref{lem:delta_method}. Define $\mathbb{P}_T^* = [\mathcal{P}_{T_0}^*, \mathcal{P}_{T_1}^*]$. Note that $\mathbb{H}_T^* = \sqrt{T_0}(\mathbb{P}_T^* - [\mathcal{P}_T, \mathcal{P}_T])$. Let $\mathbb{P}_T = [\mathcal{P}_T, \mathcal{P}_T]$. Note that by Theorem 11.25 in \cite{kosorok2008introduction}, $\sqrt{2 T_0}(\mathcal{P}_T - \mathcal{P}) \rightarrow_d \mathbb{G}$ where $\mathbb{G}$ is as defined in Lemma \ref{lem:lundeShal}.  
Since $\mathbb{H}_T \rightarrow_d \mathbb{H}$, by Lemma \ref{lem:delta_method},
$$
\begin{aligned} 
& \sup_{h \in \mathrm{BL}_1(\mathbb{E})} \Big| \mathbb{E} h\Big(\sqrt{T_0}(\phi(\mathbb{P}_T^*) - \phi(\mathbb{P}_T) \Big) -\mathbb{E} h\Big(\sqrt{T_0}(\phi([\mathcal{P}_{T_0}, \mathcal{P}_{T_1}]) - \phi([\mathcal{P}, \mathcal{P}]) \Big) \Big| \\ 
& \le  \underbrace{\sup_{h \in \mathrm{BL}_1(\mathbb{E})} \Big| \mathbb{E} h\Big(\sqrt{T_0}(\phi(\mathbb{P}_T^*) - \phi(\mathbb{P}_T) \Big) - \mathbb{E} h (\phi'_{\mathcal{P}}(\mathbb{H}))\Big|}_{(i)} + o(1).  
\end{aligned} 
$$ 
We want to show that $(i)$ converges to zero.  

Following \cite{van1996weak} (Page 380), without loss of generality, we can assume that $\phi'_P: \mathbb{D} \mapsto \mathbb{E}$ is defined and continuous on the whole space. For every $h \in   \mathrm{BL}_1(\mathbb{E})$, the function $h \circ \phi'_P$ is contained in $\mathrm{BL}_{||\phi'_P||}(\mathbb{D})$. Therefore, since $\mathbb{H}_T^* \rightarrow_d \mathbb{H}$, we can write 
$$
\sup_{h \in \mathrm{BL}_1(\mathbb{E})} \Big| \mathbb{E} h(\phi_{\mathcal{P}}'(\mathbb{H}_T^*)) - \mathbb{E} h (\phi'_{\mathcal{P}}(\mathbb{H}))\Big| = o(1).  
$$
Next, 
$$
\small 
\begin{aligned} 
\sup_{h \in \mathrm{BL}_1(\mathbb{E})} \Big| \mathbb{E} h\Big(\sqrt{T_0}(\phi(\mathbb{P}_T^*) - \phi(\mathbb{P}_T) \Big) - \mathbb{E} h(\phi_P'(\mathbb{H}_T^*))\Big| &= 
\sup_{h \in \mathrm{BL}_1(\mathbb{E})} \Big| \mathbb{E} h\Big(\sqrt{T_0}(\phi(\mathbb{P}_T^*) - \phi(\mathbb{P}_T) \Big) - \mathbb{E} h(\phi_P'(\sqrt{T_0}(\mathbb{P}_T^* - \mathbb{P}_T)))\Big| \\ 
&\le \varepsilon  + 2 P\Big(\Big| \Big| \sqrt{T_0}(\phi(\mathbb{P}_T^*) - \phi(\mathbb{P}_T) \Big) -  \phi_P'(\sqrt{T_0}(\mathbb{P}_T^* - \mathbb{P}_T)))\Big| \Big| > \varepsilon\Big),   
\end{aligned} 
$$
since we took without loss of generality the supremum over a ball with radius one. 
Notice now that $\sqrt{T_0}(\mathbb{P}_T^* - [\mathcal{P}, \mathcal{P}])$ and $\sqrt{T_0}(\mathbb{P}_T - [\mathcal{P}, \mathcal{P}])$ they both converge (unconditionally) in distribution to a separable random element that concentrate on $\mathbb{D}_0$. The first case follows since $\mathbb{H}_T^*$ converges in distribution to $\mathbb{H}$ and $\sqrt{2 T_0}(\mathcal{P}_T - \mathcal{P})$ converges in distribution to $\mathbb{G}$. The second case follows directly from the fact that $\sqrt{2 T_0}(\mathcal{P}_T - \mathcal{P})$ converges in distribution to $\mathbb{G}$. 
Using the functional delta method, we can write 
$$
\small 
\begin{aligned} 
\sqrt{T_0}(\phi(\mathbb{P}_T^*) - \phi(\mathcal{P}, \mathcal{P})) = \phi_{\mathcal{P}}'(\sqrt{T_0}(\mathbb{P}_T^* - (\mathcal{P}, \mathcal{P}))) + o_p(1), \quad  
\sqrt{T_0}(\phi(\mathbb{P}_T) - \phi(\mathcal{P}, \mathcal{P})) = \phi_{\mathcal{P}}'(\sqrt{T_0}(\mathbb{P}_T - (\mathcal{P}, \mathcal{P}))) + o_p(1). 
\end{aligned} 
$$ 
Subctracting the two equations, we obtain the desired result. 

We now want to show Hadamard differentiability. 
 We now show Hadamard differentiability for $T(\cdot,\cdot)$ for $S_0(z,w) = |z(-1,w)|^2$, with $T(\cdot,\cdot)$ being the test statistic of interest. 
The proof invokes the chain rule for Hadamard differentiable maps \citep{van2000asymptotic}. We can see $T(\cdot,\cdot)$ as the composition of maps
$$
T: (A,B) \overset{(a)}{\rightarrow} (B, w(A)) \overset{(b)}{\rightarrow} (B, S_0(z, w(A))) \overset{(c)}{\rightarrow} \int S_0(z, w(A)) dB.
$$
The map $(c)$ is Hadamard differentiable by Lemma \ref{lem:vandervaar}. 
 We have to show that $S_0(z,w)$ is itself Hadamard differentiable in $w$ at $w(\mathcal{P})$ that we write as $\bar{w}$ for short. 
We start by proving Hadamard differentiability of 
$$S_0(z,w) = [z'(-1,w)]^2.$$ We omit the $S_0$ notation for sake of simplicity in the next few lines. To show Hadamard differentiability we  need to show that 
\begin{equation} \label{eqn:absHad}
\left \| \frac{S(.,\bar{w} + t_n h_n) - S(., \bar{w})}{t_n} - S'_{\bar{w}}(.,h) \right\|_{\infty} \rightarrow 0
\end{equation}
where $$S_{\bar{w}}'((z_1,z_2),h) = 2(z_1 - z_2\bar{w})h'z_2.$$ We can rewrite the LHS above as 
\begin{equation} \label{eqn:absHad2}
||2(z_1 - z_2'\bar{w})z_2'(h_n - h)||_{\infty} + |t_n|||(z_2'h_n)^2||_{\infty} .
\end{equation}
For the first term in \eqref{eqn:absHad2} by the assumption of compact support, there must exist $c<\infty$ such that 
$$\sup_{z_1,z_2}|2(z_1 - z_2'\bar{w})z_2'(h_n - h)| < c ||h_n - h||_1$$ and since $||h_n - h|| \rightarrow 0$ the term goes to zero. Equivalently for the second term, since $(z_2'h_n)^2 < \infty$ by the compactness assumption, and $t_n \rightarrow 0$ also the second term converges to zero. We are left to show that $A \mapsto w(A)$ is Hadamard differentiable, for $w$. This directly follows by assumption. \\ \\
Finally, consider the case where $g(\cdot)$ are estimated using information from $F_-$ (i.e., using the hold-out sample indexed by $t < 0$), where $T_-$ is kept fixed with $T$ (recall the definition of $\mathcal{P}$ at the end of Appendix \ref{sec:definitions}). To invoke the functional delta method in this case, we need to show that the asymptotic statement in Lemma \ref{lem:lundeShal} holds conditionally on the initial sample $(Z_{T_-}, \cdots, Z_0)$, letting $T \rightarrow \infty$. Using the definition of weak convergence, we need to show that 
$$
\sup_{h \in \mathrm{BL}_1} \Big|\mathbb{E}\Big[h\Big(\sqrt{T_0}[\mathcal{P}_{T_0} - \mathcal{P}, \mathcal{P}_{T_1} - \mathcal{P}]\Big)\Big|Z_{T_-}, \cdots, Z_0\Big] - \mathbb{E}[h(c \mathbb{H})]\Big| \rightarrow 0, 
$$   
up to a rescaling constant $c$.
To show the claim, we exploit the $\beta$-mixing conditions. First, we observe that we can write 
$$
\sqrt{T_0}[\mathcal{P}_{T_0} - \mathcal{P}, \mathcal{P}_{T_1} - \mathcal{P}] = \sqrt{T_0}\Big[\underbrace{\frac{1}{T_0} \sum_{t = 1}^m \delta_{Z_t}  - \frac{m}{T_0} \mathcal{P}}_{(A)} + \underbrace{\frac{1}{T_0} \sum_{T_0 \ge t > m} \delta_{Z_t}  - (1 - \frac{m}{T_0}) \mathcal{P}}_{:=\mathcal{P}_{m, T_0}}, \mathcal{P}_{T_1} - \mathcal{P}\Big] , 
$$ 
for any $m$. We take $m \rightarrow \infty$ and $m/\sqrt{T_0} \rightarrow 0$, which implies $(A) = o(1)$ almost surely. 
Since $(A) = o(1)$ we can write 
$$
\begin{aligned} 
& \sup_{h \in \mathrm{BL}_1} \Big|\mathbb{E}\Big[h\Big(\sqrt{T_0}[\mathcal{P}_{T_0} - \mathcal{P}, \mathcal{P}_{T_1} - \mathcal{P}]\Big)\Big|Z_{T_-}, \cdots, Z_0\Big] - \mathbb{E}[h(\mathbb{H})]\Big| \\ &= \sup_{h \in \mathrm{BL}_1} \Big|\mathbb{E}\Big[h\Big(\sqrt{T_0}[\mathcal{P}_{m, T_0}, \mathcal{P}_{T_1} - \mathcal{P}]\Big)\Big|Z_{T_-}, \cdots, Z_0\Big] - \mathbb{E}[h(\mathbb{H})]\Big| + \varepsilon + 2P\Big(||\sqrt{T_0}(A)|| > \varepsilon| Z_{T_-}, \cdots, Z_0\Big) \\ 
& = \sup_{h \in \mathrm{BL}_1} \Big|\mathbb{E}\Big[h\Big(\sqrt{T_0}[\mathcal{P}_{m, T_0}, \mathcal{P}_{T_1} - \mathcal{P}]\Big)\Big|Z_{T_-}, \cdots, Z_0\Big] - \mathbb{E}[h(\mathbb{H})]\Big| + o(1),  
\end{aligned} 
$$
for a suitable choice of $\varepsilon = o(1)$. 
Also, note that $\sqrt{T_0}[\mathcal{P}_{m, T_0}, \mathcal{P}_{T_1} - \mathcal{P}]$ satisfies the mixing conditions in Lemma \ref{lem:lundeShal}. 

Define $\mu_m$ a sequence of probability distributions of $(Z_{T_-}, \cdots, Z_0, Z_{m}, \cdots)$, and $\mu_{\infty}$ the product measure of $(Z_{T_-}, \cdots, Z_0)$  and $(Z_m, \cdots)$. Then under the assumed mixed conditions, by Lemma \ref{lem:lunde_shal2}
$$
\sum_{m=1}^{\infty}P(||\mu_m(\cdot|h_0) - \mu_{\infty}(\cdot|h_0)||_{TV} > a_m) < \infty,  
$$ 
for $a_m \rightarrow 0$, which implies $||\mu_m(\cdot|h_0) - \mu_{\infty}(\cdot|h_0)||_{TV} \rightarrow \infty$ by the Borel-Cantelli Lemma. This implies that for any bounded function $f$
$$
\int f d\mu_m(\cdot|h_0) \rightarrow 
\int f d\mu_{\infty}(\cdot|h_0). 
$$  
Hence, we can write 
$$
\sup_{h \in \mathrm{BL}_1} \Big|\mathbb{E}\Big[h\Big(\sqrt{T_0}[\mathcal{P}_{m, T_0}, \mathcal{P}_{T_1} - \mathcal{P}]\Big)\Big|Z_{T_-}, \cdots, Z_0\Big] - \mathbb{E}[h(\mathbb{H})]\Big| = \sup_{h \in \mathrm{BL}_1} \Big|\int h\Big(\sqrt{T_0}[\mathcal{P}_{m, T_0}, \mathcal{P}_{T_1} - \mathcal{P}]\Big) d\mu_{\infty} - \mathbb{E}[h(\mathbb{H})]\Big| + o(1).
$$ 
The same reasoning applies to the bootstrapped process $\mathbb{H}_T^*$. Since $\mu_{\infty}$ is the product measure of $(Z_{T_-}, \cdots, Z_0)$  and $(Z_m, \cdots)$ the rest of the proof follows directly from the unconditional case discussed at the beginning of the proof.

\subsection{Proof of Theorem \ref{thm:consistency}}

Note that under Assumption \ref{ass:ident} we can fix $T_0$ as non-random (condition on $T_0$) without affecting the distribution of observables and unobservables. 
 To prove the claim we can write the estimator as follows 
$$
\hat{\tau} = \underbrace{\frac{1}{T - T_0} \sum_{t > T_0} Y_t  - \frac{1}{T_0/2} \sum_{t = T_0/2 + 1}^{T_0} Y_t   }_{(A)} +  \underbrace{\sum_{t = T_0/2 + 1}^{T_0} w(F_{0, -1/2}) g(X_t) - \sum_{t > T_0}^{T_0} w(F_{0}) g(X_t)  }_{(B)}. 
$$ 
Consider first $(A)$.  By beta-mixing (and hence alpha-mixing), by the Ergodic Theorem\footnote{See Theorem 3.34  in \cite{white2014asymptotic}.}
$$
(A) - \tau \rightarrow_p 0. 
$$ 

We are left to show that $(B)$ converges to zero. 
We define $\tilde{F}_1$ the empirical distribution of $X_t$ for $t > T_0$ and $\tilde{F}_{0, 1/2}$ the empirical distribution of $X_t$  for $T_0 \ge t > T_0/2 + 1$. Note that by Assumption \ref{ass:ident} and stationarity (Assumption \ref{ass:stationarity}), we obtain 
$
\mathbb{E}[\tilde{F}_1] = \mathbb{E}[\tilde{F}_{0, 1/2}] = \mathcal{P}_X
$
with $\mathcal{P}_X$ denoting the marginal distribution of $X_t$. Here, $\mathcal{P}_X$ is time-invariant by Assumption \ref{ass:stationarity} and the fact that we subctract the fixed effects components. We also define $\mathcal{P} = \mathbb{E}[F_0]$ denoting the distribution of $(X_t, Y_t)$ (time invariant by Assumption \ref{ass:stationarity}). Finally, recall the definition of $F_-$ denoting the empirical distribution over the sample $t < 0$.

We now write, in compact form, with a change of notation, and making explicit the dependence of $g(X_t, F_-)$ with $F_-$,
$$
(B)  = \int w(F_{0, -1/2}) g(X_t; F_-) d\tilde{F}_{0, 1/2} - \int w(F_{0}) g(X_t; F_-) d\tilde{F}_{1}. 
$$ 
We now discuss each component after \textit{conditioning} on $F_-$. In particular, we note that we can interpret each component as map 
$$
T(A, B; F_-) = \int w(A) g(x, F_-) dB,  
$$ 
which is Hadamard differentiable in $A$ and $B$ by Lemma \ref{lem:vandervaar} (but not necessarily in $F_-$).  Following the argument of the proof of Theorem \ref{thm:1} verbatim and invoking Lemma \ref{lem:lundeShal} and beta-mixing as in the proof of Theorem \ref{thm:1} (see Appendix \ref{sec:final}), we obtain 
$$
\begin{aligned} 
& \sqrt{T_0} (T(F_{0, -1/2}, \tilde{F}_{0,1/2}; F_-) - T(\mathcal{P}, \mathcal{P}_X; F_-)) \rightarrow_d T_{\mathcal{P}, \mathcal{P}_X}'(\mathbb{Y}; F_-), \\ &
\sqrt{T_0} (T(F_{0}, \tilde{F}_{1}; F_-) - T(\mathcal{P}, \mathcal{P}_X; F_-)) \rightarrow_d T_{\mathcal{P}, \mathcal{P}_X}'(\tilde{\mathbb{Y}}; F_-) 
\end{aligned} 
$$ 
for two tight processes $(\mathbb{Y}, \tilde{\mathbb{Y}})$. This implies that $T(F_{0, -1/2}, \tilde{F}_{0,1/2}; \mathcal{P}), T(F_0, \tilde{F}_1)$ converges in probability to the same asymptotic limit. This implies that 
$$
\int w(F_{0, -1/2}) g(X_t; F_-) d\tilde{F}_{0, 1/2} - \int w(F_{0}) g(X_t; F_-) d\tilde{F}_{1} = o_P(1), 
$$  
completing the proof. 

\subsection{Additional Results} \label{sec:hadamrd}

It is interesting to study Hadamard differentiability of the weights.  Hadamard differentiability for Least Squares has been shown in other papers, such as in  \cite{lunde2017bootstrapping}. Hence, we only need to show Hadamard differentiability for exponential weights.

\begin{thm} \label{thm:differentiability} Consider exponential weights as in Equation \eqref{eqn:expweights1} with $\eta \propto 1/T_0$, and $(Y_t, g(X_t))$ uniformly bounded. Then such weights are Hadamard differentiable.  
\end{thm} 
\begin{proof}[Proof of Theorem \ref{thm:differentiability}]
Since we consider a finite dimensional parameter space, we can prove Hadamard differentiability by showing Hadamard differentiability for each coordinate. We will assume that 
 $$\eta_t = \frac{\eta}{t}.$$ We let $l: \mathbb{R}^2 \rightarrow \mathbb{R}_+$ the loss function and we assume the loss can be at most $C< \infty$ and $\eta > 0$(which follows for quadratic loss functions and bounded random variables). The aim is to show that 
 \begin{equation}
 \hat{w}_0^{(i)}(F) = \frac{\exp(-\eta \int l(z_1, z_i) dF)}{\sum_{j=2}^{p+1} \exp(-\eta \int l(z_1, z_j) dF)}
 \end{equation}
is Hadamard differentiable with $Z = (z_1,\dots,z_{p+1}) \in [-M,M]^{p+1}$. We will use the chain rule for the Hadamard derivative. We can see $w^{(i)}(F)$ is the composition of mappings: 
$$
\begin{aligned}
A & \overset{(a)}{\mapsto} \left [ \begin{tabular}{c}\text{$\int l(z_1, z_2) dA$} \\
\dots \\ 
\text{$\int l(z_1, z_{p+1}) dA$} 
\end{tabular} \right ] 
\\
&\overset{(b)}{\mapsto} \left [ \begin{tabular}{c}\text{$\exp(-\eta \int l(z_1, z_2) dA)$} \\
\dots \\ 
\text{$\exp(-\eta \int l(z_1, z_{p+1}) dA)$}
\end{tabular} \right ] \\ &\overset{(c)}{\mapsto} \left [ \begin{tabular}{c}\text{$\exp(-\eta \int l(z_1, z_i) dA)$} \\
 \text{$\sum_{j=1}^p\exp(-\eta \int l(z_1, z_{j+1}) dA)$}
\end{tabular} \right ] 
\\
& \overset{(d)}{\mapsto} \frac{\exp(-\eta \int l(z_1, z_i) dA)}{\sum_{j=1}^{p} \exp(-\eta \int l(z_1, z_{j+1}) dA)} \end{aligned}.
$$
We will prove that each map is Hadamard differentiable under the conditions stated component wise. We start by proving that $(a)$ is Hadamard differentiable component wise. For $||h_n - h||_{\infty} \rightarrow 0$, $t_n \rightarrow 0$
\begin{align*}
\frac{\int l(z_1, z_j) d(F+t_n h_n) - \int l(z_1, z_j) dF}{t_n} &= \frac{t_n \int l(z_1, z_j)d h_n}{t_n} + \frac{\int l(z_1, z_j)d(F-F)}{t_n} \\
&\rightarrow \int l(z_1, z_j)dh.
\end{align*}
Since $l(z_1, z_j)$ is bounded, uniform convergence follows. 
Hence 
$$
\left \| \frac{\int l(z_1, z_j) d(F+t_n h_n) - \int l(z_1, z_j) dF}{t_n} - \int l(z_1,z_j)dh \right\|_{\infty} \rightarrow 0.
$$
We now move to $(b)$. Let $x$ be the argument of the map. Recall that the argument is positive. Using the mean value theorem, for $\bar{h}_n \in [h_n,h]$,
\begin{equation}
\begin{aligned}
\frac{\exp(-\eta (x + t_n h_n)) - \exp(-\eta (x))}{t_n} &= \frac{\exp(-\eta x)[\exp(-\eta t_n h_n) - 1]}{t_n} 
\\
&= \frac{\exp(-\eta x)[-\eta t_n \exp(-\eta t_n \bar{h}_n)h_n]}{t_n}  \\ &= -\eta \exp(-\eta x -\eta t_n \bar{h}_n) h_n \rightarrow -\eta \exp(-\eta x) h .
\end{aligned}
\end{equation}
Hence we have 
\begin{align*}
&\left\| \frac{\exp(-\eta (x + t_n h_n)) - \exp(-\eta (x))}{t_n} + \eta \exp(-\eta x) h \right\|_{\infty} 
\\
&\qquad \qquad = \sup_{x \in \mathbb{R}_+} | \exp(-\eta x)[\frac{\exp(-\eta t_n h_n) - 1}{t_n} + \eta h]| \\ &\qquad \qquad  \le |\frac{\exp(-\eta t_n h_n) - 1}{t_n} + \eta h| = | -\eta \exp(-\eta t_n \bar{h}_n)h_n + \eta h| \rightarrow 0
\end{align*}
since $\exp(-\eta x) \le 1$ for $x \ge 0$. Since $\eta h \exp(-\eta x)$ is linear in $h$, $(b)$ is Hadamard differentiable. The map $(c)$ is Hadamard differentiable by linearity of the Hadamard derivative. We are left to prove $(d)$. Let 
$$t_n \rightarrow 0, \quad \left \| h_n -  \biggl( \begin{tabular}{c}\text{$h_1$}    \\
 \text{$h_2$}
\end{tabular} \biggl ) \right\|_{\infty} \rightarrow 0.
$$

Take $N$ such that $p e^{-\eta C} > |t_N h_{N,2}|$. Such $N$ exists since $t_n \rightarrow 0$ and $h_{n,2} \rightarrow h_2 < \infty$. For $n > N$ we have
\begin{equation}
\begin{aligned}
&\left\|  \frac{\frac{x + t_n h_{n,1}}{y + t_n h_{n,2}} - \frac{x}{y}}{t_n} - \frac{h_1}{y} + h_2 \frac{x}{y^2} \right \|_{\infty} \\ &= \sup_{0 \le x \le 1,  y \ge r e^{-\eta C}} 
\left\| 
\frac{(h_{n,1} - h_1) - (h_{n,2} - h_2)x/y - h_1 t_n h_{n,2}/y + h_2 x h_{n,2} t_n/y^2}{y + t_n h_{n,2}}
\right \| _\infty
 \\ &\le \sup_{0 \le x \le 1,  y \ge r e^{-\eta C}} \frac{|h_{n,1} - h_1|}{|y + t_n h_{n,2}|} + \frac{|h_{n,2} - h_2||\frac{x}{y}|}{|y + t_n h_{n,2}|} + |t_n| \frac{|h_1 h_{n,2}|/|y|}{|y + t_n h_{n,2}|} + \frac{|h_2 h_{n,2}x/y^2|}{|y +  t_n h_{n,2}|}|t_n| \\ &\le  \frac{|h_{n,1} - h_1|}{ p e^{-\eta C} - |t_n h_{n,2}|} + \frac{|h_{n,2} - h_2|}{pe^{-\eta C}(p e^{-\eta C} - |t_n h_{n,2}|)} + |t_n| \frac{|h_1 h_{n,2}|}{p e^{-\eta C}(p e^{-\eta C} - |t_n h_{n,2}|)}\\ & + \frac{|h_2 h_{n,2}|}{p e^{-2\eta C}(p e^{-\eta C} - |t_n h_{n,2}|)}|t_n| \\ & \rightarrow 0.
\end{aligned}
\end{equation}
since each term is going to zero and the denominator of each term is bounded away from zero.  Hence the Hadamard derivative is $-\frac{h_1}{y} + h_2 \frac{x}{y^2}$, which is linear in $h$. 
\end{proof} 

We discuss the implication of our results for the asymptotic distribution of the test statistic in the following theorem. 

\begin{thm}  Under the conditions in Theorem \ref{thm:1} the test statistics $\mathcal{T}(A, B)$, are Hadamard differentiable tangentially at $(\mathcal{P}_0, \mathcal{P}_0)$ with derivative  $\mathcal{T}_{\mathcal{S}, \mathcal{P}_0, \mathcal{P}_0}'$. Therefore, 
$$
\mathcal{T}_{\mathcal{S}}(\mathcal{P}_{T_0}, \mathcal{P}_T) -  \mathcal{T}_{\mathcal{S}}(\mathcal{P}_0, \mathcal{P}_0)  \rightarrow_d \mathcal{T}_{\mathcal{S}, \mathcal{P}_0, \mathcal{P}_0}'(\mathbb{H})
$$
where $\mathcal{P}_0, \mathbb{H}$ are as defined respectively as the true CDF of $(Y_t^0, X_t)$ and as defined in Lemma \ref{lem:lundeShal}. 
\end{thm} 
 
The proof follows directly from the differentiability properties discussed in the derivation of Theorem \ref{thm:1}, and by the Functional Delta Method. 
Observe that the asymptotic distribution depends on the Hadamard derivative of the test statistics which ultimately depends on the weighting mechanism. Its expression is obtained through the chain rule in the proof of Theorem \ref{thm:1} .

\section{Proofs of the Results in Section  \ref{sec:weights}} \label{app:5}

 With an abuse of notation we define the prediction of the potential outcome at time $t$ as 
 $$\hat{Y}_t^0(\mathcal{F}_{t-1}):=\hat{m}_t(X_t, w^{t-1})$$
   where $w^{t-1}$ are exponential weights as in \eqref{eqn:expweights1} computed using information available only until time $t-1$. 

\begin{lem} \label{lem:cesabianchi} Consider the Synthetic Learner with the exponential weighting scheme with $\hat{g}_i: \mathcal{X} \rightarrow [-\frac{M}{2},\frac{M}{2}]$. Then  

\begin{equation} 
\begin{aligned} 
&\frac{1}{T_0} \sum_{t=1}^{T_0}(Y_t - \hat{m}_t(X_t, w^{t-1}))^2 - \min_{i \in \{1, \dots, p\}} \frac{1}{T_0}  \sum_{t=1}^{T_0}(Y_t - \hat{g}_{i}(X_t))^2 \le  \log(p)/\eta + \sum_{t=1}^{T_0} C_t^2 \eta/ 8
\end{aligned} 
\end{equation} 
where $C_t = Y_t^2 + M^2$. 
\end{lem} 
\begin{proof} [Proof of Lemma \ref{lem:cesabianchi} ]
The proof follows similarly to \cite{cesa1999prediction} with appropriate modifications for the context under consideration. Let 
$$w_{i,t} = \sum_{i =1}^p \exp(-\eta \sum_{s=1}^t l(Y_s, \hat{g}_i(X_s))), \quad W_t = \sum_{i =1}^p w_{i,t}.$$ We denote $\hat{g}_i(X_t) :=\hat{g}_{i,t}$ for expositional convenience. We recall that the weights in the first period are uniform across the learners. 

We start by deriving a lower bound:
\begin{align*}
\log(W_{T_0}/w) =\log (W_{T_0}) -\log (p) &=\log \Big(\sum_{i=1}^p \exp(- \eta \sum_{t=1}^{T_0} (Y_t - \hat{g}_{i,t})^2)\Big) -\log (p) \\ 
&\ge\log (\max_{i \in \{1, \dots, p\}} \exp(- \eta \sum_{t=1}^{T_0} (Y_t - \hat{g}_{i,t})^2) -\log (p) \\ &= -\eta \text{min}_{i \in \{1, \dots, p\}} \sum_{t=1}^T (Y_t - \hat{g}_{i,t})^2 -\log (p).
\end{align*}
Next, we derive an upper bound on the same quantity of interest. 
\begin{align*}
\log(W_{T_0}/w) &=\log (\prod_{t=1}^{T_0} W_t/W_{t-1}) = \sum_{t=1}^{T_0}\log \Big(\sum_{i=1}^p \frac{w_{i,t-1}}{W_{t-1}} \exp(-\eta (Y_t - \hat{g}_{i,t})^2)\Big) \\ &= \sum_{t=1}^{T_0}\log \Big( \mathbb{E}_{\hat{\mathbf{g}} \sim Q_t}[ \exp(-\eta (Y_t - \hat{g}_{i,t})^2)]\Big) 
\end{align*}
where $\mathbb{E}_{\hat{\mathbf{g}} \sim Q_t}$ denotes the expectation conditional on the data taken with respect to a distribution $Q_t$ on base-learners which assigns a probability proportional to  $\exp(-\eta \sum_{s=1}^{t-1} (Y_t - \hat{g}_{i,s})^2)$ to each base algorithm. Recalling Hoeffding bound on the moment generating function of a bounded random variable we observe that 
\begin{align*} 
 \log ( \mathbb{E}_{\hat{\mathbf{g}} \sim Q_t}[ \exp(-\eta (Y_t - \hat{g}_{i,t})^2)]) 
&
 \le -\eta \mathbb{E}_{\hat{\mathbf{g}} \sim Q_t}[(Y_t - \hat{g}_{i,t})^2] + \frac{\eta^2 C_t^2}{8}
 \\
  \qquad  &
  \le
   -\eta (Y_t - \mathbb{E}_{\hat{\mathbf{g}} \sim Q_t}[\hat{g}_{i,t}])^2 + \frac{\eta^2 C_t^2}{8}
   \\
    \qquad  & = -\eta (Y_t - \sum_{i=1}^p \frac{w_{i,t}}{W_t} \hat{g}_{i,t})^2 + \frac{\eta^2 C_t^2}{8}
   \\
    \qquad   &
  = 
    -\eta (Y_t - \hat{m}_t(X_t))^2 + \frac{\eta^2 C_t^2}{8} 
\end{align*}
where  $C_t = Y_t^2 + M^2$. Hence we have 
\begin{align} 
&-\eta \text{min}_{i \in \{1, \dots, p\}} \sum_{t=1}^T (Y_t - \hat{m}_t(X_t, w^{t-1}))^2 -\log (p) \\
&
\le\log (W_{T_0}/w) \le \sum_{t=1}^{T_0} -\eta (Y_t - \sum_{i=1}^p \frac{w_{i,t}}{W_t} \hat{g}_{i,t})^2 + \sum_{t=1}^{T_0} \frac{\eta^2 C_t^2}{8} .
\end{align} 
Rearranging terms we get 
\begin{equation} 
\begin{aligned} 
&\sum_{t=1}^{T_0}(Y_t - \hat{m}_t(X_t, w^{t-1}))^2 - \min_{i \in \{1, \dots, p\}}  \sum_{t=1}^{T_0}(Y_t - \hat{g}_{i,t})^2  \le \frac{\log(p)}{\eta} + \sum_{t=1}^{T_0} \frac{ C_t^2 \eta}{8} . 
\end{aligned} 
\end{equation} 
\end{proof} 

\begin{cor}[Theorem \ref{thm:regret0}] Suppose that $Y_t \le M < \infty$ almost surely. Consider $\eta =  \sqrt{\frac{8\log (p)}{M^2 T_0}}$ and the above conditions hold. Then we obtain that 
\begin{equation} 
\frac{1}{T_0}\sum_{t=1}^{T_0}(Y_t - \hat{m}_t(X_t, w^{t-1}))^2 - \min_{i \in \{1, \dots, p\}} \frac{1}{T_0} \sum_{t=1}^{T_0}(Y_t - \hat{g}_{i,t})^2 \le C \sqrt{\frac{\log(p)}{ T_0}}. 
\end{equation} 
for a constant $C < \infty$. Therefore, Theorem \ref{thm:regret0} holds. 
\end{cor}

\subsection{Proof of Theorem \ref{thm:regret1}}

We now discuss the main proof of the theorem. 

We let $w^{t-1}$ denote the exponential weights computed using information only from time $1$ up to time $t-1$ for estimating $\hat Y_t^0$.

 We start by decomposing the average loss as follows. 
	\begin{equation} \label{eqn:1A}
	\small 
	\begin{aligned} 
	 \frac{1}{T_0}\sum_{t=1}^{T_0} (Y_t - \hat{m}_t(X_t, w^{t-1}))^2 & = \frac{1}{T_0}\sum_{t=1}^{T_0} (\hat{m}_t(X_t,w^{t-1}) - \mu_t(X_t))^2 
	+ \varepsilon_{0t}^2 + 2(\hat{m}_t(X_t,w^{t-1}) - \mu_t(X_t))\varepsilon_{0t}
	\\ \frac{1}{T_0}\sum_{t=1}^{T_0} (Y_t - \hat{g}_i(X_t))^2  & = \frac{1}{T_0}\sum_{t=1}^{T_0}  (\hat{g}_i(X_t) - \mu_t(X_t))^2 + \varepsilon_{0t}^2 + 2(\hat{g}_i(X_t) - \mu_t(X_t))\varepsilon_{0t}
	\end{aligned} .
	\end{equation} 
Let $\hat{g}_{i,t} = \hat{g}_i(X_t)$. Note that 
\begin{equation} \label{eqn:2A}
\small 
\begin{aligned} 
&\frac{1}{T_0} \Big\{ \sum_{t=1}^{T_0} (Y_t - \hat{m}_t(X_t,w^{t-1}))^2 - \text{min}_{i \in \{1,\dots,p\}} \sum_{t=1}^{T_0} (\hat{g}_{i,t} - Y_t)^2 \Big\}   
\\ & = \frac{1}{T_0}\sum_{t=1}^{T_0} (\mu_t(X_t)- \hat{m}_t(X_t, w^{t-1}))^2 + 2 \text{min}_{i \in \{1,\dots,p\}} \frac{1}{T_0}\sum_{t=1}^{T_0} (\hat{m}_t(X_t,w^{t-1}) 
 - \hat{g}_i(X_t))\varepsilon_{0t} -  \frac{1}{T_0}\sum_{t=1}^{T_0}(\hat{g}_{i,t} - \mu_t(X_t))^2.
\end{aligned} 
\end{equation}
Next, we provide a bound on the cumulative one step ahead prediction error.

Using Lemma \ref{lem:cesabianchi} and the decomposition in \eqref{eqn:2A}, for a finite constant $C$ independent of $p, T_0$, we write
\begin{equation} \label{eqn:last2}
\begin{aligned} 
&\frac{1}{T_0} \{ \sum_{t=1}^{T_0} (\mu_t(X_t) - \hat{m}_t(X_t, w^{t-1}))^2 
\\
&\le \text{min}_{i \in \{1, ..., p\}}  \frac{1}{T_0}  \sum_{t=1}^{T_0} \{(\hat{g}_{i,t} - \hat{m}_t(X_t,w^{t-1}))\varepsilon_{0t}  + (\hat{g}_{i,t} - \mu_t(X_t))^2\}  + C\sqrt{\frac{\log(p)}{T_0}} 
\\
&\le \underbrace{\max_{j \in \{1, ..., p\}}  \frac{1}{T_0}  \sum_{t=1}^{T_0} (\hat{g}_{j,t} - \hat{m}_t(X_t,w^{t-1}))\varepsilon_{0t}}_{(I)}  + \underbrace{\text{min}_{i \in \{1, ..., p\}}  \frac{1}{T_0}  \sum_{t=1}^{T_0} (\hat{g}_{i,t} - \mu_t(X_t))^2  + C\sqrt{\frac{\log(p)}{T_0}}}_{(II)} .
\end{aligned}
\end{equation}
We discuss the term $(I)$. Define 
$$
V_t =  \varepsilon_t(\hat{m}_t(X_t,w^{t-1}) - \hat{g}_i(X_t)) - \mathbb{E}[\varepsilon_t(\hat{m}_t(X_t,w^{t-1}) - \hat{g}_i(X_t))|\mathcal{F}_{t-1}].  
$$ 
We have 
$$
(I) = \max_{i \in \{1, ..., p\}}  \frac{1}{T_0}  \sum_{t=1}^{T_0} \Big[V_t + \mathbb{E}[\varepsilon_t(\hat{m}_t(X_t,w^{t-1}) - \hat{g}_i(X_t))|\mathcal{F}_{t-1}] \Big]. 
$$ 
Using the triangular inequality, we have 
\begin{equation} \label{eqn:ha} 
\begin{aligned} 
|(I)| &\le \max_{i \in \{1, \cdots, p\}} \Big|\frac{1}{T_0} \sum_{t = 1}^{T_0} \varepsilon_t(\hat{m}_t(X_t,w^{t-1}) - \hat{g}_i(X_t)) - \mathbb{E}\Big[\varepsilon_t(\hat{m}_t(X_t,w^{t-1}) - \hat{g}_i(X_t))\Big|\mathcal{F}_{t-1}\Big]\Big| \\ &+ \Big|\frac{1}{T_0} \sum_{t = 1}^{T_0} \mathbb{E}\Big[\varepsilon_t(\hat{m}_t(X_t,w^{t-1}) - \hat{g}_i(X_t))\Big|\mathcal{F}_{t-1}\Big]\Big|. 
\end{aligned} 
\end{equation} 

Observe that $V_t$ is a martingale difference sequence. As a result, we can use Azuma inequality \citep{boucheron2013concentration}, and the union bound over $p$, and obtain  with probability at least $1 - \delta$, 
$$
\max_{i \in \{1, \cdots, p\}} \Big|\frac{1}{T_0} \sum_{t = 1}^{T_0} \varepsilon_t(\hat{m}_t(X_t,w^{t-1}) - \hat{g}_i(X_t)) - \mathbb{E}\Big[\varepsilon_t(\hat{m}_t(X_t,w^{t-1}) - \hat{g}_i(X_t))\Big|\mathcal{F}_{t-1}\Big]\Big|  \le C_0\sqrt{\frac{ \log(p/\delta)}{T_0}} , 
$$
where $C_0$ is a finite constant independent of $T_0, p$. 
Since $\mathbb{E}\Big[\varepsilon_t(\hat{m}_t(X_t,w^{t-1}) - \hat{g}_i(X_t))\Big|\mathcal{F}_{t-1}\Big] = 0$, we have with probability at least $1 - \delta$, 
$$
|(I)| \le C_0\sqrt{\frac{ \log(p/\delta)}{T_0}}, 
$$ 
for a constant $C_0 < \infty$.

\section{Examples}

\subsection{Failure of Stationarity with Spillovers} \label{sec:aa1}

\begin{figure}[h]
\spacingset{1}
\centering
		\includegraphics[scale=0.5]{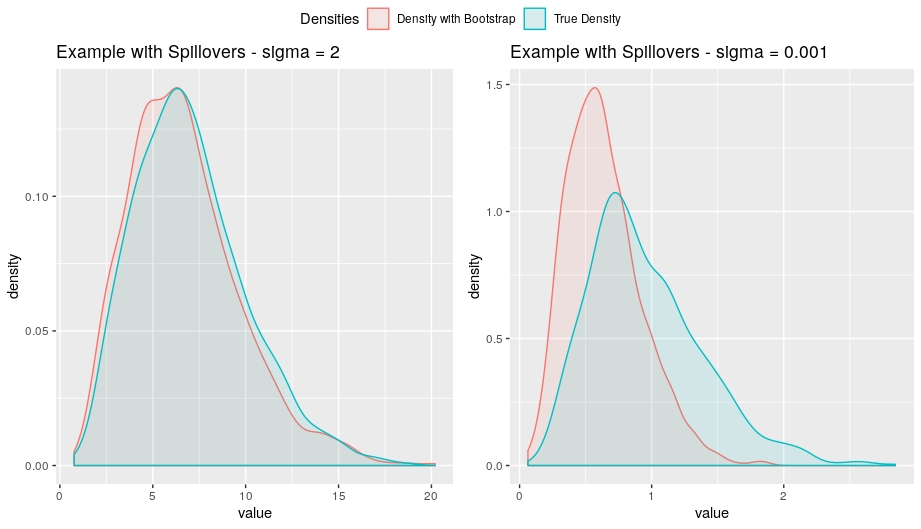}
	\caption{Explanatory example on failure of the bootstrap in the presence of spillover effects on $X_t$. Blue histogram indicates 
$\mathcal{T}_S 
= \tau_1 \chi^2(T - T_0)$ whereas the red  corresponds  to its bootstrapped counterpart $\mathcal{T}_S^* 
= \tau_2 \chi^2(T - T_0)$ where $\tau_1 = {(\sigma^2 + 2J^{-1})}/{\sqrt{T - T_0}}$, $\tau_2 =  {(\sigma^2 +  J^{-1})}/{\sqrt{T - T_0}}$. 
 Large values of $\sigma$ correspond to diminishing effects of the spillovers in which case we see good approximation properties; although in this case stationarity is still violated. For small values of $\sigma$ we see clear departures of the two distributions.} \label{fig:aa} 
	
\end{figure}

\noindent Although our results remain correct in the presence of carry-over effects   we do not expect them to hold for the case of spillover over covariates. Observe that spillovers effects on $X_t$ violate stationarity of $X_t$ and hence the validity of Theorem \ref{thm:1}.

 We provide   a counterexample where the bootstrap is likely to fail under such circumstances. We consider a case in which   covariates are a function of the treatment assignment indicator - i.e., there are spillovers from $Y_t$ to $X_t$ while keeping  $m = 0$. A factor model,   $
Y_t = a_t D_t + F_t + \mathcal{N}(0,\sigma^2)$ and $10$-covariates that possibly depend on the treatment groups  $d$, $X_{j,t}(d) = F_t + \mathcal{N}(0, 1 + d)$ are considered with $F_t \sim \mathcal{N}(0,1)
$.  Moreover, $a_t = 0$. The best estimator then   is a simple average 
  $\hat{Y}_t^0 =  \bar X_t$ where $\bar X_t$ is the sample average of $X_{t,j}$'s.
Then,
$
\mathcal{T}_S
= \tau_1 \chi^2_{T - T_0}
$ for $\tau_1 = {(\sigma^2 + 2J^{-1})}/{\sqrt{T - T_0}}$,
while  $
\mathcal{T}_S^* 
= \tau_2 \chi^2_{T - T_0}
$ with $\tau_2 =  {(\sigma^2 +  J^{-1})}/{\sqrt{T - T_0}}$.
Figure \ref{fig:aa} illustrates  over-rejection of the nulls whenever the spillover effects are strong (right panel). For larger values of $\sigma^2$, spillover effects become negligible and the density of the bootstrapped test statistic approximately agrees with  the true density   in this also, non-stationary case (left panel).

\section{Additional Numerical Experiments}

\subsection{Results for Other DGPs}

In this section we illustrate results for other DGPs considered. Results are robust across DGPs and illustrate substantial improvements of the proposed methods. These can be observed in Figure \ref{fig:1app}.

\begin{figure}[!ht]
\centering
\includegraphics[scale=0.5]{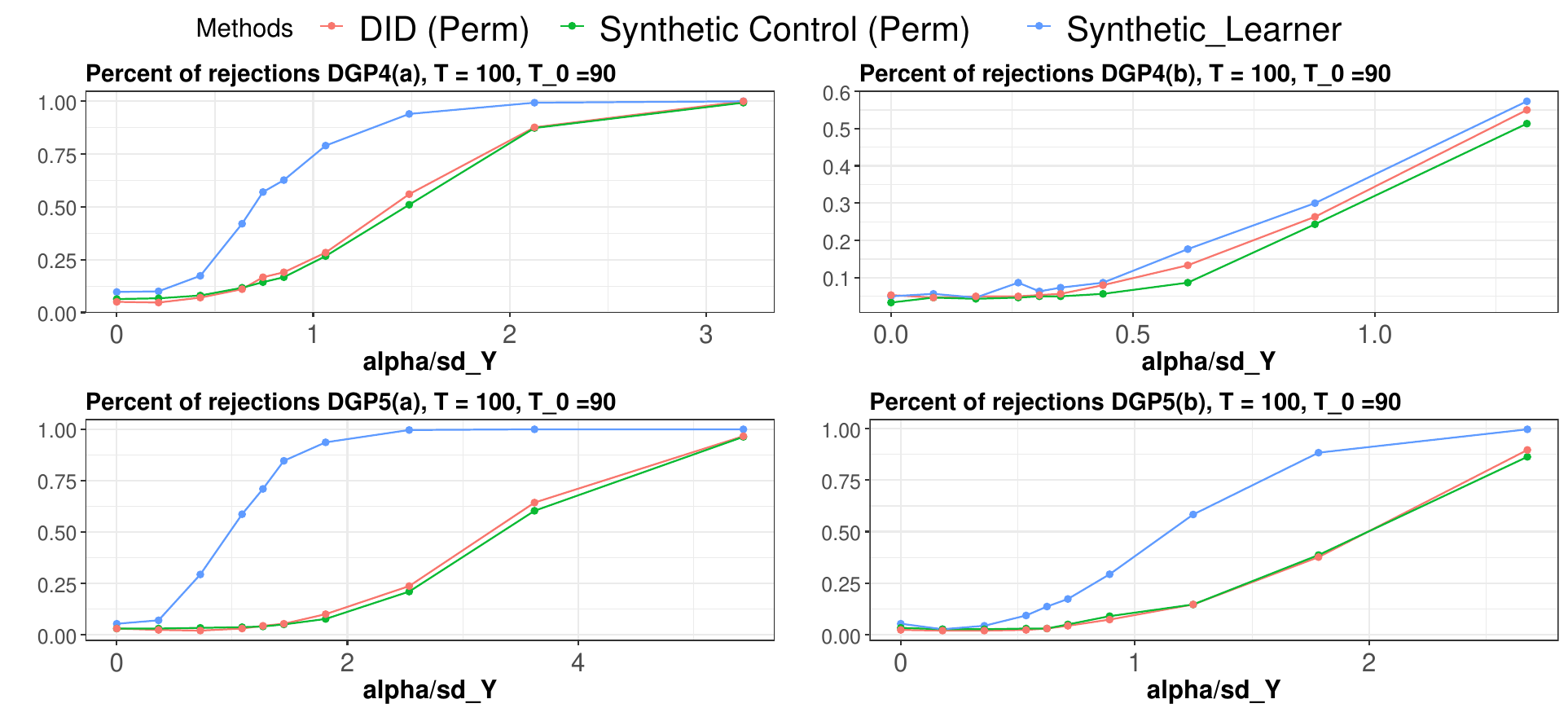}
\caption{ $T=100, p = 10$ Percentage of rejections of the null hypothesis of no treatment effects over $300$ repetitions.  The x-axis reports the policy's effect rescaled by the outcome's standard deviation. Synthetic learner has XGboost,Support Vector Regression and {\sc arima}(0,1,1) and $50$ additional non informative predictions.  The blue line denotes the proposed method; red line denotes the difference-in-differences and green line denotes the Synthetic Control method. }
\label{fig:1app}
\end{figure}

\subsection{Varying $T - T_0$} \label{app:vary_T}

In this section we vary $T - T_0$, i.e., the post treatment period, and show that results remain robust for different choices of the post-treatment period, both longer and shorter than ten. This is illustrated in Figure \ref{fig:2app}, \ref{fig:3app}. In the figures, we report the power for different values of $T$ and $T^* = T - T_0$.

\begin{figure}[!ht]
\centering
\includegraphics[scale=0.5]{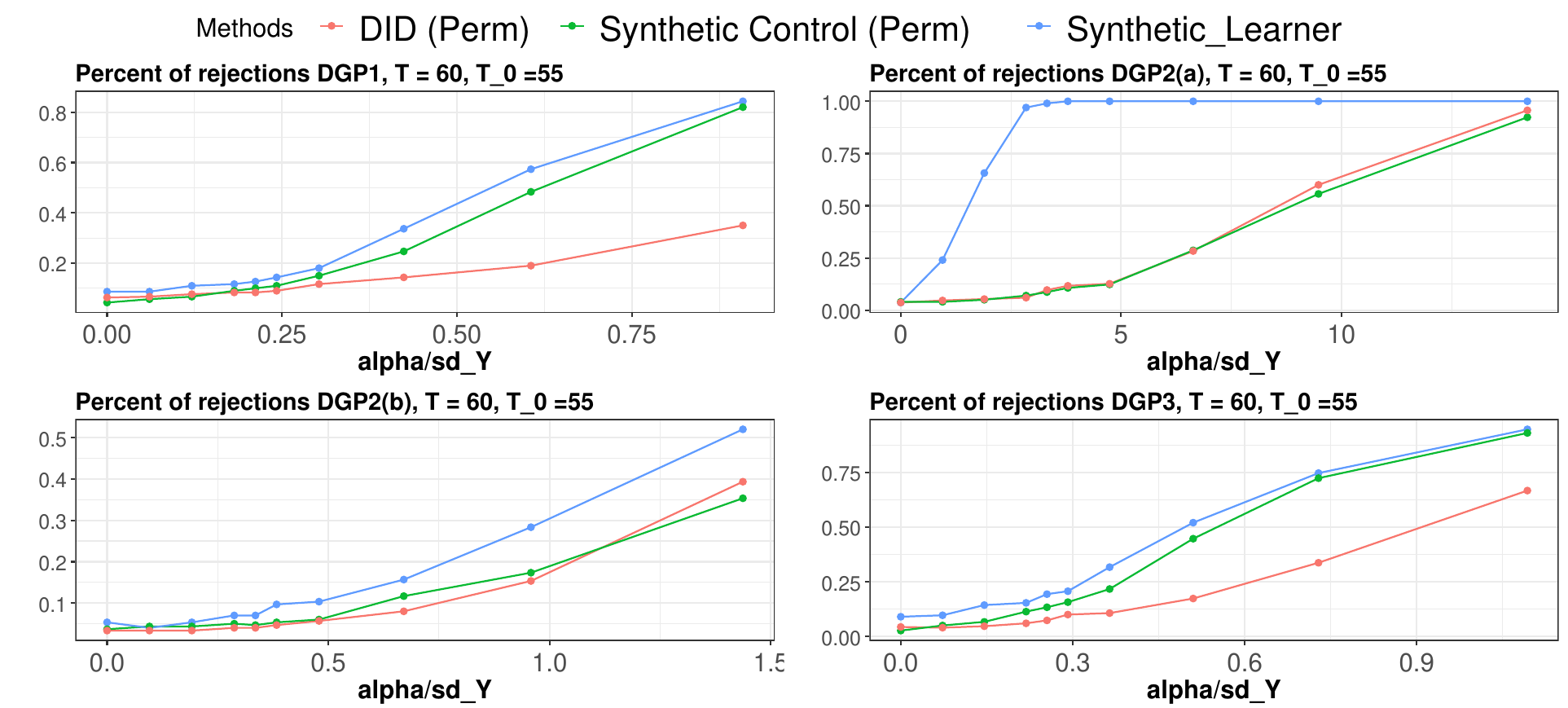}
\includegraphics[scale=0.5]{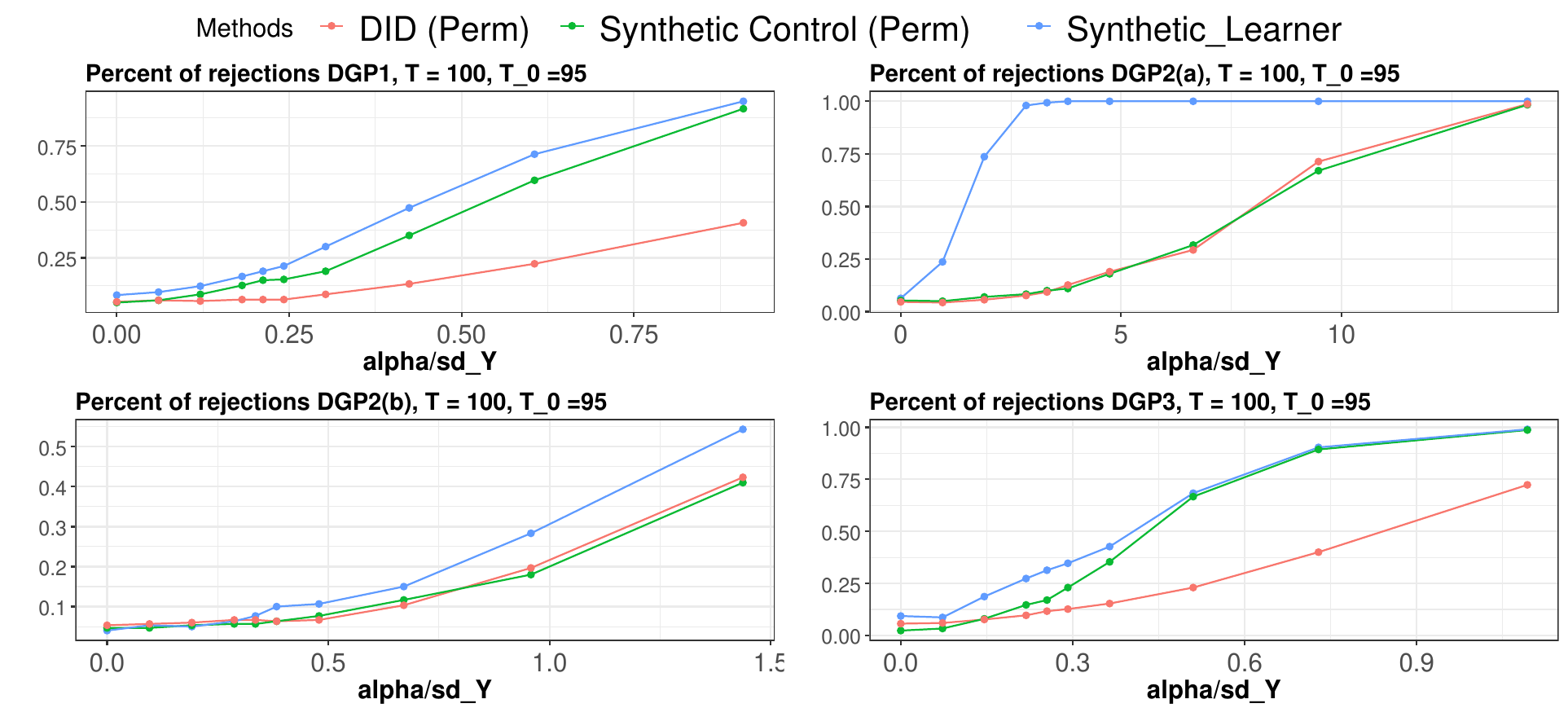}
\caption{$T\in \{60, 100\}, T^* = 5, p = 10$ Percentage of rejections of the null hypothesis of no treatment effects over $300$ repetitions. The x-axis reports the policy's effect rescaled by the outcome's standard deviation. Synthetic learner has XGboost,Support Vector Regression and {\sc arima}(0,1,1) and $50$ additional non informative predictions.  The blue line denotes the proposed method; red line denotes the difference-in-differences and green line denotes the Synthetic Control method. }
\label{fig:2app}
\end{figure}

\begin{figure}[!ht]
\centering
\includegraphics[scale=0.5]{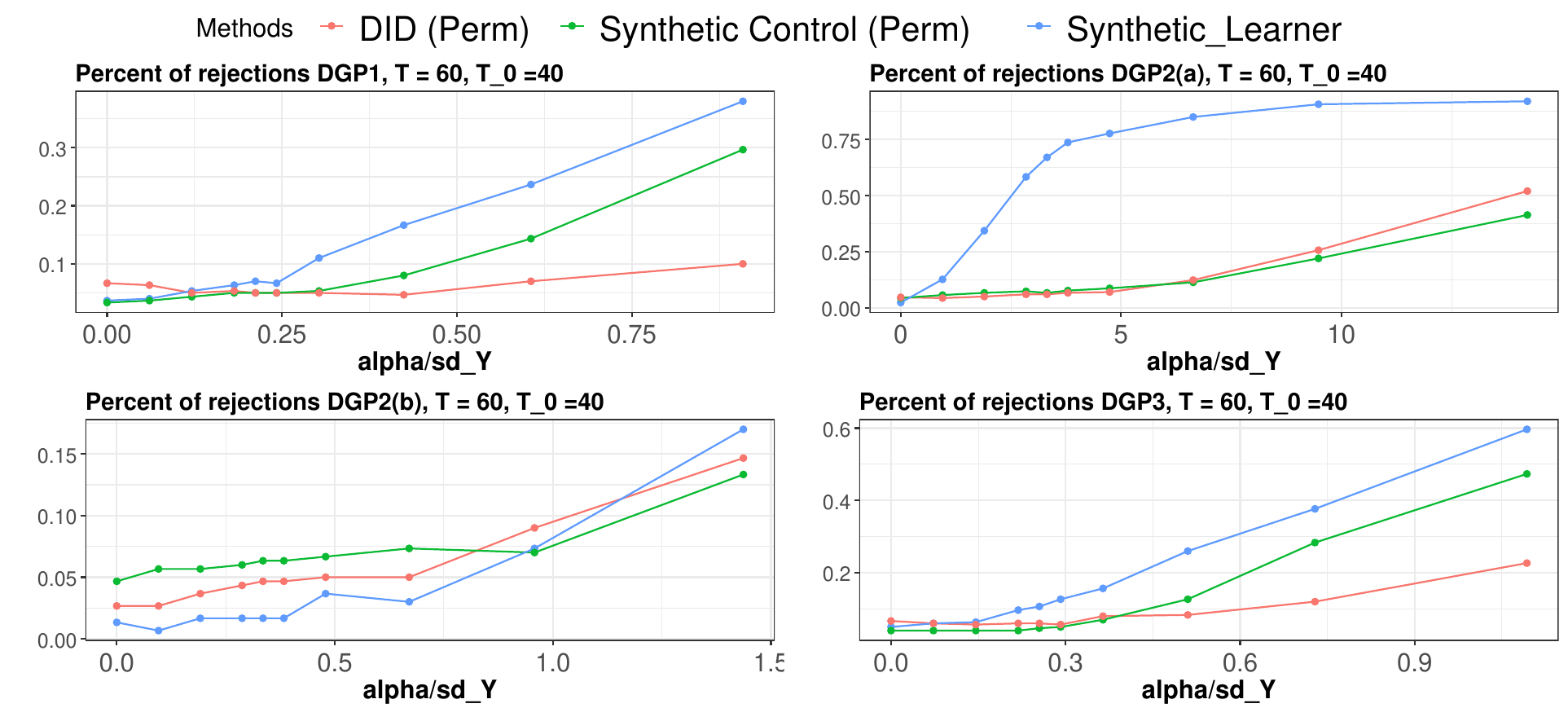}
\includegraphics[scale=0.5]{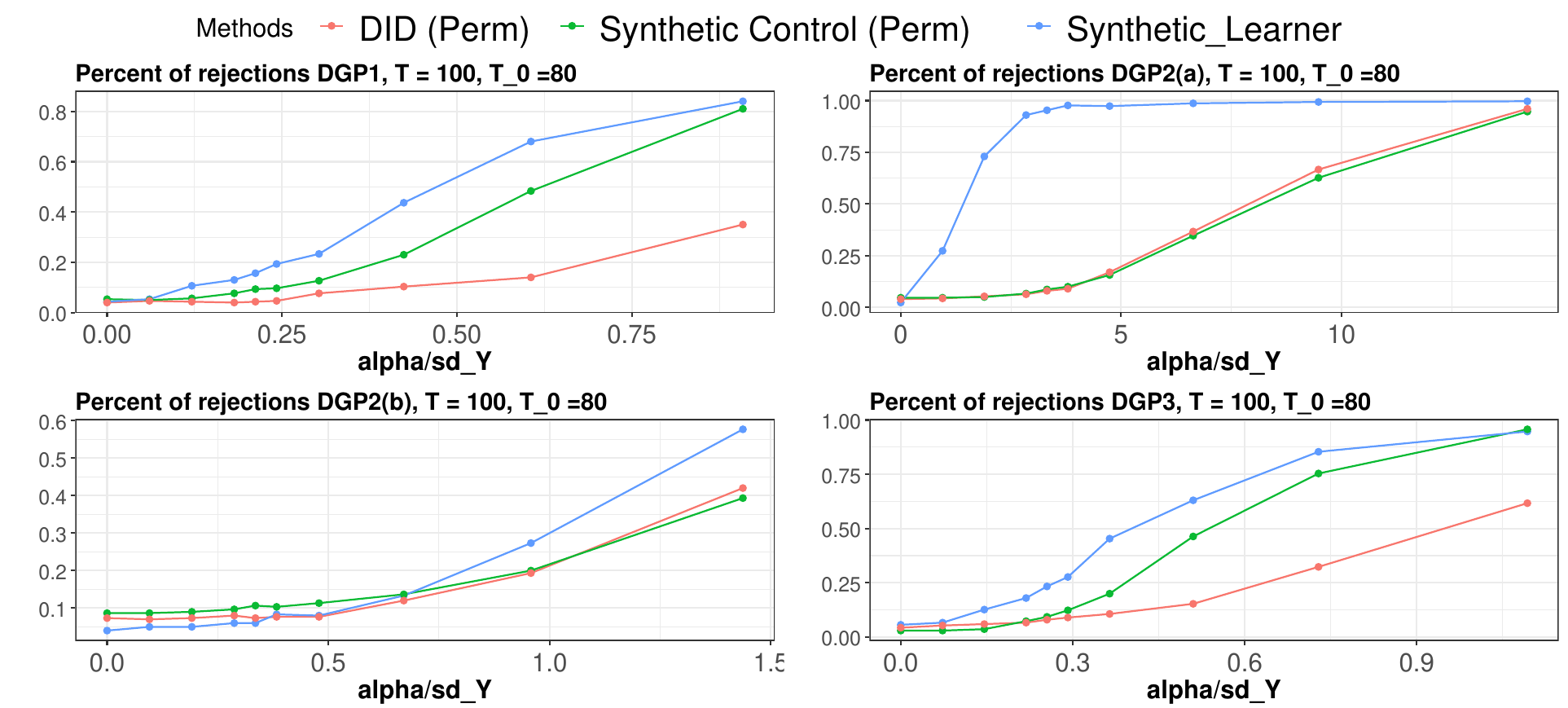}
\caption{$T\in \{60, 100\} , T^* = 20, p =10$ Percentage of rejections of the null hypothesis of no treatment effects over $300$ repetitions.  The x-axis reports the policy's effect rescaled by the outcome's standard deviation. Synthetic learner has XGboost,Support Vector Regression and {\sc arima}(0,1,1) and $50$ additional non informative predictions.  The blue line denotes the proposed method; red line denotes the difference-in-differences and green line denotes the Synthetic Control method. }
\label{fig:3app}
\end{figure}

\subsection{Oracle Study: Known Critical Value} \label{app:oracle_quantile}

In order to better understand the drivers of the power performance, we study the case where the critical value of the test is known to the researcher; we estimated it by Monte Carlo simulation since no closed-form expression for its density is available.   Figure  \ref{fig:plotsim2}  collects our results.  Synthetic Learner does not have uninformative learners, and the class consists of  XGboost, Support Vector Regression, and {\sc arima}(0,1,1). We take   $T=300$, $T_0=280$, and we let $T_-=140$. 
 
   Since the critical quantile are assumed to be known, SC and DiD are estimated using information until time $T_0$ \citep{abadie2010synthetic}, and not on the full sample as imposed by permutation methods. The proposed method outperforms uniformly DiD, and SC in almost all {\bf DGP}s considered.


\begin{figure}[!ht]
\spacingset{1}
\centering

\includegraphics[scale=0.5]{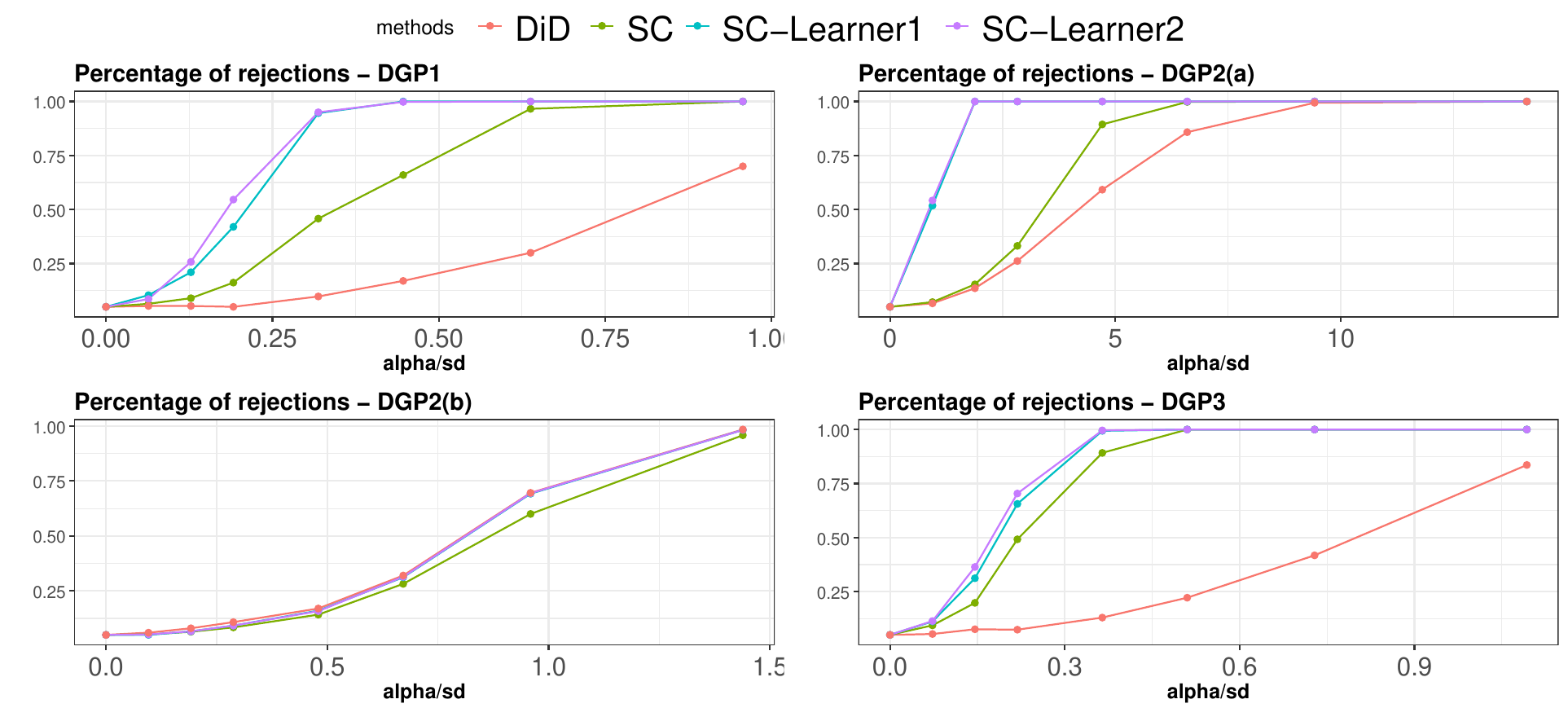}

 \caption{Percentage of rejections   over $500$ repetitions with $T=300$, $T_0=280, p = 50$, and $T_-=140$ when the critical quantile is $known$(oracle case).  The x-axis reports the policy's effect rescaled by the outcome's standard deviation. SC and DiD are estimated using all information until time $T_0$(no permutation test required). SC-Learner1 is the Synthetic Learner trained with XGboost, Support Vector Regression, {\sc arima}(0,1,1)  and $50$ additional non informative predictions. SC-Learner2 is the Synthetic Learner  which also includes classical SC  and it does not include non-informative predictions. The red line corresponds to the DiD method; the yellow line corresponds to the SC method; the blue line is SC-Learner 1 method and the purple line in SC-Learner 2.  
  } \label{fig:plotsim2}
 \end{figure}

Additional results on the oracle study are included in Figure \ref{fig:plotsim2b}.

\begin{figure}[!ht]
\spacingset{1}
\centering

\includegraphics[scale=0.5]{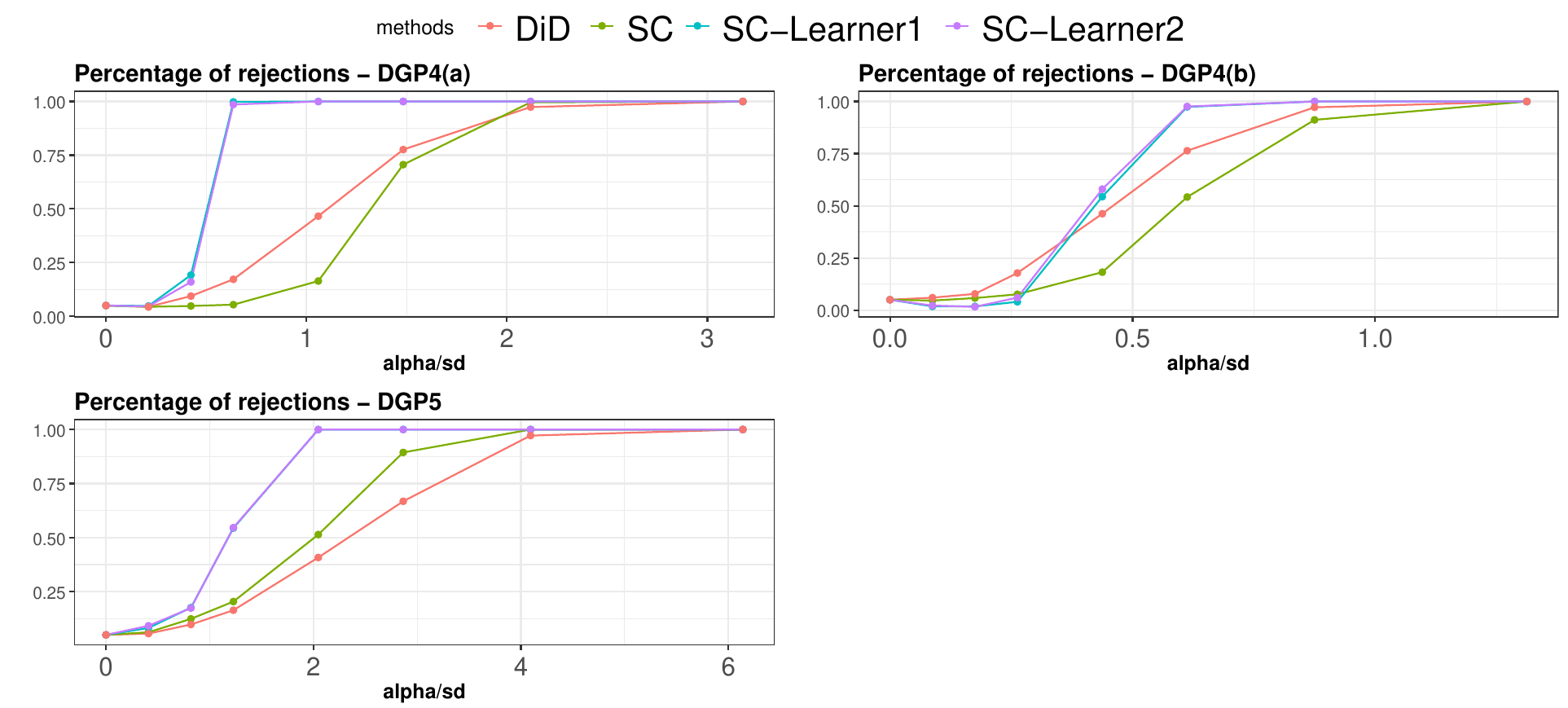}

 \caption{Percentage of rejections   over $500$ repetitions with $T=300$, $T_0=280$, $p=50$ and $T_-=140$ when the critical quantile is $known$(oracle case).  The x-axis reports the policy's effect rescaled by the outcome's standard deviation. SC and DiD are estimated using all information until time $T_0$(no permutation test required). SC-Learner1 is the Synthetic Learner trained with XGboost, Support Vector Regression, {\sc arima}(0,1,1)  and $50$ additional non informative predictions. SC-Learner2 is the Synthetic Learner  which also includes classical SC  and it does not include non-informative predictions. The red line corresponds to the DiD method; the yellow line corresponds to the SC method; the blue line is SC-Learner 1 method and the purple line in SC-Learner 2.  
  } \label{fig:plotsim2b}
 \end{figure}

\subsection{Endogenous Time of the Treatment} 

Additional results for the case of endogenous time of treatment are discussed in Figure \ref{fig:plotsim2db}, which are consistent with results in the main text.




 \begin{figure} [!ht]
\spacingset{1}
\centering

\includegraphics[scale = 0.5]{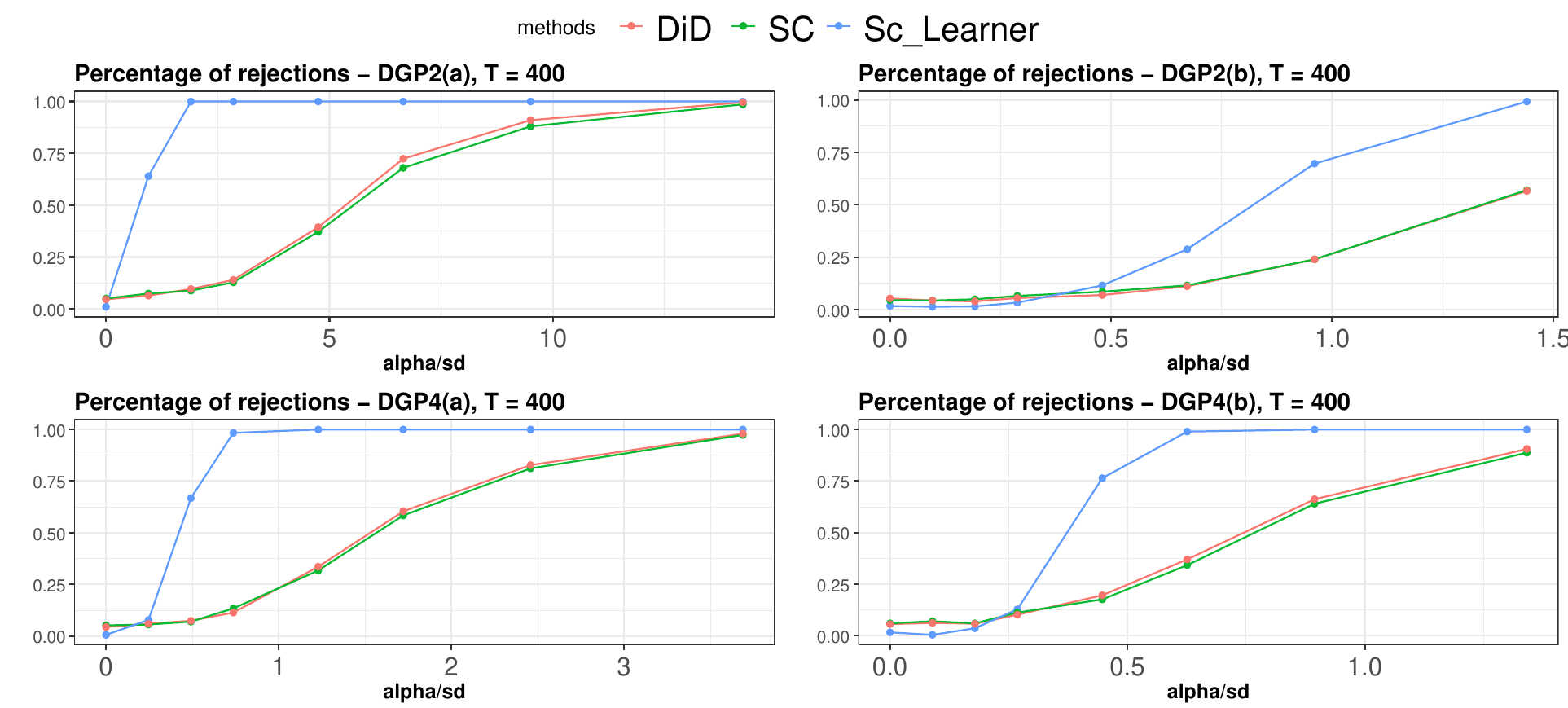}

 \caption{Percentage of rejections   over $500$ repetitions with $T=400$, $T_0=280$, $J=50$ and $T_-=140$ when the critical quantile is estimated via resampling and the time of treatment is \textit{endogenous}..  
  } \label{fig:plotsim2db}
 \end{figure}




\section{Empirical Analysis: Additional Results}  \label{sec:a_emp}

In Table \ref{tab:1a} we report results when all states with the exception of Oregon are considered, since Oregon also implemented a Medicaid reform. Table \ref{tab:2a} \ref{tab:3a} collect results for different choices of the tuning parameter $\eta$ and show robustness of the results.  

\begin{table}[!htbp] \centering 
  \caption{$90\%$ and $80\%$ critical values, t-statistic and ATT using all the states as controls. The effect estimated is over the time window 2010-2014 (first row), and consecutive windows 2011-2014 ($m = 1yr$), 2012-2014 ($m = 2yr$), 2013-2014 ($m = 3yr$).  Period 1 collects results when learners are estimated using the window between $1998$-$2006$ and weights are estimated over the period $1993$-$1997$. Period 2 corresponds to the opposite scenario. Demeaned SL denotes the SL with time-varying fixed effects.} 
  \label{tab:1a} 
\begin{tabular}{@{\extracolsep{5pt}} ccccccccc} 
\\[-1.8ex]\hline 
\hline \\[-1.8ex] 
Period 1 & SL &  & & & Demeaned SL &  &  &  \\ 
 & CV90 & CV80 & t stat & ATT & CV90 & CV80 & t stat & ATT \\ 
\hline \\[-1.8ex] 
m = 0 & $1.723$ & $1.598$ & $1.883$ & $3.525$ & $1.671$ & $1.548$ & $1.820$ & $3.503$ \\ 
m = 1yr & $1.664$ & $1.525$ & $1.830$ & $3.754$ & $1.583$ & $1.455$ & $1.754$ & $3.701$ \\ 
m = 2yr & $1.614$ & $1.477$ & $1.829$ & $3.973$ & $1.529$ & $1.401$ & $1.741$ & $3.892$ \\ 
m = 3yr & $1.535$ & $1.389$ & $1.776$ & $4.098$ & $1.450$ & $1.313$ & $1.679$ & $3.991$ \\ 
\hline \\[-1.8ex] 
Period 2 & SL &  & & & Demeaned SL &  &  &  \\ 
 & CV90 & CV80 & t stat & ATT & CV90 & CV80 & t stat & ATT \\ 
\hline \\[-1.8ex] 
m = 0 & $1.707$ & $1.574$ & $1.279$ & $5.840$ & $1.086$ & $1.007$ & $1.049$ & $5.520$ \\ 
m = 1yr & $1.642$ & $1.514$ & $1.221$ & $6.040$ & $1.027$ & $0.948$ & $0.990$ & $5.708$ \\ 
m = 2yr & $1.609$ & $1.481$ & $1.226$ & $6.243$ & $1.001$ & $0.921$ & $0.995$ & $5.901$ \\ 
m = 3yr & $1.555$ & $1.420$ & $1.194$ & $6.351$ & $0.964$ & $0.883$ & $0.972$ & $6.001$ \\ 
\hline \\[-1.8ex] 
\end{tabular} 
\end{table} 

\begin{table}[!htbp] \centering 
  \caption{$90\%$ and $80\%$ critical values, t-statistic and ATT using the southern states as controls and $\eta = 1/(T \mathrm{Var}(Y_t))$. The effect estimated is over the time window 2010-2014 (first row), and consecutive windows 2011-2014 ($m = 1yr$), 2012-2014 ($m = 2yr$), 2013-2014 ($m = 3yr$).  Period 1 collects results when learners are estimated using the window between $1998$-$2006$ and weights are estimated over the period $1993$-$1997$. Period 2 corresponds to the opposite scenario. Demeaned SL denotes the SL with time-varying fixed effects.} 
  \label{tab:2a} 
\begin{tabular}{@{\extracolsep{5pt}} ccccccccc} 
\\[-1.8ex]\hline 
\hline \\[-1.8ex] 
Period 1 & SL &  & & & Demeaned SL &  &  &  \\ 
 & CV90 & CV80 & t stat & ATT & CV90 & CV80 & t stat & ATT \\ 
\hline \\[-1.8ex] 
m = 0 & $1.321$ & $1.237$ & $1.328$ & $1.786$ & $0.901$ & $0.844$ & $0.834$ & $0.969$ \\ 
m = 1yr & $1.237$ & $1.155$ & $1.253$ & $1.932$ & $0.787$ & $0.743$ & $0.728$ & $1.034$ \\ 
m = 2yr & $1.200$ & $1.114$ & $1.244$ & $2.096$ & $0.751$ & $0.706$ & $0.695$ & $1.077$ \\ 
m = 3yr & $1.151$ & $1.057$ & $1.200$ & $2.178$ & $0.709$ & $0.659$ & $0.648$ & $1.063$ \\ 
\hline \\[-1.8ex] 
Period 2 & SL &  & & & Demeaned SL &  &  &  \\ 
 & CV90 & CV80 & t stat & ATT & CV90 & CV80 & t stat & ATT \\ 
\hline \\[-1.8ex] 
m = 0 & $0.713$ & $0.665$ & $0.321$ & $3.280$ & $0.521$ & $0.473$ & $0.246$ & $3.137$ \\ 
m = 1yr & $0.642$ & $0.597$ & $0.231$ & $3.297$ & $0.462$ & $0.416$ & $0.166$ & $3.142$ \\ 
m = 2yr & $0.631$ & $0.586$ & $0.220$ & $3.365$ & $0.453$ & $0.408$ & $0.157$ & $3.204$ \\ 
m = 3yr & $0.618$ & $0.571$ & $0.214$ & $3.392$ & $0.452$ & $0.405$ & $0.152$ & $3.198$ \\ 
\hline \\[-1.8ex] 
\end{tabular} 
\end{table} 

\begin{table}[!htbp] \centering 
  \caption{$90\%$ and $80\%$ critical values, t-statistic and ATT using the southern states as controls and $\eta = 1\%$. The effect estimated is over the time window 2010-2014 (first row), and consecutive windows 2011-2014 ($m = 1yr$), 2012-2014 ($m = 2yr$), 2013-2014 ($m = 3yr$).  Period 1 collects results when learners are estimated using the window between $1998$-$2006$ and weights are estimated over the period $1993$-$1997$. Period 2 corresponds to the opposite scenario. Demeaned SL denotes the SL with time-varying fixed effects.} 
  \label{tab:3a} 
\begin{tabular}{@{\extracolsep{5pt}} ccccccccc} 
\\[-1.8ex]\hline 
\hline \\[-1.8ex] 
Period 1 & SL &  & & & Demeaned SL &  &  &  \\ 
 & CV90 & CV80 & t stat & ATT & CV90 & CV80 & t stat & ATT \\ 
\hline \\[-1.8ex] 
m = 0 & $1.336$ & $1.254$ & $1.362$ & $1.865$ & $0.895$ & $0.841$ & $0.827$ & $0.988$ \\ 
m = 1yr & $1.256$ & $1.170$ & $1.292$ & $2.023$ & $0.792$ & $0.744$ & $0.720$ & $1.054$ \\ 
m = 2yr & $1.221$ & $1.131$ & $1.286$ & $2.197$ & $0.755$ & $0.709$ & $0.686$ & $1.094$ \\ 
m = 3yr & $1.155$ & $1.062$ & $1.245$ & $2.289$ & $0.713$ & $0.659$ & $0.641$ & $1.082$ \\ 
\hline \\[-1.8ex] 
Period 2 & SL &  & & & Demeaned SL &  &  &  \\ 
 & CV90 & CV80 & t stat & ATT & CV90 & CV80 & t stat & ATT \\ 
\hline \\[-1.8ex] 
m = 0 & $1.281$ & $1.186$ & $0.513$ & $4.690$ & $0.547$ & $0.496$ & $0.236$ & $3.143$ \\ 
m = 1yr & $1.219$ & $1.127$ & $0.441$ & $4.805$ & $0.486$ & $0.437$ & $0.157$ & $3.161$ \\ 
m = 2yr & $1.197$ & $1.108$ & $0.439$ & $4.928$ & $0.485$ & $0.431$ & $0.149$ & $3.225$ \\ 
m = 3yr & $1.164$ & $1.074$ & $0.427$ & $4.970$ & $0.480$ & $0.427$ & $0.144$ & $3.211$ \\ 
\hline \\[-1.8ex] 
\end{tabular} 
\end{table}

In Figure \ref{fig:counterfactual1}, we plot the observed outcome and the estimated counterfactual for the percentage of individuals not having economic access to health-care.

 \begin{figure}[h]
\centering
\includegraphics[height=8cm,width=15cm]{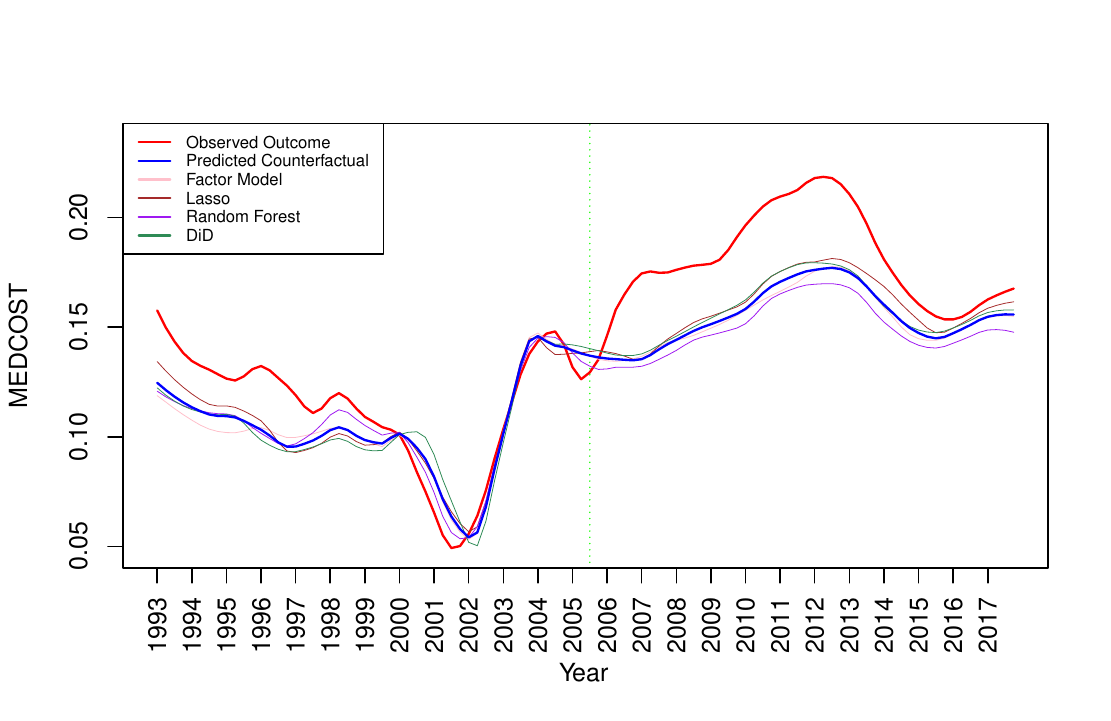}
\caption{Observed and predicted counterfactual of the percentage of childless adults who are not able to afford health-care expenses with demeaned SL and southern states as controls. The red line is the observed time series in Tennessee, and the blue line is the estimated counterfactual under no dis-enrollment. The other lines are predictions of the counterfactual for each base-learner used. In the plot, the time series are smoothed using a local polynomial.
} \label{fig:counterfactual1}
\end{figure}

 \section{Additional Algorithms}

\begin{algorithm}[!ht]    \caption{Synthetic Control Bootstrap}\label{alg:alg4}
    \begin{algorithmic}[1]
     \Require Observations $\{Y_{t}, X_t\}_{t=T_{-}}^{T}$, time of the treatment-$T_0$,  carryover effect size-$m$, tuning parameter $\eta>0$,   learners  $f_1,\dots, f_p$, null hypothesis values $\{a_{t}^o\}_{t > T_0}$
    \State Split the pre-treatment period into two parts: $t \in [T_{-},0]$ and $t \in [1,T_0]$;
    \State Form  predictions   $\hat{g}_j =f_j(\{Y_{t}, X_t\}_{t < 1})$ , $j \in 1,\dots,p$; 
    \State Compute $Y_{0t}^o$ for all $t$
    \For{$b=1,\dots, B$}
        \State Return a sample $\{\tilde{Y}_t^*, X_t^*\}$ of size $T$  by performing circular block bootstrap  on  $\{Y_{0t}^o,X_t \}$ for $t \in \{1,\dots, T\}$ 
 \State  Compute ${\widehat{\mathbf{w}}_0 }^*$ on $\{\tilde{Y}_t^*, \hat{\mathbf{g}}(X_t^*) \}_{1 \leq t \le T_0}$    according to \eqref{eqn:expweights1}
  \State Compute the predicted counterfactual 
   $
{\  \hat Y_{0t}^{0} }^{*}= \sum_{i = 1}^{p} {\  {\widehat{\mathbf{w}}}_0^{*} }^{(i)} \hat{g}_{i}^0(X_t^*) 
   $, $t > T_0$. 
        \State  Compute the   test statistics of interest for sharp and average null 
 $$\mathcal{T}_S^{b} = (T - T_0)^{-1/2}\sum_{t > T_0 } \bigl(\tilde{Y}_{t}^* - {\  \hat Y_{0t}^{0} }^{*}\bigl)^2 \qquad \mbox{or}  \qquad \mathcal{T}_A^{b} = (T - T_0)^{-1} \Bigl( \sum_{t > T_0 } \tilde{Y}_{t}^* - {\  \hat Y_{0t}^{0} }^{*} \Bigl)^2;$$

\EndFor
 
      \Return $ q_{1-\alpha}^*$ as $(1-\alpha)$-th quantile  of the sample 
      
   \hskip 100pt   $\mathcal{T}_S^{1} ,\mathcal{T}_S^{2}, \dots,\mathcal{T}_S^{B} \qquad  $ or $\qquad \mathcal{T}_A^{1} , \mathcal{T}_A^{2},\dots,\mathcal{T}_A^{B}  $.
         \end{algorithmic}
\end{algorithm}

 \end{document}